\documentclass[11pt]{article}
\usepackage{amsmath,amsthm,amssymb,amscd}
\usepackage{xurl}
\usepackage[hidelinks]{hyperref}
\usepackage{geometry}
\usepackage{tikz-cd}
\usepackage{booktabs} 
\usepackage{graphicx} 
\usepackage{multirow}
\usepackage{tikz}
\usetikzlibrary{calc,positioning,3d}
\usepackage[normalem]{ulem} 

\usepackage{float}

\usepackage{subcaption}
\usepackage{caption}

\usepackage{makecell} 

 \usepackage{threeparttable}

\newtheoremstyle{nature}
{6pt}          
{6pt}          
{\normalfont}  
{}             
{\itshape}     
{.}            
{0.5em}        
{}

\theoremstyle{nature}
\newtheorem{theorem}{Theorem}

\newtheorem{proposition}{Proposition}

\theoremstyle{definition}

\title{CAKR: Commutative algebra k-mer representations for genomics}
\author{
	Faisal Suwayyid$^{1,2}$,
	Yuta Hozumi$^{2,\dagger}$,
	Mushal Zia$^{2}$, \\
	JunJie Wee$^{2}$,
	Hongsong Feng$^{3}$,
	and Guo-Wei Wei$^{2,4,5,*}$ \\[4pt]
	$^{1}$Mathematics Department,\\
	King Fahd University of Petroleum and Minerals, Dhahran 31261, KSA.\\
	$^{2}$Department of Mathematics,\\
	Michigan State University, East Lansing, MI 48824, USA.\\
	$^{3}$Department of Mathematics and Statistics,\\
	University of North Carolina at Charlotte, Charlotte, NC 28223, USA.\\
	$^{4}$Department of Mathematics,\\
	University of Georgia, Athens, GA 30602, USA.\\
	$^{5}$Department of Biochemistry and Molecular Biology,\\
	University of Georgia, Athens, GA 30602, USA.\\[4pt]
	$^{\dagger}$Current address: School of Mathematics,
	Georgia Institute of Technology, Atlanta, GA, USA.\\
	$^{*}$Correspondence: Guo-Wei Wei
	(\href{mailto:guowei.wei@uga.edu}{guowei.wei@uga.edu}).
}

\date{\small Accepted for publication in \textit{Nature Communications}.}

\begin{document}

 \maketitle

\begin{abstract}
Despite the availability of various sequence analysis models, comparative genomic analysis remains a challenge in genomics, genetics, and phylogenetics. Commutative algebra, a fundamental tool in algebraic geometry and number theory, has rarely been used in data and biological sciences. In this study, we introduce commutative algebra $k$-mer representations as a nonlinear algebraic framework for analyzing genomic sequences. This representation bridges commutative algebra, algebraic topology, combinatorics, and machine learning to establish a mathematical framework for comparative genomic analysis. We evaluate its effectiveness on three tasks including genetic variant classification, phylogenetic tree reconstruction, and viral classification, typically requiring alignment-based, alignment-free, and machine-learning approaches, respectively. In this work, we show that commutative algebra $k$-mer representations outperform five state-of-the-art sequence analysis methods across twelve primary datasets, with two additional supplementary fragment-placement benchmarks, especially in viral classification, and maintain relatively stable predictive accuracy as dataset size increases, underscoring scalability and robustness.

\end{abstract}

\clearpage
\section{Introduction}
\label{sec:introduction}

Comparative genomics examines genetic variation across species and populations to study evolution, identify functional elements, assess diversity, and reconstruct phylogeny~\cite{rubin2000comparative, frazer2004vista}. Comparative analysis of genomic sequences includes phylogenetic inference, functional annotation, and phenotype classification~\cite{nei1996phylogenetic, bellgard1999dynamic}. These analyses require a genome space, a metric space whose points represent genomes and whose distances capture biologically meaningful similarities. An effective metric should reflect structural, functional, and evolutionary relationships, enabling robust comparison and downstream analyses.

Traditional approaches rely on alignment-based methods that identify substitutions, insertions, and deletions through global or local optimization. Tools such as Clustal Omega~\cite{sievers2011fast}, MAFFT~\cite{katoh2013mafft}, and MUSCLE~\cite{edgar2004muscle} are effective for closely related sequences~\cite{patino2021recombination}. However, their computational cost scales poorly with sequence length and dataset size, and their accuracy deteriorates for highly divergent sequences, limiting their applicability in phylogenetics~\cite{edgar2004muscle}.

Alignment-free methods address these limitations by mapping sequences to fixed-length vectors~\cite{vinga2014alignment}, enabling scalability~\cite{zielezinski2017alignment, bonham2014alignment} and whole-genome analysis~\cite{bernard2016alignment, zielezinski2019benchmarking}. Frequency-based approaches use \(k\)-mer count vectors~\cite{blaisdell1986measure}, while others employ entropy~\cite{tribus1971energy}, Lempel–Ziv~\cite{otu2003new}, or Kolmogorov complexity~\cite{li2008kolmogorov} measures. Model-based alignment-free approaches, such as CAFE~\cite{lu2017cafe}, further derive explicit evolutionary distances from $k$-mer frequency statistics under substitution assumptions. Although efficient, many neglect positional or structural information, limiting their performance in genetic variant analysis~\cite{zielezinski2017alignment}.

Recent advances aim to enrich such feature representations. The natural vector method (NVM) encodes statistical moments of (k)-mer positions~\cite{yu2013real, deng2011novel}. Chaos game representation (CGR) maps sequences into fractal images~\cite{jeffrey1990chaos, randic2013milestones}, while Fourier power spectrum (FPS) analysis extracts dominant periodicities~\cite{hoang2015fourier, yin2014fourier}. Multi-scale integration has also been achieved via fuzzy integrals~\cite{saw2019alignment, yu2010novel}. Persistent challenges include sensitivity to parameter choices, such as (k), weights, and dimensions~\cite{yu2024optimal, sims2009alignment}, and a limited capacity to detect biologically significant variants. These challenges motivate novel computational approaches to genomics.

Commutative algebra, the study of commutative rings, ideals, and modules, underpins key areas of modern mathematics, including algebraic geometry, number theory, and homological algebra. Despite its foundational role in pure mathematics, it has hardly been applied in science and technology due to its abstractness and lack of metric. Recently, Suwayyid and Wei introduced multi-scale analysis to commutative algebra, enabling the potential application of abstract nonlinear algebra to data science and learning \cite{suwayyid2025persistent}.

In this work, we introduce commutative algebra to genomics. By leveraging persistent Stanley-Reisner theory (PSRT)~\cite{suwayyid2025persistent}, we propose commutative algebra \(k\)-mer representations (CAKR) to integrate \(k\)-mers representations of sequences~\cite{hozumi2024revealing} with persistence modules arising from Stanley-Reisner constructions. CAKR is evaluated on twelve primary datasets: one for genetic variant classification, seven for phylogenetic tree reconstruction, and four for viral classification derived from the National Center for Biotechnology Information (NCBI) Virus database. In addition, we include two supplementary fragment-placement benchmarks for barcode-based influenza query matching and comparison with SEPP.

We systematically compare CAKR against five state-of-the-art alignment-free methods: the Natural Vector Method (NVM)~\cite{sun2021geometric}, the Markov \(k\)-string model (MKS)~\cite{QiWangHao2004}, Feature Frequency Profile with Jensen-Shannon (FFP-JS) divergence and with Kullback-Leibler (FFP-KL) divergence \cite{sims2009alignment, jun2010whole}, and  Fourier Power Spectrum (FPS)~\cite{hoang2015fourier} and an alignment-based method, MAFFT~\cite{katoh2013mafft}.
Section~\ref{sec:results} presents the experimental results. Section ~\ref{sec:discussion} discusses performance, limitations, and generalization. Section~\ref{sec:methods} describes the proposed CAKR methodology and introduces a new purity metric for evaluating the quality of phylogenetic tree reconstruction.

\section{Results}
\label{sec:results}

\subsection*{An overview of CAKR}

\begin{figure}[!h]
	\centering
	\begin{subfigure}[b]{0.99\textwidth}
		\includegraphics[width=\textwidth]{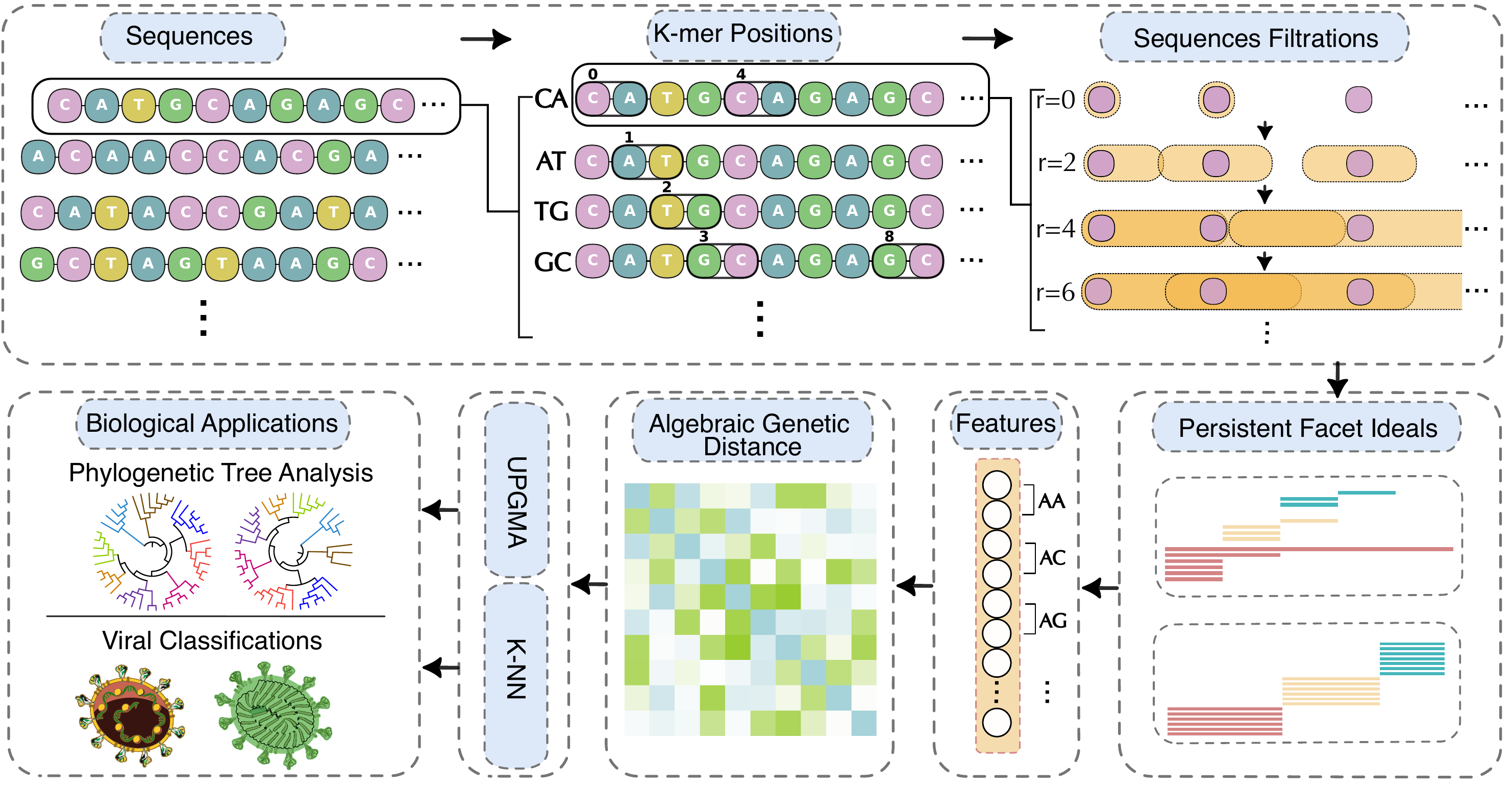}
	\end{subfigure}
\caption{\textbf{Illustration of the CAKR workflow.}
	Given a query sequence, \(k\)-mers are first extracted. For each \(k\)-mer,
	the set of its occurrence positions within the sequence is treated as a
	sequence of integers. The persistent facet numbers associated with these
	sequences of integers are then computed and used to represent the
	corresponding \(k\)-mer. The feature vectors of all \(k\)-mers of the same
	size are concatenated to construct an algebraic representation. Pairwise
	distances between these sequences are subsequently defined and used for
	tasks such as genome variant classification, phylogenetic tree
	reconstruction, and genome classification. SARS-CoV-2 artwork adapted
	from NIAID NIH BIOART Source BIOART-000465
	\cite{niaid2024sarscov2_bioart} (public domain; courtesy of NIAID);
	figure assembled by the authors using Inkscape.
	CAKR, commutative algebra \(k\)-mer representation; UPGMA, unweighted
	pair group method with arithmetic mean; K-NN, \(K\)-nearest neighbors.
	The symbol \(r\) denotes the filtration radius. Nucleotide colours denote
	A (teal), C (purple), G (blue), and T (yellow); the remaining colours
	distinguish schematic components and carry no quantitative meaning.}
\label{fig:workflow}
\end{figure}

To evaluate the effectiveness of the proposed CAKR method, illustrated by the workflow in Fig.~\ref{fig:workflow}, we consider three key applications: genetic variant classification, phylogenetic tree reconstruction, and viral classification. Genetic variant classification is an important application in genetics and bioinformatics. It tracks genetic variations in selected gene lists associated with specific diseases, phenotypes, and populations, for which alignment methods are typically favored. Additionally, phylogenetic tree reconstruction of genetic sequences plays a fundamental role in elucidating evolutionary relationships both among and within species, where alignment-free methods have advantages \cite{nei1996phylogenetic, bellgard1999dynamic}. Finally, viral classification is a general machine learning approach for the genetic analysis and prediction of unknown viral sequences. It is challenging to design a unified approach for these diverse tasks.

\subsection*{Genetic variant classification} \label{subsec:viral_variant}

\begin{figure}[!h]
	\centering
	\begin{subfigure}[b]{1.0\textwidth}
		\includegraphics[width=\textwidth]{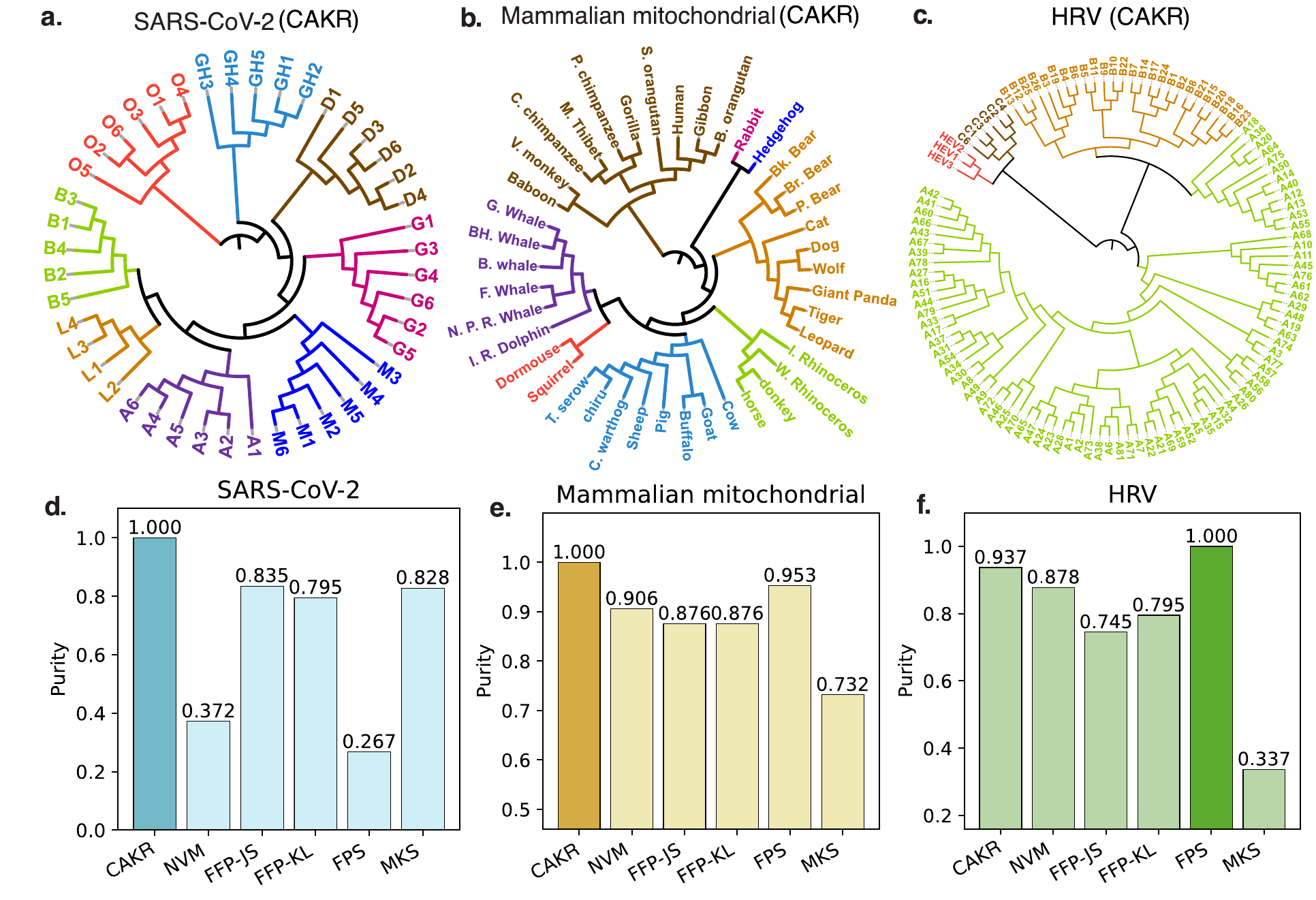}
	\end{subfigure}
\caption{\textbf{Illustration and comparison of the proposed
		CAKR model with other methods in variant classification and phylogenetic
		tree reconstruction.}
	\textbf{a,} CAKR classification of SARS-CoV-2 variants.
	\textbf{b--c,} CAKR phylogenetic tree reconstructions of the mammalian
	mitochondrial genome and HRV datasets.
	\textbf{d--f,} Comparison of prediction accuracies of different methods
	on the SARS-CoV-2, mammalian mitochondrial, and HRV datasets,
	respectively.
	CAKR, commutative algebra \(k\)-mer representation; NVM, natural vector
	method; FFP-JS, feature frequency profile with Jensen--Shannon divergence;
	FFP-KL, feature frequency profile with Kullback--Leibler divergence; FPS,
	Fourier power spectrum; MKS, Markov \(k\)-string model; HRV, human
	rhinovirus. Branch and label colours in \textbf{a--c} denote the known
	variant or taxonomic groups shown by the corresponding labels. In
	\textbf{d--f}, the darker bar indicates the highest-performing method,
	whereas the paler bars indicate the remaining methods. The methods compared
	in \textbf{d--f} are identified beneath the corresponding bars. Source data
	are provided as a Source Data file.}
\label{fig:tree-accuracies}
\end{figure}

The emergence of SARS-CoV-2 variants during the COVID-19 pandemic posed critical challenges for monitoring viral evolution, guiding public health measures, and informing vaccine and therapeutic design~\cite{tao2021biological}. Variant differences can be subtle, often involving only a few mutations in the spike protein receptor-binding domain~\cite{chen2022omicron}, despite genome lengths of approximately 29.9~kb.
For the genetic variant classification task, we analyzed 44 complete SARS-CoV-2 genomes from GISAID~\cite{hozumi2024revealing, Li2024GenomeGrassmann, khare2021gisaid}.
This dataset follows the curated benchmark of Li~et al.~\cite{Li2024GenomeGrassmann}, where representative complete genomes were selected to cover major variant lineages (Alpha, Beta, Gamma, Delta, Lambda, Mu, GH/490R, and Omicron) while maintaining a compact evaluation set. Each genome is labeled according to its assigned lineage, and branches are annotated and color-coded accordingly in our phylogenetic trees. We employed CAKR with \(k=6\), as determined by our data-driven scale selection framework, which identifies an optimal \(k\)-mer size through a weighted consensus of coverage, information-sparsity, and stability-adjusted entropy criteria applied to the dataset. Because SARS-CoV-2 genomes are highly similar across lineages, we use a slightly larger \(k\) to capture subtle variant-defining patterns while remaining computationally practical. NVM was run with \(k=5\), while FFP-KL, FFP-JS, and MKS used \(k=3\); FPS required no \(k\)-mer parameter.
Among six alignment-free methods, only CAKR achieved perfect clustering, producing well-defined monophyletic groups (Fig.~\ref{fig:tree-accuracies}a, Supplementary Fig.~1) and accurately resolving inter-variant relationships. FFP-JS, FFP-KL, and MKS attained moderate accuracy but mis-clustered some Delta, Gamma, GH/490R, and Omicron samples, whereas NVM and FPS failed to produce meaningful phylogenetic structure. Performance was quantified using the label-based purity metric given by equation~\eqref{eq:purity_metric}, which measures the proportion of samples with identical labels grouped under a common ancestor. Overall, CAKR achieved the highest purity score among all methods (Fig.~\ref{fig:tree-accuracies}d).

\subsection*{Phylogenetic tree reconstruction}
\label{subsec:phylogenetic}

Phylogenetic reconstruction is central to elucidating evolutionary relationships among taxa~\cite{nei1996phylogenetic, bellgard1999dynamic}. We assessed CAKR on six benchmark datasets from~\cite{hozumi2024revealing}, encompassing complete genomes and gene sequences with established taxonomic labels. Sequence lengths range from $\sim$2{,}000 nt for influenza \textit{HA} genes to $\sim$17{,}000 nt for mammalian mitochondrial and Ebola virus genomes, and up to several megabases for bacterial genomes.

Phylogenetic trees were inferred using CAKR and compared with those from representative alignment-free methods. Following~\cite{hozumi2024revealing}, we set \(k = 3\) for FFP-KL, FFP-JS, and MKS models, \(k = 5\) for NVM, and selected \(k=4\) for CAKR using our data-driven $k$-mer scale selection framework; FPS required no \(k\)-mer parameter. All trees were built with the UPGMA algorithm and visualized in Interactive Tree Of Life (iTOL) v6~\cite{letunic2024interactive}. As a supplementary sensitivity analysis, we also compared CAKR trees built with UPGMA and Neighbor-Joining (NJ). UPGMA achieved perfect purity on the rhinovirus, mammalian, and SARS-CoV-2 datasets, whereas NJ reduced purity for rhinovirus and mammalian data while matching UPGMA on SARS-CoV-2 (Supplementary Section~7; Supplementary Fig.~15).

In the mammalian mitochondrial dataset, the goal was to assess how accurately each method recovers clades consistent with known host species classifications. CAKR achieved perfect concordance, producing monophyletic clades for all major mammalian orders, including delineated Primates, Cetacea, and Artiodactyla. It also preserved internal coherence within Carnivora and Perissodactyla (Fig.~\ref{fig:tree-accuracies}b; Supplementary Fig.~3). Furthermore, CAKR attained the highest purity score among all methods (Fig.~\ref{fig:tree-accuracies}e).

Other methods showed varying levels of discordance: FPS partially recovered major orders but split Artiodactyla; NVM misplaced several taxa and fragmented Carnivora; FFP-JS and FFP-KL preserved many lineages but mis-grouped smaller orders and split Carnivora; and MKS produced the weakest resolution, fragmenting multiple orders entirely.

In the HRV dataset, CAKR separates HRV-B and HRV-C into coherent clusters and places HEV as an outgroup (Fig.~\ref{fig:tree-accuracies}c), while HRV-A shows less consistent clustering, with a small number of HRV-A taxa occurring near the HRV-B region of the tree. On the other hand, FPS achieved perfect classification, cleanly separating HRV-A, HRV-B, and HRV-C clades, with the outgroup (HEV) distinctly isolated (Supplementary Fig.~4). NVM, FFP-JS, and FFP-KL produced moderate results, partially recovering the three clades but misplacing several HRV-A genomes within HRV-B subtrees, indicating reduced sensitivity to closely related subtypes. MKS performed poorest, failing to recover a coherent subgroup structure and producing extensive intermixing among clades, along with poor outgroup resolution.

In the HEV dataset, CAKR, FFP-JS, FFP-KL, and NVM all achieved perfect genotypic clustering (Supplementary Fig.~5). MKS misclassified one Group~3 genome, while FPS performed worst, failing to separate Groups~3 and~4, though correctly clustering Group~1.

In the influenza \textit{HA} dataset, CAKR, FFP-JS, and FFP-KL all achieved perfect subtype classification (Supplementary Fig.~6), though FFP-JS and FFP-KL grouped H3N2 and H2N2 under a shared node, unlike CAKR. NVM failed to cluster H1N1 cohesively and misclassified one H2N2 sequence, while MKS misclassified an H1N1. FPS performed worst, correctly clustering only the H7N9 subtype.

In the ebolavirus dataset, all methods correctly separated the five known species (Supplementary Fig.~7). NVM, FFP-JS, FFP-KL, and MKS distinguished epidemic lineages except for the 1996 outbreak, with NVM producing the longest inter-clade branches. FFP-JS and FFP-KL showed shorter branches, indicating weaker sensitivity, while MKS grouped EBOV and RESTV under a common ancestor. FPS failed to resolve the epidemic-level structure within EBOV.

To evaluate performance on longer sequences, we applied CAKR to 30 complete bacterial genomes ranging from 0.9 to 6.5~Mb. Despite the substantial increase in sequence size and complexity, all methods except NVM and FPS recovered correct family-level groupings (Supplementary Fig.~8). Overall, CAKR consistently delivered accurate and biologically coherent phylogenies, outperforming existing alignment-free methods across diverse genomic datasets. A qualitative cross-dataset summary of phylogenetic reconstruction behavior is provided in Supplementary Table~2.

To quantify phylogenetic consistency throughout our experiments, we use the purity metric in Eq.~(\ref{eq:purity_metric}) as our primary evaluation measure (Fig.~\ref{fig:tree-accuracies}). Purity summarizes the extent to which sequences sharing the same taxonomic label form coherent (approximately monophyletic) groups in the inferred tree, providing a label-based assessment aligned with known classifications.

In addition, we report the normalized Robinson-Foulds (nRF) distance between CAKR trees and the MAFFT-based reference trees as a complementary, topology-only measure of tree-to-tree similarity. Averaged across datasets, CAKR attained an nRF of $0.39$. Further details are provided in Supplementary Section~3.

To further assess discrimination in a narrow-clade setting, we applied CAKR and the competing baselines to two closely related Salmonella datasets. In a focused dataset of $294$ genomes, comprising S. enterica subspecies arizonae ($n=35$), diarizonae ($n=17$), enterica ($n=194$), and salamae ($n=48$), CAKR achieved the highest average purity ($0.8885$), exceeding MKS ($0.8831$), FFP-KL and FFP-JS (both $0.7417$), FPS ($0.6920$), and NVM ($0.5938$). We then expanded the analysis to a broader dataset of $519$ closely related Salmonella genomes. On this more extensive narrow-clade benchmark, MKS and FFP remained competitive, while NVM and FPS exhibited substantially reduced resolution. Moreover, CAKR again achieved the highest purity, outperforming MKS and FFP by at least $3.9\%$.

\begin{table}[H]
	\centering
	\caption{\textbf{One-nearest-neighbor viral classification accuracy.}
		CAKR was compared with NVM, FFP-JS, FFP-KL, FPS and MKS on four NCBI viral genome datasets. Values report 1-NN classification accuracy. Source data are provided as a Source Data file.}
	\label{tab:1nn_accuracy}
	\begin{threeparttable}
		\begin{tabular}{lcccccc}
			\toprule
			\textbf{Data} & \textbf{CAKR} & \textbf{NVM} & \textbf{FFP-JS} & \textbf{FFP-KL} & \textbf{FPS} & \textbf{MKS} \\
			\midrule
			NCBI 2020     & \textbf{0.932} & 0.879 & 0.862 & 0.862 & 0.732 & 0.734 \\
			NCBI 2022     & \textbf{0.920} & 0.875 & 0.870 & 0.870 & 0.732 & 0.735 \\
			NCBI 2024     & \textbf{0.891} & 0.829 & 0.825 & 0.826 & 0.656 & 0.637 \\
			NCBI 2024 All & \textbf{0.892} & 0.825 & 0.832 & 0.832 & 0.647 & 0.647 \\
			\bottomrule
		\end{tabular}
		\begin{tablenotes}[flushleft]
			\footnotesize
			\item Note: Boldface indicates the highest accuracy within each row.
		\end{tablenotes}
	\end{threeparttable}
\end{table}

\subsection*{Viral classification }
\label{subsec:viral_classification}

\begin{figure}[!h]
	\centering
	\begin{subfigure}[b]{1\textwidth}
		\includegraphics[width=\textwidth]{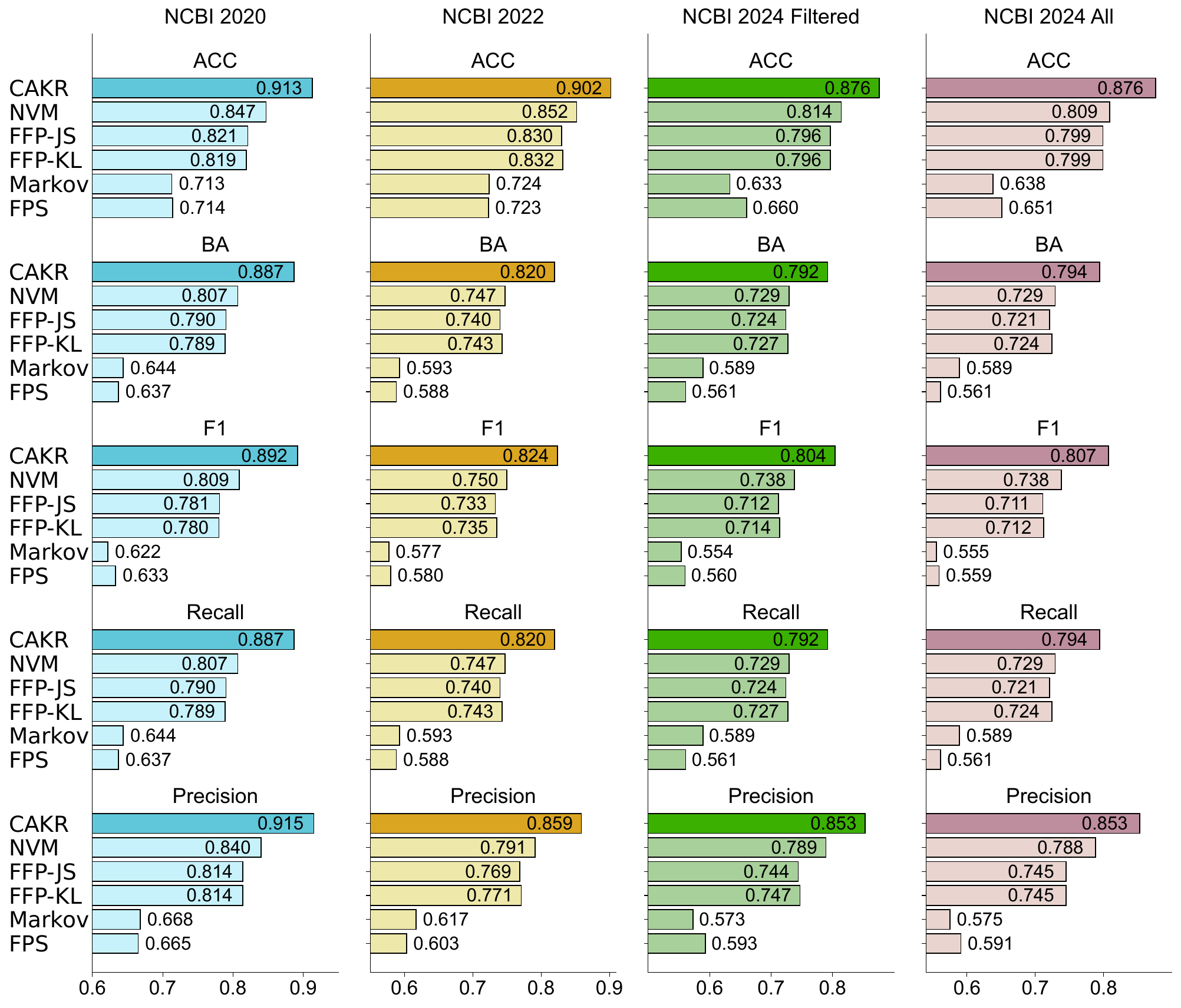}
	\end{subfigure}
\caption{\textbf{Performance of CAKR in 5-nearest-neighbor
		viral classification.}
	CAKR was compared with five alignment-free methods, NVM, FFP-JS, FFP-KL,
	FPS, and MKS, on four NCBI viral genome datasets: NCBI 2020, NCBI 2022,
	NCBI 2024, and NCBI 2024 All. Classification was performed using a
	5-nearest-neighbor classifier with stratified 5-fold cross-validation
	repeated over 30 random seeds. Performance was evaluated using accuracy
	(ACC), balanced accuracy (BA), macro-F1 score, recall, and precision.
	Macro-averaged scores are reported to account for class imbalance among
	viral families. CAKR achieved the highest score across all datasets and
	evaluation metrics.
	CAKR, commutative algebra \(k\)-mer representation; NVM, natural vector
	method; FFP-JS, feature frequency profile with Jensen--Shannon divergence;
	FFP-KL, feature frequency profile with Kullback--Leibler divergence; FPS,
	Fourier power spectrum; MKS, Markov \(k\)-string model; NCBI, National
	Center for Biotechnology Information; ACC, accuracy; BA, balanced
	accuracy; F1, macro-averaged F1 score. Panel colours distinguish the four
	dataset versions: blue, gold, green, and pink represent NCBI 2020,
	NCBI 2022, NCBI 2024, and NCBI 2024 All, respectively. Within each panel, the darker bar in each metric group indicates the
	highest-performing method, whereas the paler bars indicate the remaining
	methods. Method names are indicated beside the bars in panel \textbf{a}
	and apply in the same order to panels \textbf{b--d}. Source data are
	provided as a Source Data file.}
\label{fig:NCBI}
\end{figure}

In this task, we conduct viral classification experiments on the four NCBI datasets described in Supplementary Section 1, following the experimental design outlined by Hozumi et al.~\cite{hozumi2024revealing}. In particular, we adopt the same comparative benchmarking strategies to ensure consistency. The labels assigned to viral sequences in the NCBI Virus Database are regularly revised, as the International Committee on Taxonomy of Viruses (ICTV) continuously updates viral classifications based on new scientific findings. This ongoing process reflects the complexity of the classification problem and emphasizes the importance of using methods that remain reliable despite changes in biological taxonomy. A summary of the dataset versions, filtering criteria, and sample sizes is provided in Supplementary Table~1.
Moreover, many of the NCBI genomes analyzed in this study already include ambiguous nucleotides (e.g., \texttt{N}, \texttt{R}, \texttt{W}, \texttt{Y}), reflecting uncertainty or partial information that commonly arises in real sequencing/assembly pipelines. The stable performance observed across these datasets therefore provides empirical evidence that CAKR is robust to a moderate level of incompleteness.

Two classification tasks are considered. The first employs a 1-nearest neighbor (1-NN) classifier on feature representations derived from the persistent facet ideals featurization, following the methodology introduced by Sun et al.~\cite{sun2021geometric}. A test sample is deemed correctly classified if its nearest neighbor, under the algebraic distance metric, shares the same viral family label. This protocol models realistic scenarios in which newly sequenced viral genomes are annotated based on proximity to previously characterized reference genomes.

The second task involves a 5-nearest neighbors (5-NN) classifier, using the same experimental protocol described in~\cite{hozumi2024revealing}. Performance is assessed via 5-fold cross-validation repeated over 30 random seeds to ensure statistical robustness. To control for class imbalance and maintain classification reliability, the evaluation is restricted to viral families with at least 15 representative sequences. In both tasks, we empirically observed that increasing the value of \( k \) in the \( k \)-mer algebraic representations generally leads to a decline in model performance. This behavior is expected, as larger values of \( k \) result in more distinct representations for each genome, irrespective of their biological classification, thereby reducing the amount of shared structural information that can be effectively used for grouping. The model exhibits consistently strong performance for \( k = 3, 4, 5 \), with \( k = 4 \) selected for use in both tasks.

For each dataset, stratified 5-fold cross-validation was conducted using 30 independent random seeds to obtain performance metrics. Classification performance was assessed using accuracy (ACC), balanced accuracy (BA), macro-F1 score (F1), recall, and precision. All metrics were computed using the macro-averaging scheme to ensure equal weight across viral families, regardless of class imbalance.

All methods exhibit a consistent decrease in accuracy from 2020 to 2024 across both classification tasks. As summarized in Table~\ref{tab:1nn_accuracy}, our model exhibits a strong predictive performance under the 1-nearest neighbor (1-NN) classification protocol. On the NCBI 2020 dataset, consisting of 6,993 samples, the model achieved an accuracy of \( 0.932 \). For the larger NCBI 2022 dataset comprising 11,428 samples, an accuracy of \( 0.920 \) was obtained. For the NCBI 2024 dataset, we evaluated performance under two settings: for the filtered set of 12,154 samples, an accuracy of \( 0.891 \) was obtained, while for the complete set containing 13,645 samples, the model achieved an accuracy of  \( 0.892 \).

Furthermore, using 5-nearest neighbors classification with 5-fold cross-validation, as summarized in Fig.~\ref{fig:NCBI} and Supplementary Table~3, our proposed model demonstrates higher predictive performance compared to existing state-of-the-art approaches across all datasets. Specifically, on the NCBI 2020 dataset, our model achieved an accuracy of \( 0.913 \). On the larger NCBI 2022 dataset, the model attained an accuracy of \( 0.902 \). On the filtered NCBI 2024 dataset, the model achieved an accuracy of \( 0.876 \), while on the full NCBI 2024 dataset containing all samples, an accuracy of \( 0.876 \) was recorded. These results underscore the robustness and effectiveness of our method in modeling complex viral genome spaces and establishing reliable predictive frameworks for viral classification tasks.

As a further assessment of robustness, we quantified the stability of method rankings using Kendall’s coefficient of concordance ($W$). By treating evaluation metrics as judges within each dataset and datasets as judges within each metric, we consistently observed high agreement ($W=0.965$--$0.982$ and $W=0.964$--$0.979$, respectively; Supplementary Tables~4 and~5). The associated $\chi^2$ tests confirmed that this concordance is statistically significant, indicating that the relative ordering of methods is largely stable with respect to both the choice of performance metric and the specific dataset version. In particular, the top-ranked position of CAKR remained unchanged across all metrics and datasets, demonstrating that the observed performance hierarchy is statistically robust rather than metric- or dataset-dependent. For further details, see Supplementary Section~8.

\section{Discussion}
\label{sec:discussion}

\paragraph{Overall performance.}
Across all benchmark evaluations, CAKR demonstrates consistently strong overall performance relative to the five baseline alignment-free methods. On six phylogenetic tree\textendash construction datasets, it achieves the highest overall tree accuracy, with the purity index~\eqref{eq:purity_metric} exceeding that of every competitor by at least \(4\) percentage points on average (absolute difference; Fig.~\ref{fig:tree-accuracies}). On the four extensive NCBI collections curated in~\cite{yu2024optimal, hozumi2024revealing, sun2021geometric}, CAKR attains the top macro-averaged scores for accuracy (ACC), balanced accuracy (BA), \(F_{1}\), recall, and precision, surpassing the runner-up by \(4\text{-}7\) percentage points across all datasets (Fig.~\ref{fig:NCBI}).

This improvement stems from CAKR's explicit encoding of the spatial distribution of \(k\)-mers through locality-sensitive features, in contrast to most alignment\textendash free methods, which rely primarily on \(k\)-mer frequency distributions with limited positional or spatial information. This positional awareness enables CAKR to resolve subtle yet biologically consequential sequence variants, thereby enhancing performance in large-scale viral classification, as well as in phylogenetic tree reconstruction and genetic variant classification.

\paragraph{Comparable to alignment methods for variant classification}
Alignment-based tools such as MAFFT are highly effective for within-species variant inference, where sequences are sufficiently similar and share reliable positional homology. As divergence increases or when heterogeneous genomic content reduces the stability or interpretability of a single global alignment across taxa, alignment-based comparisons may become less robust; in such remote-homology regimes, profile-HMM approaches (e.g., HMMER) are often used to detect conserved gene/protein-family relationships. To compare CAKR and MAFFT under conditions favorable to alignment methods, we analyzed three benchmark datasets, such as SARS-CoV-2, mammalian mitochondrial genomes, and human rhinovirus (HRV). The trees produced by CAKR and MAFFT show strong agreement: both recover the principal SARS-CoV-2 variant groups (Alpha, Beta, Gamma, Delta, Lambda, Mu, GH/490R, Omicron), correctly group Primates, Carnivora, and Cetacea in the mammalian set, and delineate HRV-A, HRV-B, HRV-C with the HEV outgroup (Fig.~\ref{fig:tree-accuracies}; Supplementary Fig.~9). Hence, CAKR attains alignment-level agreement in regimes where alignment is presumed strongest for variant inference, while retaining robustness and efficiency for heterogeneous or multi-species data. On the other hand, CAKR can directly compare the similarity of  entirely different sequences, such as those of HRV, HIV, and SARS-CoV-2, for which MAFFT, Clustal Omega, and other alignment methods do not work.

\paragraph{Robustness.}
CAKR maintains high accuracy as the dataset size and quality vary. In a 5–NN setting, its accuracy declines only modestly, from \(91.3\%\) on the NCBI 2020 dataset consisting of 6993 sequences to \(87.6\%\) on the 13\,645-sequence {NCBI 2024 All} dataset, exhibiting similar stability when error-containing reads are included. A 1–NN evaluation shows the same pattern, with performance decreasing from \(93.2\%\) to \(89.2\%\) across the same data progression. Empirically, \(k\in\{3,4,5, 6\}\) suffices for strong results, with \(k=4\) giving the best single-model accuracy; combining these values in an ensemble further boosts performance with negligible additional cost.

In addition to dataset-scale robustness, Supplementary Section~10 evaluates robustness to incomplete and fragmented sequence inputs. Under simulated truncation, contiguous deletion, and fragmentation, CAKR retained strong clustering purity for most benchmark datasets, with the retained fraction of sequence information having the largest effect on performance. The average-purity trends across \(k\) under simulated incompleteness are summarized in Supplementary Fig.~19. A complementary barcode-based fragment placement experiment with external influenza queries further showed that partial sequences can still be matched to the correct viral type and homologous gene segment in an alignment-free manner.

To assess whether performance is driven mainly by simple compositional effects, we further performed sequence-bias control experiments involving low-complexity masking, repeat masking, GC normalization, and mono- and dinucleotide shuffling. Across datasets, the strongest degradations occurred under shuffling, supporting the view that CAKR captures higher-order sequence organization rather than simple composition alone (Supplementary Section~11; Supplementary Figs.~20-25).

\paragraph{Interpretability of CAKR.}
Rational learning requires explainable features and an interpretable neural network design.
Fig.~\ref{fig:psrt_ex} illustrates the interpretability of CAKR through two representative examples, demonstrating how the Persistent Stanley-Reisner Theory (PSRT) encodes the filtration structure via the algebraic decomposition of facet ideals.
Panels~(a-d) depict a synthetic example of eight points uniformly placed on a circle of radius~1. As the filtration parameter increases, simplicial complexes are constructed by adding a $k$-simplex whenever all its vertices lie within the closed ball centered at at least one of its vertices. Panel~(b) shows the barcodes of persistent facet ideals: $\mathcal{P}_0$ encodes vertex lifespans (isolation), $\mathcal{P}_1$ tracks edge persistence until subsumed into higher simplices, and $\mathcal{P}_2$, $\mathcal{P}_3$ capture 2- and 3-simplices respectively. The barcodes reveal uniform connectivity and regular geometric structure in this example.

Panels~(e-f) present a biologically motivated example, examining the distribution of the 1-mer \texttt{C} in the N1-U.S-P primer. Each occurrence of \texttt{C} is treated as an integer in 1D space. Equal-length bars in $\mathcal{P}_0$ reflect regularly spaced cytosines, while longer bars indicate isolated ones. Notably, one isolated \texttt{C} connects to its nearest neighbor at filtration radius 5, and its corresponding edge vanishes at radius 7, suggesting nearby higher-order interactions. The brief lifespans in $\mathcal{P}_2$ indicate that 2-simplices (triangles) are formed quickly and filled soon thereafter, implying tight local clustering among cytosines.

Panels~(c) and~(g) show the persistent $f$-vectors, $f_k(r)$, quantifying the number of active $k$-simplices at filtration value $r$, while Panels~(d) and~(h) show the corresponding $h$-vectors, $h_k(r)$, measuring the incremental additions of independent $k$-facets. In both examples, increases in $f_1$ signal edge formation, and any growth in $f_k$ for $k \geq 2$ is necessarily preceded by growth in $f_1$, making $f_1$ a necessary precursor for higher-order structure. The $h$-vectors further reveal when such structures are nontrivial versus when they are absorbed into larger simplices.

Altogether, CAKR’s interpretability stems from its ability to encode $k$-mer spatial configurations as algebraic signatures derived from the evolution of facet ideals across the filtration. The persistent $f$- and $h$-vectors provide structured summaries of how generators of these ideals emerge and interact at varying scales. In particular, they capture the combinatorial and geometric regularity of $k$-mer distributions, offering insights into clustering, isolation, and interaction patterns among sequence motifs in both synthetic and biological settings.

\paragraph{Generalizability.}
Because CAKR encodes a sequence purely as a word over a finite alphabet, it extends naturally beyond the DNA datasets analysed here. In practice, we transcribed RNA sequences to their DNA equivalents, but the same framework also accommodates amino-acid strings for protein sequence modeling. More broadly, any categorical sequence data built from a limited alphabet of symbols can be modeled by the proposed commutative-algebraic constructions. The proposed CAKR therefore provides a versatile mathematical foundation for sequence analysis and, by extension, for data science tasks involving symbolic or ordered data.

Although all datasets considered in this work consist of single-segment viral or bacterial genomes, the CAKR representation is defined for arbitrary DNA strings and therefore extends directly to multi-chromosomal organisms. For such genomes, one may either (i) featurize and aggregate the CAKR descriptors of all chromosomes to obtain a single organism-level representation for phylogeny and classification or (ii) analyze chromosomes individually to investigate chromosome-specific structure. In the organism-level setting, genomes from the same species would generally be expected to cluster together in the CAKR feature space. In the chromosome-level setting, chromosomes from the same species may still appear closer to one another because they often share similar compositional and mutational biases, and therefore exhibit broadly similar low-order \(k\)-mer distributions. Moreover, CAKR leverages both \(k\)-mer content and locality/positional information, which can further promote within-species similarity when such genomic signatures are shared. However, the method does not assume that all chromosomes from a given species must form a single tight cluster. It can therefore capture biologically meaningful heterogeneity associated with sex chromosomes, structural variation, repeats, or horizontal transfer.

\paragraph{Extended comparisons.}
Several well-established methods provide useful additional points of comparison with CAKR. Among alignment-free approaches, we considered CVTree~\cite{zuo2015cvtree3} and Mash~\cite{ondov2016mash}, while among widely used phylogenetic inference frameworks we included IQ-TREE~3 and RAxML-NG~\cite{wong2025iqtree3,kozlov2019raxmlng}. We evaluated these methods on representative benchmark datasets, including SARS-CoV-2, mammalian mitochondrial genomes, and rhinoviruses where applicable.

For the SARS-CoV-2 dataset, the purity scores were 0.9028 for CVTree (Supplementary Fig.~10), 0.9028 for Mash (Supplementary Fig.~11), 0.9600 for IQ-TREE~3 (Supplementary Fig.~2(a)), and 0.9028 for RAxML-NG (Supplementary Fig.~2(b)). On the mammalian dataset, CVTree and Mash achieved purity scores of 0.8438 and 0.9531, respectively, whereas on the rhinovirus dataset their purity scores were 0.7778 and 0.6485, respectively. These results show that both additional alignment-free baselines and standard phylogenetic inference frameworks can perform strongly on selected datasets, with IQ-TREE~3 performing particularly well on SARS-CoV-2. At the same time, CAKR remains competitive while maintaining a fully alignment-free, sequence-representation-based framework. The corresponding phylogenetic trees produced by CVTree, Mash, IQ-TREE~3, and RAxML-NG are shown in Supplementary Figs.~10, 11, and~2, respectively, and the detailed CAKR-SEPP fragment-placement comparison is provided in Supplementary Section~10 and Source Data file.

In addition, because fragmentary sequence placement is an important practical setting that is not fully captured by whole-tree reconstruction alone, we also compared CAKR in the Supplementary Information with SEPP, a standard HMM-based phylogenetic placement framework~\cite{mirarab2012sepp}. In the influenza B \textit{PB1} fragment-placement experiment, the two frameworks showed substantial agreement at the level of ranked candidate placements: exact Top~1 agreement was observed for 6 of 15 fragmentary queries, and 11 of 15 fragments showed overlap between the top three CAKR candidates and the top two SEPP placements. Since the two approaches rely on different scoring systems, the comparison is most meaningful in terms of agreement in ranked placements rather than raw score magnitudes. These results indicate that CAKR compares favorably with established alternatives across both whole-sequence phylogenetic reconstruction and fragment-placement settings, while retaining the practical advantage of being fully alignment free. Detailed results for this comparison are provided in Supplementary Section~10 and Source Data file.

\paragraph{Computational efficiency.}
Although CAKR consistently outperformed the competing methods in phylogenetic accuracy and viral classification, it remains computationally practical for genome-scale analyses. As expected from its richer algebraic construction, CAKR is more demanding at the feature-extraction stage. Thus, the practical trade-off is a higher one-time featurization cost in exchange for improved downstream phylogenetic and classification performance. For fixed \(k\), all methods exhibit similar runtime growth with respect to sequence length, while CAKR incurs a higher per-sequence featurization cost due to Vietoris-Rips complex construction and the computation of persistent facet invariants. For example, facet-based CAKR featurization increases from \(\approx 4.6\times 10^{-3}\,\mathrm{s}\) at length \(10\) to \(\approx 2.8\times 10^{2}\,\mathrm{s}\) at length \(3\times 10^{6}\). These timings refer to featurization only; pairwise distances are subsequently computed in the resulting feature space, where the low dimensionality at small \(k\) keeps distance calculations inexpensive. Complete wall-clock runtime comparisons against NVM, FFP-JS, FFP-KL, FPS, MKS, and MAFFT across multiple sequence lengths and \(k\)-mer sizes are provided in Supplementary Section~5 and Supplementary Fig.~12.

In terms of memory, the final CAKR feature representation is compact for the parameter choices used in this study. For \(N\) DNA sequences, \(k\)-mer size \(k\), \(F\) filtration values, and maximum facet dimension \(d_{\max}\), the dense feature matrix has size
\(
N\times \bigl((d_{\max}+1)F4^k\bigr).
\)
In our experiments, we used \(d_{\max}=0\) and \(F=3\), so the storage requirement is \(24N4^k\) bytes in double precision. Consequently, for the values \(k=3,4,5,6\) used in the benchmarks, the final stored feature arrays remain modest for the phylogenetic datasets and manageable for the larger viral classification datasets. The main computational burden is therefore the one-time featurization step, rather than the storage of the resulting features or the subsequent pairwise distance calculations. Detailed memory estimates are provided in Supplementary Section~5.

Overall, CAKR computes richer topological and algebraic descriptors than simpler \(k\)-mer frequency methods, which makes the initial feature-extraction stage more expensive. However, the resulting fixed-length representations are compact, inexpensive to compare, and can be reused for downstream distance calculation, clustering, tree reconstruction, and classification. Thus, CAKR is particularly useful when improved phylogenetic or classification accuracy is prioritized over minimizing featurization time, while simpler \(k\)-mer methods may remain preferable for very rapid exploratory screening.

\paragraph{Incomplete sequences.}
Our method is primarily designed for global, alignment-free comparison of complete sequences, since it relies on the overall \(k\)-mer distribution and the associated positional and combinatorial structure across the sequence. For this reason, CAKR is not intended as a direct replacement for local alignment, short-read mapping, or contig assembly methods, which require explicit base-level local matching.

At the same time, to assess applicability beyond ideal full-length inputs, we performed additional analyses on incomplete genomes, fragmented sequences, and fragment placement-style inference (Supplementary Section~10). In a simulated incompleteness experiment, benchmark datasets were degraded through truncation, contiguous deletion, and fragmentation to mimic partial recovery and draft-quality assemblies. Across most datasets, clustering purity remained strong over a wide range of degradation settings, with the retained fraction of sequence information emerging as the dominant determinant of performance, whereas fragmentation alone did not produce a consistent collapse in clustering quality. In particular, Ebolavirus, rhinovirus, HEV, and influenza \textit{HA} gene remained comparatively robust, whereas mammalian mitochondrial genomes showed more moderate sensitivity. For SARS-CoV-2, purity remained low across nearly all incomplete settings, indicating that the difficulty in this case is driven more by the intrinsic similarity of the sequences rather than by incompleteness itself.

We further evaluated a barcode-based fragment placement setting using external influenza query sequences. In this proof-of-concept experiment, CAKR barcodes computed from partial or complete query genes were compared against a reference barcode library derived from influenza genome segments. Each of the 14 external query sequences attained its highest cosine similarity to the correct viral type and the corresponding homologous gene segment (Source Data file), indicating that CAKR can retain meaningful discriminatory information even for previously unseen partial inputs.

These results show that CAKR does not rely exclusively on perfectly assembled full genomes and can remain effective under substantial truncation, deletion, and fragmentation, while also supporting fragment placement-style inference in an alignment-free manner. Nevertheless, when incompleteness becomes substantial, the resulting representations may deviate more strongly from those of the corresponding complete genomes, potentially introducing variability in clustering, distance estimation, and downstream inference. Large-scale benchmarking on true low-coverage assemblies or shotgun metagenomic datasets remains an important direction for future work.

\section{Methods}
\label{sec:methods}

In this section, we provide an overview of persistent Stanley-Reisner theory. Then, we provide the construction of the k-mer algebraic representations of sequences.
We also propose a new purity metric to assessing the performance of phylogenetic tree reconstruction tools.

\subsection{Persistent Stanley-Reisner theory}
\label{subsec:PSRT}

\begin{figure}[!h]
	\centering
	\begin{subfigure}[b]{0.99\textwidth}
		\includegraphics[width=\textwidth]{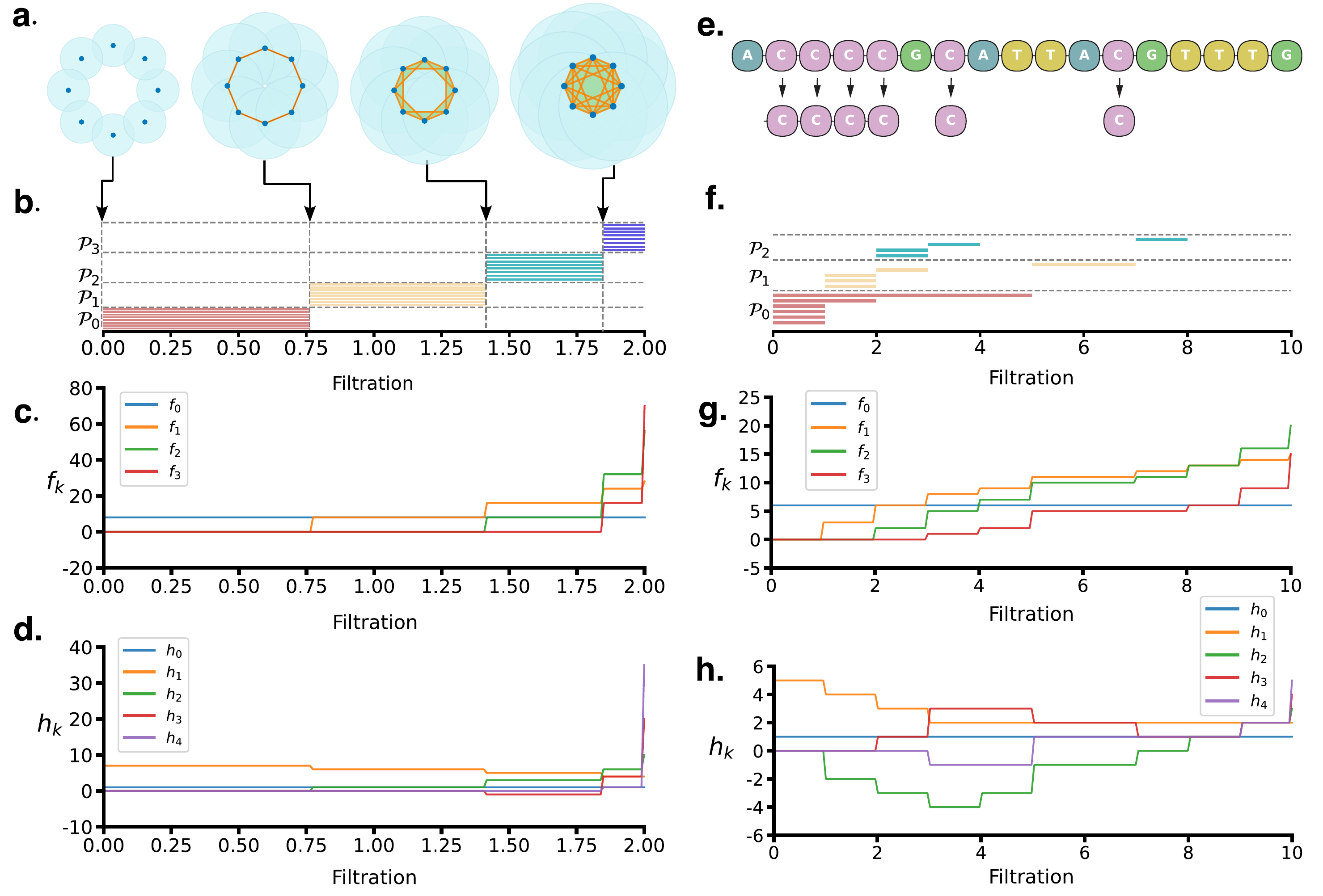}
	\end{subfigure}
\caption{\textbf{Illustrative example of persistent homology
		and persistent Stanley--Reisner invariants.}
	\textbf{a,} A filtration of a simplicial complex arising from an octagon.
	\textbf{b,} Persistent facet ideal barcodes derived from the same
	filtration, encoding combinatorial face-level activity rather than
	homology.
	\textbf{c--d,} Persistent \(f\)-vector and \(h\)-vector curves,
	respectively, derived from the same filtration.
	\textbf{e,} The N1-U.S-P primer sequence, with the positions of the
	nucleotide C marked.
	\textbf{f,} Persistent facet ideal barcodes of these positions,
	reflecting the activity of the persistent facet ideals under the induced
	filtration.
	\textbf{g--h,} Persistent \(f\)-vector and \(h\)-vector curves,
	respectively, derived from the filtration of the sequence.
	The symbol \(r\) denotes the filtration radius; \(\mathcal{P}_i\) denotes
	the collection of persistent \(i\)-dimensional facet ideals; and \(f_k\)
	and \(h_k\) denote the \(k\)th components of the persistent \(f\)- and
	\(h\)-vectors, respectively. In \textbf{a}, colours distinguish vertices,
	edges, and higher-dimensional simplices. In \textbf{b} and \textbf{f},
	barcode colours distinguish \(\mathcal{P}_0,\mathcal{P}_1,
	\mathcal{P}_2,\) and \(\mathcal{P}_3\), as indicated by the in-panel key.
	In \textbf{c--d} and \textbf{g--h}, curve colours distinguish the indexed
	\(f_k\)- and \(h_k\)-components indicated by the in-panel labels.
	Nucleotide colours in \textbf{e} denote A (teal), C (purple), G (green),
	and T (yellow). Source data are provided as a Source Data file.}
\label{fig:psrt_ex}
\end{figure}

Persistent Stanley-Reisner theory is a novel framework for algebraic data analysis, leveraging tools from commutative algebra. Unlike traditional topological data analysis, which emphasizes geometric and topological features, such as loops and voids, through persistent homology~\cite{su2025topological},
persistent Stanley-Reisner theory focuses on the algebraic and combinatorial structure of simplicial complexes, using invariants derived from commutative algebra\cite{suwayyid2025persistent}. A filtration process is then applied to these complexes to track the evolution and persistence of such features across multiple spatial or geometric scales. This approach introduces algebraic invariants such as persistent $h$-vectors, $f$-vectors, graded Betti numbers, and facet ideals, thus providing a new algebraic perspective within the broader framework of algebraic data analysis.

\subsubsection*{Persistent Stanley-Reisner structures over a filtration}

Let \(k\) be a field, and let \(\Delta\) be a simplicial complex on the finite vertex set \(V = \{x_1, \dots, x_n\}\). Suppose \(f: \Delta \to \mathbb{R}\) is a monotone function, i.e., \(f(\tau) \le f(\sigma)\) whenever \(\tau \subseteq \sigma\), which induces an increasing filtration
\[
\widetilde{f} := \left( \Delta^t \right)_{t \in \mathbb{R}}, \quad \text{where} \quad \Delta^t := \left\{ \sigma \in \Delta \,\middle|\, f(\sigma) \le t \right\}.
\]

Let \(S = k[x_1, \dots, x_n]\) be the standard graded polynomial ring over \(k\), and for each \(t \in \mathbb{R}\), define the Stanley-Reisner ideal of \(\Delta^t\) as
\[
I^t := \left\langle x_{i_1} \cdots x_{i_r} \,\middle|\, \{x_{i_1}, \dots, x_{i_r}\} \notin \Delta^t \right\rangle \subseteq S,
\]
with corresponding Stanley-Reisner ring
\[
k[\Delta^t] := S / I^t.
\]

Since the filtration is increasing, the subcomplexes satisfy \(\Delta^s \subseteq \Delta^t\) for \(s \le t\), which implies a descending chain of monomial ideals:
\[
I^s \supseteq I^t \quad \text{for all } s \le t.
\]

Each ideal \(I^t\) admits a canonical primary decomposition indexed by the facets of \(\Delta^t\):
\begin{equation}
	I^t = \bigcap_{\sigma \in \mathcal{F}(\Delta^t)} P_\sigma, \quad \text{where} \quad P_\sigma := (x_i \mid x_i \notin \sigma).
\end{equation}
We refer to the collection \(\mathcal{P}^t := \{ P_\sigma \mid \sigma \in \mathcal{F}(\Delta^t) \}\) as the facet ideals of \(\Delta^t\).

To capture the dimension-wise structure, we stratify by face dimension: for each \(i \ge 0\), where we define
\begin{equation}
	\mathcal{P}_i^t := \left\{ P_\sigma \in \mathcal{P}^t \,\middle|\, \dim(\sigma) = i \right\},
\end{equation}
so that
\[
\mathcal{P}^t = \bigsqcup_{i=0}^{\dim(\Delta^t)} \mathcal{P}_i^t.
\]

We define persistence algebraically as follows: a facet ideal \(P_\sigma \in \mathcal{P}_i^t\) is said to persist to level \(t' > t\) if \(P_\sigma \in \mathcal{P}_i^{t'}\). The set of such persistent \(i\)-dimensional primes is
\begin{equation}
	\mathcal{P}_i^{t, t'} := \mathcal{P}_i^t \cap \mathcal{P}_i^{t'}.
\end{equation}
The corresponding facet persistent number is given by
\begin{equation}
	\mathcal{F}_i^{t,t'} := \left| \mathcal{P}_i^{t, t'} \right|,
\end{equation}
which records the number of \(i\)-dimensional prime components common to \(\Delta^t\) and \(\Delta^{t'}\).

The collection \(\{\mathcal{F}^{t,t'}_i\}_{i,t,t'}\) serves as a combinatorial invariant encoding the persistence of prime facets in the Stanley-Reisner filtration, providing an algebraic analogue of topological barcodes in persistent homology.

\subsubsection{Persistent graded Betti numbers of Stanley-Reisner rings}

Let \(k\) be a field and \(S = k[x_1, \dots, x_n]\) the standard graded polynomial ring. For each filtration level \(t \in \mathbb{R}\), the Stanley-Reisner ring \(k[\Delta^t] := S / I^t\) inherits a natural \(\mathbb{Z}\)-graded \(S\)-module structure and admits a minimal graded free resolution:
\begin{equation}\label{eq:min-free-res}
	\cdots \longrightarrow
	\bigoplus_{j} S(-j)^{\beta_{i,j}(k[\Delta^t])}
	\longrightarrow \cdots \longrightarrow
	k[\Delta^t] \longrightarrow 0,
\end{equation}
where \(\beta_{i,j}(k[\Delta^t]) := \dim_k \operatorname{Tor}^S_i(k[\Delta^t], k)_j\) are the {graded Betti numbers}.

Hochster’s formula relates these graded Betti numbers to the topological Betti numbers of the induced subcomplexes:
\begin{equation}\label{eq:Hochster-general}
	\beta_{i,j+i}(k[\Delta^t])
	=
	\sum_{\substack{W \subseteq V \\ |W| = j+i}}
	\operatorname{rank} \widetilde{H}_{j-1}(\Delta_W^t; k),
\end{equation}
where \(\widetilde{H}_{j-1}(\Delta_W^t; k)\) denotes the \((j-1)\)-st reduced simplicial homology group over \(k\), and
\(\Delta_W^t := \{ \sigma \in \Delta^t \mid \sigma \subseteq W \}\) is the subcomplex induced on the vertex set \(W \subseteq V\) \cite{bruns1998cohen}.

In particular, Hochster’s formula can be reformulated in terms of the (non-reduced) Betti numbers of induced subcomplexes. For each integer \(i \ge 0\), the following identities hold:
\begin{align}
	\beta_{i,i+1}(k[\Delta^t]) &= \sum_{\substack{W \subseteq V \\ |W| = i+1}} \left( \beta_0(\Delta_W^t) - 1 \right), \label{eq:Hochster-j1} \\
	\beta_{i,i+j}(k[\Delta^t]) &= \sum_{\substack{W \subseteq V \\ |W| = i+j}} \beta_{j-1}(\Delta_W^t), \quad \text{for all } j \ge 2, \label{eq:Hochster-j2}
\end{align}
where \(\Delta_W^t\) denotes the subcomplex of \(\Delta^t\) induced on the vertex subset \(W \subseteq V\), and \(\beta_r(\Delta_W^t)\) denotes the \(r\)-th Betti number of \(\Delta_W^t\) with coefficients in \(k\).

To refine this in a persistent setting, for \(t \le t'\), we define the {persistent graded Betti number}
\begin{equation}\label{eq:persistent-Betti-short}
	\beta_{i, i+j}^{t,t'}(k[\Delta])
	:= \sum_{\substack{W \subseteq V \\ |W| = i + j}}
	\operatorname{rank} \left( \iota_{j-1}^{t,t'} : \widetilde{H}_{j-1}(\Delta_W^t)
	\to
	\widetilde{H}_{j-1}(\Delta_W^{t'}) \right),
\end{equation}
where \(\iota_{j-1}^{t,t'}\) is the homomorphism on reduced homology induced by inclusion. This provides a multigraded algebraic refinement of classical persistent Betti numbers, encoding both topological persistence and the combinatorial properties of the evolving homology classes.

In the special case where \(|W| = |V|\), the persistent graded Betti number reduces to
\[
\beta_{i, |V|}^{t,t'} = \beta_{|V| - i - 1}^{t,t'},
\]
recovering the classical persistent Betti number of homological degree \(|V| - i - 1\). More generally, the family \(\{\beta_{i, i+j}^{t,t'}\}_{i,j}\) encodes a richer multiscale invariant that interpolates between algebraic and topological persistence.

\subsubsection{Persistent \texorpdfstring{\(f\)}{f}- and \texorpdfstring{\(h\)}{h}-vectors}

Let \((\Delta^t)_{t \in \mathbb{R}}\) be a filtration of a finite \((d-1)\)-dimensional simplicial complex \(\Delta\), induced by some face function \(f: \Delta \to \mathbb{R}\). For each fixed level \(t\), the complex \(\Delta^t\) consists of those faces \(\sigma \in \Delta\) with \(f(\sigma) \le t\). At each level \(t\), one may associate the classical combinatorial invariants of face counts and their derived quantities.

The {\(f\)-vector} of \(\Delta^t\) is defined as
\[
f(\Delta^t) = \left(f_{-1}^t, f_0^t, f_1^t, \dots, f_{d-1}^t\right),
\]
where \(f_i^t\) denotes the number of \(i\)-dimensional faces in \(\Delta^t\), \(f_{-1}^t = 1\) by convention, and \( d = d(t) = \dim(\Delta^t) + 1 \). The associated \( h \)-vector is defined as \( h(\Delta^t) = (h_0^t, \dots, h_{d}^t) \), where
\begin{equation}\label{eq:h-vector}
	h_m^t = \sum_{j=0}^{m}
	\binom{d(t) - j}{m - j}
	(-1)^{m-j} f_{j-1}^t,
	\quad \text{for } m = 0, \dots, d(t),
\end{equation}
and \( h_m^t = 0 \) for all \( m > d(t) \).

This transformation is invertible, with the inverse relation given by
\begin{equation}\label{eq:f-vector}
	f_{m-1}^t = \sum_{i=0}^m
	\binom{d(t) - i}{m - i} h_i^t,
	\quad \text{for } m = 0, \dots, d(t),
\end{equation}
where \( d(t) = \dim(\Delta^t) + 1 \).

To extend these invariants to the persistent setting, one replaces the classical Betti numbers with the persistent graded Betti numbers \(\beta_{i,j}^{t,t'}\) defined over filtration levels \(t \le t'\) \cite{suwayyid2025persistent}. This enables a multiscale, combinatorial interpretation of how face structures persist across different filtration levels.

Let \((\Delta^t)_{t \in \mathbb{R}}\) be a filtration of a simplicial complex \(\Delta\). The {persistent \(h\)-vector} between levels \(t \le t'\) is defined as
\begin{equation}\label{eq:persistence_h}
	h_m^{t,t'} := \sum_{j=0}^{m}
	\binom{n - d(t') + m - j - 1}{m - j}
	\left(
	\sum_{i=0}^{j} (-1)^i \beta_{i,j}^{t,t'}
	\right),
	\quad \text{for } m = 0, \dots, d(t'),
\end{equation}
where \(\beta_{i,j}^{t,t'}\) denotes the persistent graded Betti numbers of \(k[\Delta]\) over \([t,t']\), and let \(d(t') = \dim(\Delta^{t'}) + 1\).

The corresponding {persistent \(f\)-vector} is then defined via the inverse transformation:
\begin{equation}\label{eq:persistence_f}
	f_{m-1}^{t,t'} := \sum_{i=0}^m
	\binom{d(t') - i}{m - i} \, h_i^{t,t'},
	\quad \text{for } m = 0, \dots, d(t').
\end{equation}

These persistent vectors capture how the combinatorial structure of \(\Delta\) evolves through the filtration, blending face enumeration with homological persistence. In contrast to the classical static \(f\)- and \(h\)-vectors, their persistent counterparts reflect the dynamic appearance and disappearance of faces and their relations across multiple scales, providing richer algebraic-combinatorial invariants for analysis.

We now consider the following simplifications, which will play a central role in the \( k \)-mer algebraic representation framework. These observations refine the relationship between persistent \(h\)-vectors and persistent graded Betti numbers, serving to streamline computations in applications involving Vietoris-Rips complexes derived from sequence data.

\medskip

Let \( \beta_{i,j}^{t,t'} \) denote the persistent graded Betti numbers of the Stanley-Reisner ring \( k[\Delta] \) over the filtration interval \([t, t']\), as defined in equation~\eqref{eq:persistent-Betti-short}. To streamline notation, we set
\begin{equation}
	B_j := \sum_{i=0}^j (-1)^i \beta^{t,t'}_{i,j}.
\end{equation}
Alongside this, one introduces the coefficients
\begin{equation}
	\alpha_j^{(m)} := \binom{n - d(t') + m - j - 1}{m - j},
\end{equation}
which appear in the linear transformation relating the \( h \)-vector of a simplicial complex to its graded Betti numbers.

\medskip

It follows that the persistent \( h \)-vector component \( h_m^{t,t'} \) satisfies the identity
\begin{equation}\label{eq:discussion-h-decomp}
	h_m^{t,t'} = \sum_{j=0}^m \alpha_j^{(m)} B_j, \qquad \text{for each } m \in \mathbb{N}.
\end{equation}

\medskip

Additional structural identities among the persistent Betti numbers further simplify this formula. In particular, it is known that
\[
\beta^{t,t'}_{0,0} = 1, \qquad
\beta^{t,t'}_{i,i} = 0 \quad \text{for all } i \ge 1, \qquad
\beta^{t,t'}_{0,j} = 0 \quad \text{for all } j \ge 1, \qquad
\beta^{t,t'}_{i,j} = 0 \quad \text{for all } i > j.
\]
Consequently, one obtains
\[
B_0 = \beta^{t,t'}_{0,0} = 1,
\]
and for each \( j \ge 1 \), the alternating sum simplifies to
\[
B_j = \sum_{i=1}^{j-1} (-1)^i \beta^{t,t'}_{i,j}.
\]

\subsection{\(k\)-mer algebraic representations of sequences}
\label{subsec:kmer_representation}

In this section, we review the \(k\)-mer representation framework introduced by Hozumi et al.~\cite{hozumi2024revealing}, which provides a foundational method for embedding sequences as collections of integer sequences in a geometric space.
Let \( \mathcal{A} \) be a finite alphabet and let \( k > 0 \) be an integer. A {\(k\)-mer} over \( \mathcal{A} \) is a word \( \mathbf{x} = x_1x_2\cdots x_k \in \mathcal{A}^k \). Given a fixed \( k \)-mer \( \mathbf{x} \in \mathcal{A}^k \), we define the {\(k\)-mer indicator function} \( \delta_{\mathbf{x}}: \mathcal{A}^k \to \{0, 1\} \) by
\begin{equation}
	\delta_{\mathbf{x}}(\mathbf{y}) =
	\begin{cases}
		1, & \text{if } \mathbf{y} = \mathbf{x}, \\
		0, & \text{otherwise}.
	\end{cases}
\end{equation}

\medskip

Given a sequence \( S = s_1 s_2 \cdots s_N \in \mathcal{A}^N \), we define the set of positions at which the \( k \)-mer \( \mathbf{x} \) occurs in \( S \) as
\begin{equation}
	S^{\mathbf{x}} = \left\{ i \in [1, N - k + 1] \,\middle|\, \delta_{\mathbf{x}}(s_i s_{i+1} \cdots s_{i+k-1}) = 1 \right\}.
\end{equation}

The corresponding pairwise distance matrix \( \mathbf{D}^{\mathbf{x}} = ( d^{\mathbf{x}}_{ij} | \; i,j \in S^{\mathbf{x}}) \) is defined by
\begin{equation}
	d^{\mathbf{x}}_{ij} = |i - j|, \qquad \text{for all } i, j \in S^{\mathbf{x}}.
\end{equation}

\medskip

These distance matrices serve as the input for persistent Stanley-Reisner computations. Specifically, for each \( k \)-mer \( \mathbf{x} \in \mathcal{A}^k \), the corresponding sequence of integers \( S^{\mathbf{x}} \subset \mathbb{R} \) gives rise to a family of Stanley-Reisner algebraic feature vectors computed over a filtration interval \( [r_0, r_1] \). For filtration values \( r, r' \in [r_0, r_1] \) with \( r \leq r' \), we define
\[
\mathbf{v}^{r,r'}_{\mathbf{x}} = \left( v^{r, r'}_i(\mathbf{x}) \right)_{i \in \mathbb{N}},
\]
where \( v^{r, r'}_i(\mathbf{x}) \) denotes a persistent invariant of dimension \( i \), such as the \( f \)-vector, \( h \)-vector, or facet number, associated with the Vietoris-Rips complex built from \( S^{\mathbf{x}} \).

\medskip

To simplify notation, we restrict to the diagonal case \( r = r' \), and denote the resulting feature vector by
\[
\mathbf{v}_{\mathbf{x}} := \left( v_i(\mathbf{x}) \right)_{i \in \mathbb{N}}.
\]
For a fixed integer \( k > 0 \), the full representation of the sequence \( S \in \mathcal{A}^N \) is given by the concatenation of these vectors over all \( k \)-mers:
\begin{align*}
	\mathbf{v}^k_S &:= \left( \mathbf{v}_{\mathbf{x}} \;\middle|\; \mathbf{x} \in \mathcal{A}^k \right) \\
	&= \left( v_i(\mathbf{x}) \;\middle|\; \mathbf{x} \in \mathcal{A}^k \right)_{i \in \mathbb{N}},
\end{align*}
which we refer to as the \(k\)-mer algebraic representation of \( S \) at level \( k \). This construction yields a feature vector indexed jointly by algebraic dimension \( i \) and \( k \)-mer \( \mathbf{x} \in \mathcal{A}^k \).

\subsection{Algebraic genetic distances}
\label{sec:algebraic_metric}

To compare two sequences \( S_1 \in \mathcal{A}^{N_1} \) and \( S_2 \in \mathcal{A}^{N_2} \), we define a family of weighted Euclidean metrics that aggregate Stanley-Reisner algebraic information across both the algebraic dimensions and \( k \)-mer lengths. Let \( a_{k,i} \geq 0 \) denote a non-negative weight assigned to homological dimension \( i \) at scale \( k \). The dimension- and scale-weighted algebraic distance is defined by
\begin{equation}
	d_v(S_1, S_2) := \sum_{k=1}^K \sum_{i=0}^{D_k} a_{k,i} \cdot \left\| \mathbf{v}^k_{S_1, i} - \mathbf{v}^k_{S_2, i} \right\|_2,
\end{equation}
where \( \mathbf{v}^k_{S, i} := \left( v_i(\mathbf{x}) \;\middle|\; \mathbf{x} \in \mathcal{A}^k \right) \) is the vector of dimension-\( i \) features computed over all \( k \)-mers in \( S \), and \( D_k \) is the maximum dimension considered for \( k \).

There are various strategies for selecting the weights \( a_{k,i} \). A common choice in the literature is to assign exponentially decaying weights across scales, for example \( a_{k,i} = 2^{-(k-1)} \), which emphasizes shorter \( k \)-mers while retaining contributions from larger scales. To additionally incorporate the homological dimension, one may use
\begin{equation}
	a_{k,i} = 2^{-(i \cdot K + k - 1)},
\end{equation}
which imposes exponential decay both in the scale \( k \) and in the dimension \( i \).

In contrast to these monotone decay schemes, we adopt a data-driven weighting centered around a preferred scale \( k^\ast \), selected as described in \ref{kchoice}. Specifically, we assign the largest weight to \( k^\ast \) and impose exponential decay as \( k \) moves away from this scale in either direction. This leads to the choice
\begin{equation}
	a_{k,i} = \frac{2^{-|k - k^\ast|} \cdot 2^{-i}}{\displaystyle \sum_{\ell=1}^K \sum_{j=0}^{D_\ell} 2^{-|\ell - k^\ast|} \cdot 2^{-j}},
\end{equation}
where the normalization ensures that the weights sum to one.

This construction preserves the multiscale nature of the representation while concentrating the contribution around the empirically most informative \( k \)-mer length. In particular, scales near \( k^\ast \) receive the highest weight, while both smaller and larger \( k \)-values are exponentially downweighted, thereby balancing resolution and stability in the resulting metric.

\medskip

Within the CAKR framework, three distinct types of persistent algebraic features are employed to define pairwise distances between sequences. Specifically, the distance \( d_f(S_1, S_2) \) is derived from the \( f \)-vector curves, where the feature vector \( v \) is computed from these curves; the distance \( d_h(S_1, S_2) \) is defined analogously using \( h \)-vector curves; and the distance \( d_{\mathcal{F}}(S_1, S_2) \) is based on the facet count vectors associated with the underlying filtration.

\medskip

The final composite distance, integrating these three feature types, is given by
\begin{equation}
	d(S_1, S_2) := d_f(S_1, S_2) + d_h(S_1, S_2) + d_{\mathcal{F}}(S_1, S_2),
\end{equation}
which captures a broad range of persistent characteristics across multiple features and filtration levels. This composite metric constitutes the core of the CAKR approach to alignment-free sequence comparison.

\medskip

In the applications considered in this work, we restrict our attention to a single feature representation, namely, the facet vector curves, and employ a fixed window length \( k \) for \( k \)-mers.

\subsubsection{Choice of $k$-mer size and feature representation} \label{kchoice}

Rather than fixing a single $k$-mer size a priori, we adopt a data-driven strategy to determine an optimal scale $k^\ast$ for each dataset. For each candidate $k$, we construct the feature matrix $X^{(k)} \in \mathbb{R}^{n \times 4^k}$ and evaluate three complementary criteria capturing structural, informational, and stability properties of the representation.

Specifically, we consider: (i) a coverage-based criterion that identifies a structural scale via the elbow of the coverage curve, (ii) an information-sparsity score that balances entropy with feature sparsity, and (iii) a stability-adjusted entropy that penalizes contributions from rare (singleton) features. These criteria capture distinct aspects of the trade-off between expressiveness and robustness across $k$.

Across datasets, these criteria do not yield a consistent ordering. In particular, the information-sparsity score tends to favor smaller values of $k$, while the stability-adjusted entropy closely follows the structural elbow. To reconcile these differences, we introduce a scoring framework that evaluates each criterion relative to the coverage-based structural reference, rewarding earlier selections while penalizing large deviations. This induces a ranking of the criteria, from which we derive weights reflecting their relative importance.

The global scale $k^\ast$ is then selected via a weighted consensus rule combining the three criteria. Full definitions of the criteria, the scoring function, and the weighting procedure are provided in the Supplementary Information (Supplementary Section~9). Representative dataset-specific scale-selection curves for the SARS-CoV-2, mammalian mitochondrial, and rhinovirus datasets are shown in Supplementary Figs.~16-18.

In addition to the persistent facet count vector, we also tested the $f$- and $h$-vector representations. Their performances were found to be comparable; however, the facet representation yielded more stable and compact features, whereas the $f$- and $h$-vectors tended to produce larger components.

For the Vietoris-Rips filtration, we employed thresholds at $0,\,4^k,\,2\times4^k,\ldots$, ensuring adequate coverage of $k$-mer co-occurrence patterns. A relatively small number of filtration steps was chosen to maintain computational efficiency and limit memory usage, while still achieving stable performance across datasets.

Further empirical analyses of the effect of \(k\)-mer size across methods, including direct 1-NN accuracy curves and cumulative-best phylogenetic performance summaries over \(k\), are provided in Supplementary Section~6 and Supplementary Figs.~13 and~14.

\subsection{Computational simplifications of the persistent $h$-vectors and $f$-vectors}
Vietoris-Rips simplicial complexes arising from the \(k\)-mer algebraic representations possess a structural property that significantly simplifies their algebraic analysis. Specifically, many of the persistent graded Betti numbers vanish in higher homological degrees, which reduces the complexity of computations involving persistent \(h\)-vectors.

Let $(X,d)$ be a finite metric space (or a finite set of points in a metric space).
For a scale parameter $\varepsilon \ge 0$, the Vietoris-Rips complex
$\mathrm{VR}_\varepsilon(X)$ is the abstract simplicial complex whose vertex set is $X$,
and where a finite subset $\sigma=\{x_0,\dots,x_p\}\subseteq X$ spans a $p$-simplex
whenever all pairwise distances are bounded by $\varepsilon$, i.e.,
\[
d(x_i,x_j)\le \varepsilon \quad \text{for all } 0\le i<j\le p.
\]
Equivalently, $\mathrm{VR}_\varepsilon(X)$ is the clique (flag) complex of the
proximity graph on $X$ with an edge between $x$ and $y$ whenever $d(x,y)\le \varepsilon$.
We refer to \cite{edelsbrunner2010computational} for background and standard properties.

As $\varepsilon$ increases, these complexes form a nested family
\[
\mathrm{VR}_{\varepsilon_1}(X)\subseteq \mathrm{VR}_{\varepsilon_2}(X)
\quad \text{whenever } \varepsilon_1 \le \varepsilon_2,
\]
which is called the Vietoris-Rips filtration. Persistent homology is then computed
from the induced maps on homology along this filtration, tracking the birth and death of
topological features across scales.

The Vietoris-Rips construction is widely used in topological data analysis because it is
entirely determined by pairwise distances and is therefore straightforward to build from a
distance matrix; moreover, it serves as a computationally efficient proxy for other
distance-based complexes (e.g., \v{C}ech-type constructions) while retaining meaningful
multi-scale topological information \cite{edelsbrunner2010computational}.

The following proposition formalizes this observation and highlights its relevance to \(k\)-mer algebraic representations:

\begin{proposition}\label{prop:VR-persistent-betti}
	Let \( \Delta = \Delta^t \) denote the Vietoris-Rips complex at scale \( t \) associated with a sequence \( X \subset \mathbb{R} \). Then for every subset \( W \subseteq V \) and every \( j \geq 2 \), the persistent Betti numbers satisfy
	\[
	\beta_{j-1}^{t,t'}(\Delta_W) = 0.
	\]
	As a consequence,
	\[
	\beta_{i,i+j}^{t,t'} = 0 \quad \text{for all } i \geq 1,\ j \geq 2,
	\]
	and the only potentially nonzero contributions occur in degree shifts of one, namely
	\[
	\beta_{i,i+1}^{t,t'} = \sum_{\substack{W \subseteq V \\ |W| = i+1}} \left( \beta_0^{t,t'}(\Delta_W) - 1 \right).
	\]
	Therefore, the alternating sum of persistent Betti numbers at total degree \( j \) simplifies to
	\[
	B_j := \sum_{i=0}^j (-1)^i \beta_{i,j}^{t,t'} = (-1)^{j-1} \beta_{j-1,j}^{t,t'} \quad \text{for all } j \geq 1.
	\]
	In particular, the persistent \( h \)-vector expression in equation~\eqref{eq:discussion-h-decomp} becomes
	\[
	h_m^{t,t'} = \alpha_0^{(m)} + \sum_{j=1}^m \alpha_j^{(m)} (-1)^{j-1} \beta_{j-1,j}^{t,t'}, \quad \text{with } h_0^{t,t'} = 1.
	\]
\end{proposition}

\medskip

To establish Proposition~\ref{prop:VR-persistent-betti}, we prove a more general structural result concerning Vietoris-Rips complexes over sequences in \( \mathbb{R} \). Let \( X \subseteq \mathbb{R} \) be a finite sequence, and let \( \mathrm{VR}_\epsilon(X) \) denote the Vietoris-Rips complex at scale \( \epsilon \).

Each facet \( F \subseteq \mathrm{VR}_\epsilon(X) \) admits a unique minimal element \( x = \inf(F) \in X \). Moreover, if another facet \( G \subseteq \mathrm{VR}_\epsilon(X) \) satisfies \( \inf(G) = x \), then necessarily \( G = F \). That is, the minimal element uniquely determines the facet. Consequently, the assignment \( x \mapsto F_x \), where \( F_x \) denotes the unique facet with minimal element \( x \), satisfies
\[
F_x = F_y \quad \Longleftrightarrow \quad x = y.
\]
In particular, the collection of facets is in bijective correspondence with the set of the minimal elements of the facets, and hence can be linearly ordered by their infima:
\[
F_x \leq F_y \quad \Longleftrightarrow \quad x \leq y.
\]

\begin{proposition}
	Let \(X \subseteq \mathbb{R}\) be a finite sequence. Then for all \(q \geq 1\),
	\[
	H_q(\mathrm{VR}_\epsilon(X)) = 0.
	\]
\end{proposition}

\begin{proof}
	Let \(\mathrm{VR}_\epsilon(X) = \bigcup_{i=1}^n F_{x_i}\), where the facets \(F_{x_i}\) are ordered such that \(x_1 < x_2 < \cdots < x_n\).

	We proceed by induction on \(n\).

	{Base case:} When \(n = 1\), \(\mathrm{VR}_\epsilon(X)\) is a single simplex, which is contractible. Therefore, \(H_q = 0\) for all \(q \geq 1\).

	{Inductive step:} Assume the result holds for \(n-1\) facets, where \(n > 1\). Let:
	\[
	K_1 = \bigcup_{i=1}^{n-1} F_{x_i}, \quad K_2 = F_{x_n}, \quad K = K_1 \cup K_2.
	\]
	Note that \(K_1\), \(K_2\), and \(K_1 \cap K_2\) are all simplicial complexes. Notice that the vertices of \(K_1 \cap K_2\) lie within the interval \([x_{n}, x_{n-1} + \epsilon]\), whose length is at most \(\epsilon\), the intersection \(K_1 \cap K_2\) is either empty or a simplex.

	By the Mayer-Vietoris sequence, we obtain the long exact sequence in homology:
	\[
	\cdots \to H_q(K_1 \cap K_2) \to H_q(K_1) \oplus H_q(K_2) \to H_q(K) \to H_{q-1}(K_1 \cap K_2) \to \cdots
	\]

	By the inductive hypothesis, \(H_q(K_1) = 0\) for all \(q \geq 1\), and since \(K_2\) is a simplex, \(H_q(K_2) = 0\) as well. Furthermore, \(K_1 \cap K_2\) is either a simplex or empty, so:
	\[
	H_q(K_1 \cap K_2) = 0 \quad \text{for all } q \geq 1.
	\]

	Thus, the exact sequence reduces to:
	\[
	0 \to H_q(K) \to 0 \quad \Rightarrow \quad H_q(K) = 0 \quad \text{for all } q \geq 2.
	\]

	To analyze \(H_1(K)\), we consider:
	\[
	0 \to H_1(K) \to H_0(K_1 \cap K_2) \to H_0(K_1) \oplus H_0(K_2) \to H_0(K) \to 0.
	\]
	If \( K_1 \cap K_2 = \emptyset \), then \( H_1(K) = 0 \) since \( K = K_1 \sqcup K_2 \) is a disjoint union of two contractible subcomplexes. Thus, the result holds in this case.

	Suppose instead that \( K_1 \cap K_2 \neq \emptyset \). Then \( K_1 \cap K_2 \) is a simplex and hence contractible, in particular connected. Since \( K_2 \) is also a simplex, it is connected and contractible. Moreover, the inclusion of \( K_2 \) into \( K = K_1 \cup K_2 \) does not change the number of connected components, so \( K \) and \( K_1 \) have the same number of components. Therefore, the canonical map
	\[
	H_0(K_1) \oplus H_0(K_2) \longrightarrow H_0(K)
	\]
	has kernel of dimension one, namely \( \dim H_0(K_2) = 1 \). Since \( K_1 \cap K_2 \) is connected, the induced map
	\[
	H_0(K_1 \cap K_2) \longrightarrow H_0(K_1) \oplus H_0(K_2)
	\]
	is injective. Consequently, in the Mayer–Vietoris sequence, the connecting homomorphism
	\[
	H_1(K) \longrightarrow H_0(K_1 \cap K_2)
	\]
	must be the zero map. It follows that \( H_1(K) = 0 \) in this case as well. By induction on the number of simplices, we conclude that
	\[
	H_q(\mathrm{VR}_\epsilon(X)) = 0 \quad \text{for all } q \geq 1.
	\]
\end{proof}

This structural property does not extend to higher-dimensional ambient spaces \( X \subset \mathbb{R}^d \) for \( d \geq 2 \); for instance, consider the Vietoris-Rips complex formed from the vertices of a regular hexagon in \( \mathbb{R}^2 \). Therefore, Proposition~\ref{prop:VR-persistent-betti} is a consequence of the special linear ordering available in one-dimensional point clouds. This leads to a significant simplification in the computation of persistent Betti numbers arising from \(k\)-mer algebraic representations.

Given a simplicial complex \( \Delta \), its 1-skeleton induces an undirected graph \( G(\Delta) \) with vertex set \( V \) and edge set
\[
E(\Delta)
\;:=\;
\Bigl\{
\{u,v\}\subseteq V \ \Bigm|\
\{u,v\}\in\Delta,\ \nexists\, w\in V \text{ with } u< w< v
\Bigr\}.
\]
Equivalently,
\[
E(\Delta)
\;=\;
\Bigl\{
\{u,v\}\in\Delta \ \Bigm|\ (u,v)\cap V=\varnothing
\Bigr\}.
\]
When \( \Delta \) is a Vietoris-Rips complex built on \(k\)-mer representations in \( \mathbb{R} \), the persistent Betti numbers of \( \Delta \) are entirely determined by the topology of the associated graph \( G(\Delta) \). This follows directly from Proposition~\ref{prop:VR-persistent-betti}. We formalize this relationship in the following theorem:

\begin{theorem}
	Let \( \Delta \) be a Vietoris-Rips simplicial complex on a finite sequence \( X \subset \mathbb{R} \), and let \( G(\Delta) = (V, E(\Delta)) \) denote its 1-skeleton. Then the persistent Betti numbers of \( \Delta \) satisfy:
	\[
	\beta_{i,i+1}(G(\Delta)) = \beta_{i,i+1}(\Delta),
	\quad \text{and} \quad
	\beta_{i,i+j}(G(\Delta)) = \beta_{i,i+j}(\Delta) = 0 \quad \text{for all } j \geq 2.
	\]
\end{theorem}

\subsection{Purity metric for assessing the performance of phylogenetic tree reconstruction methods}

We introduce a purity metric for assessing monophyly in phylogenetic trees.
Let \( S \) be a finite set of size \( n = |S| \), and let \( \mathcal{P} = \{S_1, S_2, \dots, S_k\} \) be a partition of \( S \) into disjoint subsets such that \( \bigcup_{i=1}^k S_i = S \). We define the {purity} of the partition \( \mathcal{P} \) as
\begin{equation}
	\mathrm{purity}(\mathcal{P}) = \sum_{i=1}^{k} \left( \frac{|S_i|}{n} \right)^2.
\end{equation}
This quantity reflects the degree to which elements are concentrated within the subsets of the partition. A higher purity indicates that the majority of elements reside in a small number of large subsets, while a lower purity corresponds to a more evenly distributed partition.

To illustrate, consider several representative scenarios. If the partition is perfect in the sense that all elements are grouped into a single subset, i.e., \( k = 1 \), then
\[
\mathrm{purity}(\mathcal{P}) = \left( \frac{n}{n} \right)^2 = 1,
\]
which is the maximal possible value. If the partition consists of \( n \) singleton subsets (i.e., \( |S_i| = 1 \) for all \( i \)), then
\[
\mathrm{purity}(\mathcal{P}) = \sum_{i=1}^n \left( \frac{1}{n} \right)^2 = \frac{1}{n},
\]
which is minimal. If the partition consists of two equal-sized subsets, each of size \( n/2 \), then
\[
\mathrm{purity}(\mathcal{P}) = 2 \left( \frac{1}{2} \right)^2 = \frac{1}{2}.
\]
Finally, if one subset dominates the partition, for example, with \( |S_1| = n - 1 \) and \( |S_2| = 1 \), then
\[
\mathrm{purity}(\mathcal{P}) = \left( \frac{n - 1}{n} \right)^2 + \left( \frac{1}{n} \right)^2 = 1 - \frac{2(n - 1)}{n^2},
\]
which approaches 1 as \( n \to \infty \), but is strictly less than 1 for any finite \( n \).

Let \( S \) be a finite set of leaf nodes in a phylogenetic tree, and let each element of \( S \) be assigned a categorical label (e.g., species, clade, or functional class). For each label \( \ell \), let \( S^{(\ell)} \subseteq S \) denote the set of leaves with label \( \ell \), and let \( n^{(\ell)} = |S^{(\ell)}| \) be the number of such leaves.

To assess the purity of the tree with respect to label \( \ell \), we identify all maximal subtrees whose leaves are exclusively labeled \( \ell \). These subtrees define a partition \( \mathcal{P}_\ell = \{ S_1, S_2, \dots, S_k \} \) of \( S^{(\ell)} \), where each \( S_i \subseteq S^{(\ell)} \) is the set of leaves in a pure subtree.

The purity of the label \( \ell \) is then defined as:
\begin{equation}
	\mathrm{purity}(\mathcal{P}_\ell) = \sum_{i=1}^{k} \left( \frac{|S_i|}{n^{(\ell)}} \right)^2,
\end{equation}
where the numerator \( |S_i| \) denotes the size of a pure subtree and the denominator normalizes by the total number of leaves of label \( \ell \).

A purity of \( 1.0 \) indicates that all leaves of label \( \ell \) are perfectly clustered under a single subtree (i.e., monophyletic), while a lower purity reflects fragmentation of that label across multiple subtrees. Averaging the purity scores across all labels provides an overall measure of the taxonomic coherence of the tree:
\begin{equation}
	\label{eq:purity_metric}
	\mathrm{avg\_purity} = \frac{1}{|\mathcal{L}|} \sum_{\ell \in \mathcal{L}} \mathrm{purity}(\mathcal{P}_\ell),
\end{equation}
where \( \mathcal{L} \) is the set of all unique labels in the tree.

This approach is beneficial for evaluating the extent to which a phylogenetic tree respects known groupings, such as taxonomic families or functional clusters, without requiring an explicit reference.

\section*{Software and computational environment}

The CAKR analyses were implemented in Python using NumPy 1.26.4, scikit-learn 1.4.2, SciPy 1.13.1, GUDHI 3.10.1, Biopython 1.84, pandas 2.2.2, and ETE Toolkit 3.1.3. The exact source-code release used in this study is CAKR v1.0.0, archived on Zenodo at \url{https://doi.org/10.5281/zenodo.21426540}. Comparative analyses additionally used MAFFT, IQ-TREE 3, RAxML-NG, Mash, CVTree3, SEPP, and Interactive Tree Of Life (iTOL) v6. Phylogenetic trees were constructed using the UPGMA algorithm and, in the supplementary sensitivity analyses, the Neighbor-Joining algorithm.

\section*{Data Availability}

The genomic sequence data analyzed in this study were obtained from NCBI, including GenBank and the NCBI Virus resource, and from GISAID for the SARS-CoV-2 dataset. GenBank accession.version identifiers are provided in Supplementary Tables~8--13, and the GISAID accession identifiers for the 44 SARS-CoV-2 genomes are provided in Supplementary Table~7. The underlying GISAID sequences and metadata are not redistributed because of GISAID access restrictions and must be obtained directly from GISAID by registered users in accordance with its terms of use. The curated and processed datasets, excluding the underlying GISAID files, are available from Zenodo at \url{https://doi.org/10.5281/zenodo.18757928} \cite{suwayyid2026cakrdata}. Additional data are provided in the Supplementary Information and Source Data file. Further supporting materials not included in these files are available from the corresponding author upon request. Source data are provided with this paper.

\section*{Code Availability}

The source code used to implement the CAKR framework and perform the comparative analyses reported in this study is publicly available on GitHub at \url{https://github.com/FaisalSuwayyid/CAKL}. The version used in this study, CAKR v1.0.0, has been archived on Zenodo at \url{https://doi.org/10.5281/zenodo.21426540} \cite{suwayyid2026cakrcode}.

\bibliographystyle{unsrt}
\bibliography{references}

\section*{Acknowledgements}

We gratefully acknowledge all data contributors, i.e., the Authors and their Originating laboratories responsible for obtaining the specimens, and their Submitting laboratories for generating the genetic sequence and metadata and sharing via the GISAID Initiative, on which this research is based.

F.S. gratefully acknowledges the support of King Fahd University of Petroleum and Minerals.

\section*{Funding Statement}

This work was supported in part by NIH grants R01AI164266 and R35GM148196, the MSU Research Foundation, the University of Georgia, and the Georgia Research Alliance.

\section*{Author Contributions}

F.S. designed the method and study, wrote the code, performed the computational studies, wrote the first draft, and revised the manuscript. Y.H. designed the method and study, collected the data, performed the computational studies, and revised the manuscript. M.Z. designed the study, wrote the code, performed computational studies and revised the manuscript.
J.J.W. revised the manuscript. H.F. wrote code, prepared figures, and revised the manuscript. G.-W.W. designed the study, conceptualized and supervised the project, acquired funding, and revised the manuscript.

\section*{Competing Interests}
The authors declare no competing interests.

\section*{Figure Legends}

\clearpage
\begingroup
\setcounter{section}{0}
\setcounter{subsection}{0}
\setcounter{figure}{0}
\setcounter{table}{0}
\setcounter{equation}{0}
\renewcommand{\figurename}{Supplementary Fig.}
\renewcommand{\tablename}{Supplementary Table}
\renewcommand{\theHsection}{supp.\arabic{section}}
\renewcommand{\theHsubsection}{supp.\arabic{section}.\arabic{subsection}}
\renewcommand{\theHfigure}{supp.\arabic{figure}}
\renewcommand{\theHtable}{supp.\arabic{table}}
\renewcommand{\theHequation}{supp.\arabic{equation}}

\begin{center}
{\LARGE\bfseries Supplementary Information\\[0.35em]
for\\[0.35em]
CAKR: Commutative algebra $k$-mer representations for genomics\par}
\vspace{1.5em}
{\large Faisal Suwayyid, Yuta Hozumi, Mushal Zia, JunJie Wee,\\
Hongsong Feng, and Guo-Wei Wei\par}
\end{center}
\vspace{1em}
\section{Datasets}

To rigorously evaluate the proposed method CAKR, we assembled a suite of benchmark datasets spanning diverse applications in genomics and virology.

	\textbf{Genetic Variant Classification.}
	To evaluate the discriminatory capacity of CAKR in distinguishing between genetic variants, we employed the severe acute respiratory syndrome coronavirus 2 (SARS-CoV-2) dataset curated by Li~{et~al.}~\cite{Li2024GenomeGrassmann, hozumi2024revealing}, which comprises representative genome sequences from multiple SARS-CoV-2 lineages. The dataset consists of 44 complete genomes of SARS-CoV-2, sourced from GISAID. These genomes are classified according to their variant lineages, including Alpha, Beta, Gamma, Delta, Lambda, Mu, GH/490R, and Omicron. Phylogenetic branches and labels are annotated and color-coded to reflect these variant classifications.

	\textbf{Phylogenetic Tree Reconstruction.}
	Seven benchmark genome collections were employed for phylogenetic tree reconstruction. These datasets, compiled in early studies and detailed in~\cite{hozumi2024revealing}, span a range of evolutionary scales and biological taxa, providing a robust framework for assessing the accuracy of phylogenetic inference.

	The datasets span a broad range of sequence lengths. For example, Influenza A hemagglutinin (HA) genes consist of approximately 2{,}000 nucleotides; human rhinovirus (HRV) genomes and hepatitis E viruses (HEV) range around 7{,}000 nucleotides; mammalian mitochondrial genomes and Ebola virus (EBOV) genomes contain roughly 17{,}000 nucleotides; and bacterial genomes typically range from several hundred thousand to a few million nucleotides.

	The mammalian mitochondrial dataset comprises 41 species across several mammalian orders: Primates, Carnivora, and Cetacea from Euarchontoglires, and Artiodactyla, Perissodactyla, Lagomorpha, Rodentia, and Erinaceomorpha from Laurasiatheria. The objective is to evaluate the extent to which each method reconstructs clades consistent with established host species classifications.

	The HRV dataset comprises 113 complete HRV genomes consisting of three main groups, HRV-A, HRV-B, and HRV-C, along with three HEV outgroup sequences.

	The HEV dataset comprises 48 complete genomes of HEV grouped into four major genotypic categories, Group 1, Group 2, Group 3 and Group 4.

	The influenza HA genes dataset contains 30 Influenza A hemagglutinin (HA) genes classified into six well-characterized subtypes—H1N1, H2N2, H3N2, H5N1, H7N9, and H7N3.

	Ebolavirus genomes dataset includes 59 complete genomes of Ebola virus categorized into five viral types: Bundibugyo virus (BDBV), Reston virus (RESTV), Ebola virus (EBOV), Sudan virus (SUDV), and Tai Forest virus (TAFV), where EBOV sequences are further annotated by epidemic location and year, enabling evaluation of phylogenetic resolution at both species and outbreak levels.

	The bacterial genomes dataset comprises 30 complete bacterial genomes, classified into nine bacterial families: Bacillaceae, Borreliaceae, Burkholderiaceae, Clostridiaceae, Desulfovibrionaceae, Enterobacteriaceae, Rhodobacteraceae, Staphylococcaceae, and Yersiniaceae. The genome sizes of Borreliaceae range from approximately 0.9 to 2.5 Mb, whereas those of Enterobacteriaceae span 4.0 to 6.5 Mb.

	The narrow-clade {Salmonella} dataset comprises $519$ bacterial genomes drawn from two species-level groups: {Salmonella bongori} ($n=31$) and {Salmonella enterica}, represented by the subspecies {arizonae} ($n=70$), {diarizonae} ($n=72$), {enterica} ($n=200$), {houtenae} ($n=23$), {indica} ($n=3$), {londinensis} ($n=3$), and {salamae} ($n=117$).

	\textbf{Viral Family Classification.}
	For viral classification tasks, we adopted four datasets derived from the NCBI Virus Database (\url{https://www.ncbi.nlm.nih.gov/labs/virus/vssi/}), each annotated with taxonomic labels at the viral family level. These datasets include:

	\begin{enumerate}
		\item \textbf{NCBI 2020}: Contains 6,993 viral genomes and was originally collected in Sun~{et~al.}~\cite{sun2021geometric};
		\item \textbf{NCBI 2022}: Comprises 11,428 genomes, as used in Yu~{et~al.}~\cite{yu2024optimal};
		\item \textbf{NCBI 2024}: A refined version of the NCBI All dataset, from which entries lacking the ``-viridae'' suffix and sequences containing invalid nucleotides were removed as in \cite{hozumi2024revealing}.
		\item \textbf{NCBI 2024 All}: Includes 13,645 genomes collected by Hozumi~{et~al.} as of January 20, 2024~\cite{hozumi2024revealing};
	\end{enumerate}

Reference genomes were obtained directly from the NCBI Virus database, with viral family labels defined according to the taxonomy of the International Committee on Taxonomy of Viruses (ICTV). It is important to note that the NCBI database undergoes continual curation. Consequently, several reference sequences used in prior studies are no longer available and were excluded from our analysis, following the filtering strategy of~\cite{hozumi2024revealing}. Additionally, certain viral sequences have been reassigned to updated taxonomic lineages. For consistency and comparability, we retained the original lineage assignments as reported in the source publications~\cite{hozumi2024revealing, yu2024optimal, sun2021geometric}.

To ensure sufficient representation within each taxonomic class, viral families represented by a single reference genome were excluded from all datasets. A comprehensive overview of dataset composition, including filtering criteria and collection metadata, is presented in Table~\ref{tab:ncbi-dataset-summary}. For further methodological details on dataset construction and curation, we refer the reader to~\cite{hozumi2024revealing, yu2024optimal, sun2021geometric}.

\textbf{Barcode-based fragment-placement dataset.}
For the barcode-based fragment-placement experiment, we constructed an influenza reference-query dataset designed to assess whether CAKR can identify the viral type and homologous gene segment of an external query sequence using alignment-free barcode comparison. Using influenza viruses obtained from NCBI Virus as a model system, the reference library was built from one representative strain for each of the four influenza types, A, B, C, and D, and included all annotated gene segments available for each type, namely eight segments for influenza A and B and seven segments for influenza C and D. Specifically, the reference set consisted of a human influenza A(H3N2) isolate collected in California in 2020, a human influenza B isolate collected in California in 2021, a human influenza C isolate collected in China in 2020, and a swine influenza D isolate collected in Oklahoma in 2011. All reference segments were annotated as complete coding sequences.

The external query set consisted of 14 influenza gene sequences selected to differ from the reference strains in collection year, geographic origin, host background, or coding-sequence completeness. This set included 11 complete coding sequences and 3 partial coding sequences. Influenza A and B queries were taken from human isolates collected in 2026, influenza C queries from human isolates collected in 2024 and 2025, and influenza D queries from bovine isolates collected in 2024. For each query and each reference gene segment, a CAKR barcode representation was computed, and cosine similarity was used to compare the query barcode against the reference barcode library. This dataset provides a controlled benchmark for evaluating alignment-free fragment placement across all four influenza types and across both complete and partial coding sequences.

\textbf{Fragment-placement comparison with SEPP.}
For the fragment-placement comparison with SEPP, we used a separate influenza Type-B PB1 dataset consisting of two components: a fixed backbone reference set and an external fragment query set. The backbone reference set comprised 15 complete or near-complete PB1 gene sequences from influenza Type B viruses, sampled mainly during 2019 to 2021, and used to construct the reference multiple sequence alignment and backbone phylogenetic tree. To generate the query set, five additional influenza Type B PB1 sequences from 2026 were randomly selected, and each was truncated to produce three fragment lengths, 300 bp, 600 bp, and 900 bp, resulting in a total of 15 fragmentary query sequences. The reference sequences were aligned with MAFFT and used as the fixed input structure for SEPP, while the same reference and query sets were also used for the CAKR-based placement analysis. This dataset provides a controlled benchmark for comparing alignment-based HMM placement and alignment-free barcode-based placement on the same gene segment under varying fragment lengths.

\textbf{Source Data file, sheet ``CAKR vs SEPP''.}
This sheet contains the fragment-placement comparison between CAKR and SEPP on an influenza Type-B PB1 benchmark. It lists 15 fragmentary query sequences generated from five 2026 Type-B PB1 sequences by truncation to lengths of 300 bp, 600 bp, and 900 bp. For each query, the table reports the top-ranked placement candidates predicted by CAKR, based on cosine similarity scores, and by SEPP, based on likelihood weight ratios, together with the relative separation between the first and second candidates.

\textbf{Source Data file, sheet ``cosine similarity epiflu''.}
This sheet contains the cosine-similarity matrix for the barcode-based influenza fragment-placement experiment. Columns correspond to the reference barcode library, built from annotated genome segments of representative influenza A, B, C, and D strains, and rows correspond to 14 external influenza query gene sequences, comprising 11 complete coding sequences and 3 partial coding sequences. Each entry gives the cosine similarity between a query barcode and a reference-segment barcode, enabling identification of the closest viral type and homologous gene segment.

\begin{table}[!h]
		\caption{Summary of NCBI viral genome datasets, including collection date, preprocessing steps, number of families, and number of sequences \cite{hozumi2024revealing}.}
	\label{tab:ncbi-dataset-summary}
	\centering
	\renewcommand{\arraystretch}{1.2}
	\begin{tabular}
		{@{}p{5cm}@{\hspace{5pt}}c@{\hspace{10pt}}c@{\hspace{15pt}}c@{\hspace{18pt}}p{5cm}@{}}
		\toprule
		\textbf{Dataset (Reference)} & \textbf{Date} & \textbf{\#Fam.} & \textbf{\#Seq.} & \textbf{Preprocessing Criteria} \\
		\midrule
		NCBI 2020~\cite{sun2021geometric} & Mar 2020 & 83 & 6,993 &
		\makecell[l]{Unknown Baltimore class\\ Unknown family\\ Families with $<$2 sequences} \\
		\midrule
		NCBI 2022~\cite{yu2024optimal} & Mar 2022 & 123 & 11,428 &
		\makecell[l]{Partial sequences\\ Unknown family\\ Families with $<$2 sequences\\ Invalid nucleotides} \\
		\midrule
		NCBI 2024~\cite{hozumi2024revealing} & Jan, 2024 & 199 & 12,154 &
		\makecell[l]{Partial sequences\\ Unknown family\\ Only ``-viridae''\\ Families with $<$2 sequences\\ Invalid nucleotides} \\
		\midrule
		NCBI 2024~\cite{hozumi2024revealing} All & Jan, 2024 & 209 & 13,645 &
		\makecell[l]{Partial sequences\\ Unknown family\\ Families with $<$2 sequences} \\
		\bottomrule
	\end{tabular}
\end{table}

\section{Genetic variant classification}

\begin{figure}[!t]
	\centering
	\begin{subfigure}[b]{1.0\textwidth}
		\includegraphics[width=\textwidth]{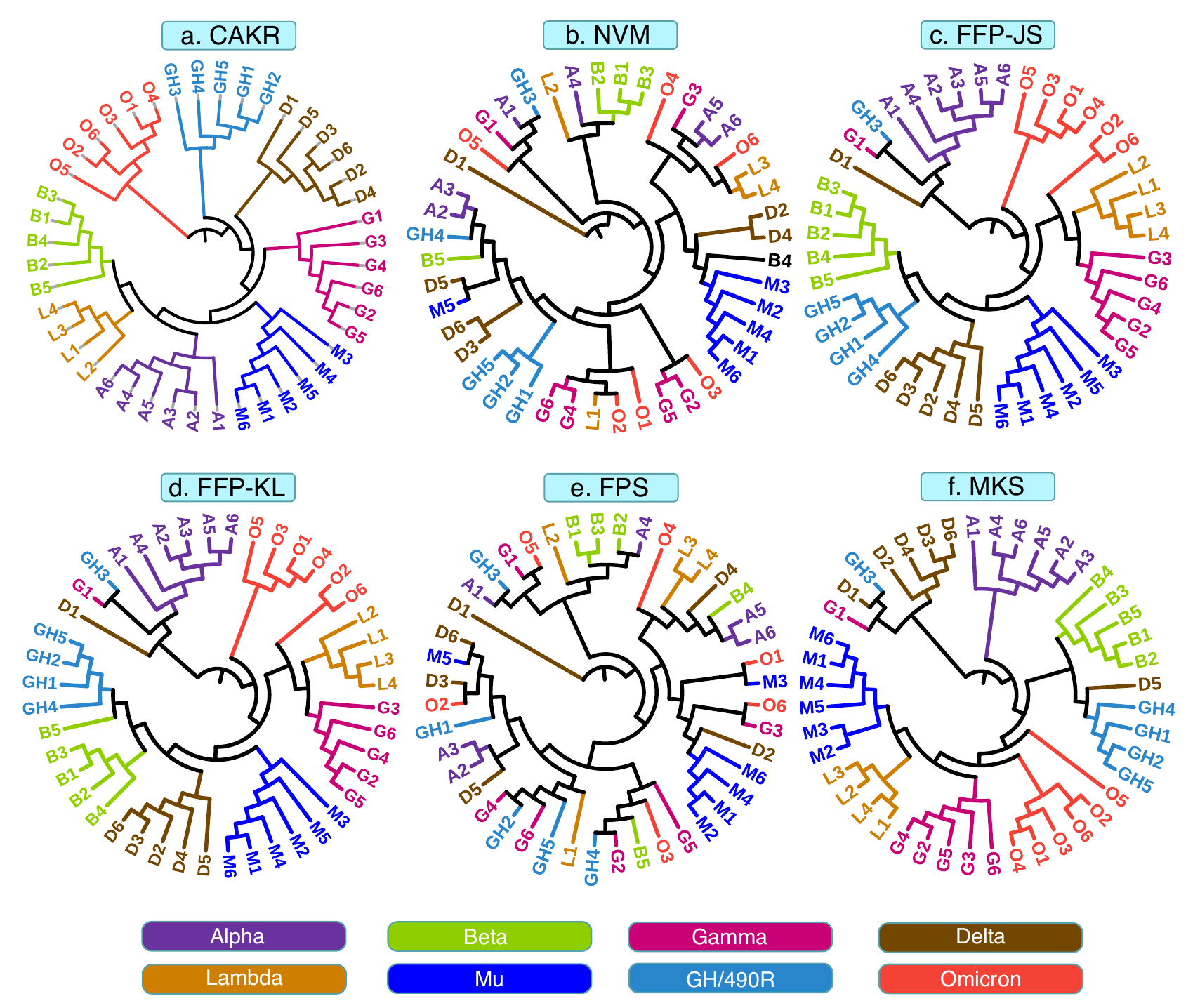}
	\end{subfigure}
	\caption{Performance comparison of various methods for SARS-CoV-2 variant classification was conducted on a dataset comprising 44 complete genomes of severe acute respiratory syndrome coronavirus 2 (SARS-CoV-2), sourced from the GISAID database. CAKR accurately grouped all variant sequences. NVM revealed minimal structure. FFP-KL and FFP-JS misclassified three genomes; MKS misclassified three as well; FPS produced no discernible clustering. For comparison with alignment-based phylogenetic inference, additional SARS-CoV-2 trees reconstructed using IQ-TREE 3~\cite{wong2025iqtree3} and RAxML-NG~\cite{kozlov2019raxmlng} are shown in Supplementary Fig.~\ref{fig:sarscov2_iqtree_raxml}.}
	\label{fig:sarscov2}
\end{figure}

Supplementary Fig.~\ref{fig:sarscov2} shows the phylogenetic trees inferred by six alignment-free methods on the {SARS-CoV-2} dataset used for genetic variant classification. Among these methods, our {CAKR} approach achieved the highest concordance with known variant lineages, outperforming the other five. Compared with the MAFFT-based tree (Supplementary Fig.~\ref{fig:mafft}), which serves as a state-of-the-art alignment-based benchmark, the CAKR tree captures the same high-level clade structure and accurately delineates all major SARS-CoV-2 lineages. While subtle differences in internal branching order exist, both trees identify consistent and biologically meaningful variant groupings. For further reference, we also reconstructed SARS-CoV-2 phylogenies using IQ-TREE 3~\cite{wong2025iqtree3} and RAxML-NG~\cite{kozlov2019raxmlng}, shown in Supplementary Fig.~\ref{fig:sarscov2_iqtree_raxml}.

Notably, the CAKR tree closely mirrors the MAFFT tree in its high-level topology, despite being derived entirely without the use of sequence alignment. This highlights the strength of CAKR as a reliable and scalable alignment-free approach to phylogenetic inference. Its ability to reproduce biologically meaningful relationships among viral variants reinforces its potential for large-scale genomic studies where alignment may be computationally prohibitive or error-prone.

\begin{figure}[!t]
	\centering
	\begin{subfigure}[b]{0.49\textwidth}
		\centering
		\includegraphics[width=0.7\textwidth]{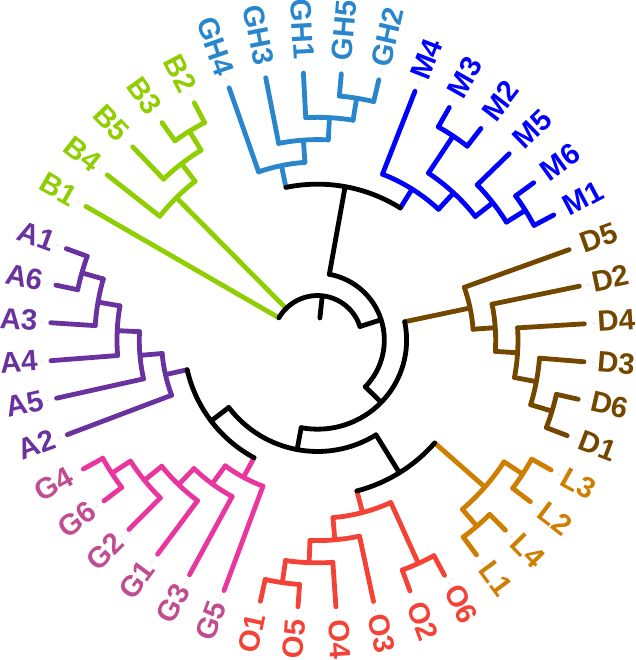}
		\caption{IQ-TREE 3~\cite{wong2025iqtree3}.}
	\end{subfigure}
	\hfill
	\begin{subfigure}[b]{0.49\textwidth}
		\centering
		\includegraphics[width=0.7\textwidth]{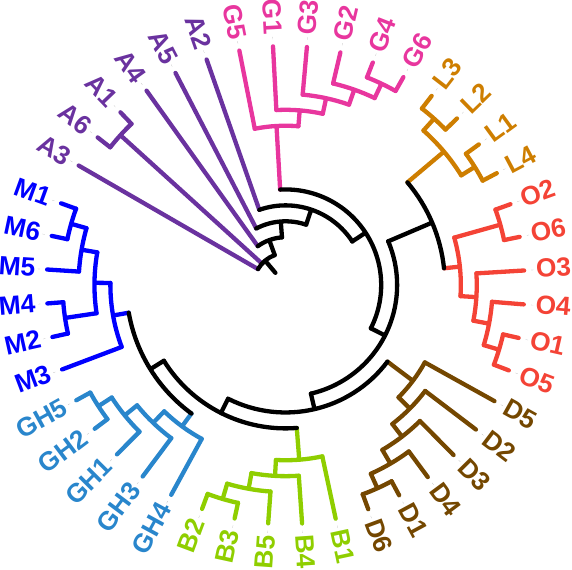}
		\caption{RAxML-NG~\cite{kozlov2019raxmlng}.}
	\end{subfigure}
	\caption{Phylogenetic trees for the SARS-CoV-2 dataset inferred using IQ-TREE 3~\cite{wong2025iqtree3} and RAxML-NG~\cite{kozlov2019raxmlng}. The analysis was conducted on a dataset comprising 44 complete genomes of severe acute respiratory syndrome coronavirus 2 (SARS-CoV-2), sourced from the GISAID database.}
	\label{fig:sarscov2_iqtree_raxml}
\end{figure}

 \begin{figure}[!b]
	\centering
	\begin{subfigure}[b]{1.0\textwidth}
		\includegraphics[width=\textwidth]{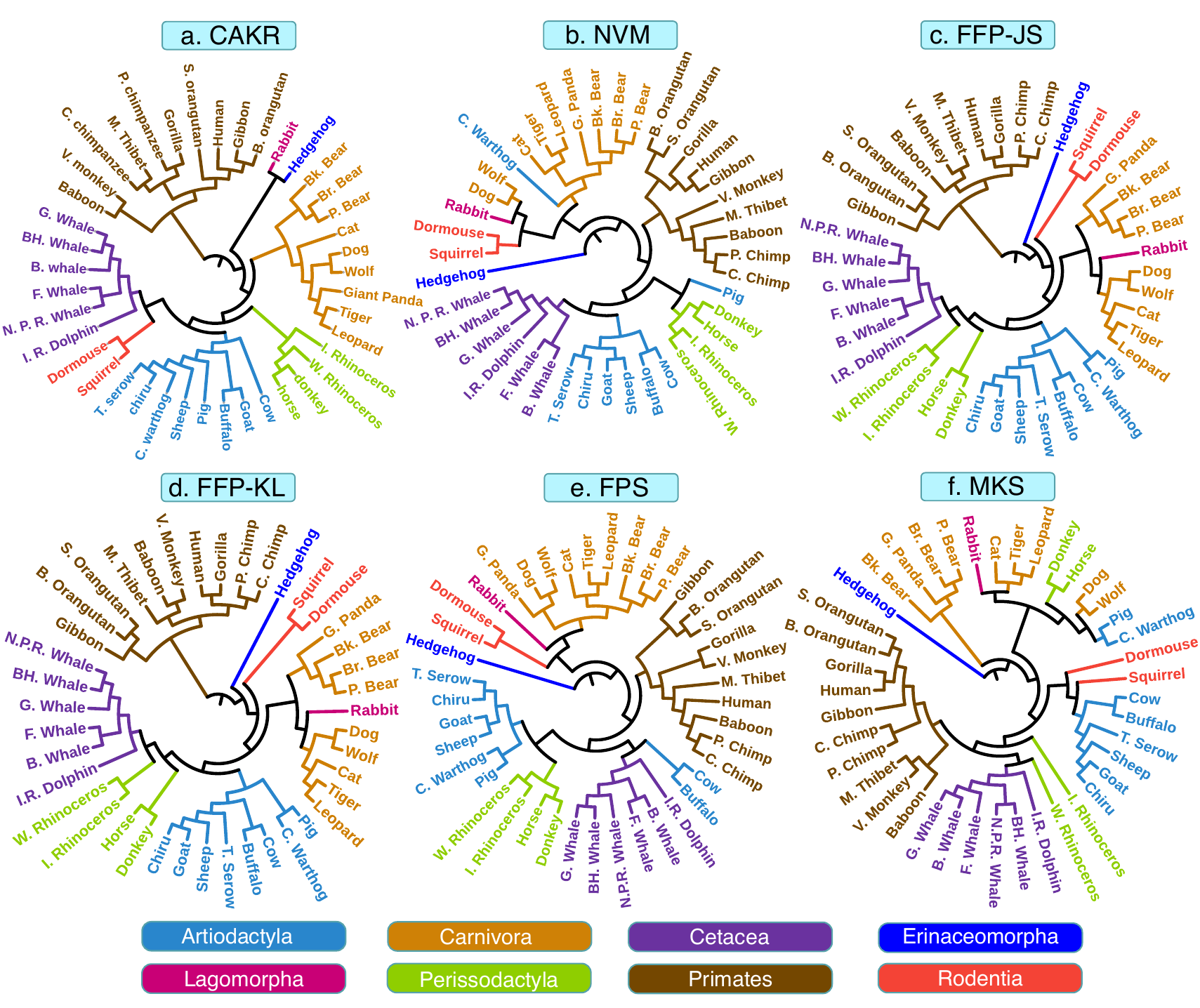}
	\end{subfigure}
	\caption{Performance comparison of various methods on the dataset of 41 complete mammalian mitochondrial genomes. The CAKR method accurately clustered all sequences by their known taxonomy. NVM failed to group warthog and pig with {Artiodactyla}, and produced fragmented {Carnivora} clades. Both FFP-JS and FFP-KL separated {Carnivora} and {Perissodactyla} into multiple clades. The MKS method resulted in three distinct {Carnivora} clades, failed to form coherent clusters for {Rodentia} and {Perissodactyla}, and fragmented Artiodactyla. The FPS method split {Artiodactyla} into two separate clades.}
	\label{fig:mammalian}
\end{figure}

\section{Phylogenetic tree reconstruction}

This section presents the phylogenetic trees generated by the six methods evaluated in this study across the six datasets used for phylogenetic analysis. Among these methods, {CAKR} demonstrated consistently stable performance across all datasets. In contrast, the other methods showed varying performance depending on the dataset, highlighting their sensitivity to data characteristics, and generally performed poorer than CAKR. Moreover, CAKR achieved a normalized Robinson-Foulds (nRF) distance of $0.41$ for the Ebolavirus dataset, $0.20$ for Influenza, $0.38$ for HEV, and $0.42$ for the Mammalian dataset at $k = 4$, as well as $0.43$ for the HRV dataset at $k = 5$.

Supplementary Figs. \ref{fig:cvtree} and \ref{fig:mash} provide additional phylogenetic tree analysis of SARS-CoV-2, mammalian mitochondrial genomes, and HRV datasets using CVTree and Mash, respectively. Table~\ref{tab:phylo_qualitative_summary} provides a qualitative summary of the overall phylogenetic reconstruction performance of the evaluated methods across all benchmark datasets.

\begin{table}[!h]
	\centering
	\caption{Qualitative summary of phylogenetic reconstruction performance across datasets.}
	\label{tab:phylo_qualitative_summary}
	\renewcommand{\arraystretch}{1.2}
	\setlength{\tabcolsep}{4pt}
	\small
	\begin{tabular}{lcccccc}
		\toprule
		\textbf{Dataset} & \textbf{CAKR} & \textbf{NVM} & \textbf{FFP-JS} & \textbf{FFP-KL} & \textbf{FPS} & \textbf{MKS} \\
		\midrule
		Mammalian mitochondrial & Excellent & Weak & Moderate & Moderate & Moderate & Weak \\
		HRV & Strong & Moderate & Moderate & Moderate & Excellent & Weak \\
		HEV & Excellent & Excellent & Excellent & Excellent & Weak & Moderate \\
		Influenza HA & Excellent & Weak & Excellent & Excellent & Weak & Moderate \\
		Ebolavirus & Excellent & Strong & Strong & Strong & Moderate & Strong \\
		Bacterial genomes & Excellent & Weak & Excellent & Excellent & Weak & Excellent \\
		\bottomrule
	\end{tabular}
\end{table}

  \begin{figure}[!b]
 	\centering
 	\begin{subfigure}[b]{1.0\textwidth}
 		\includegraphics[width=\textwidth]{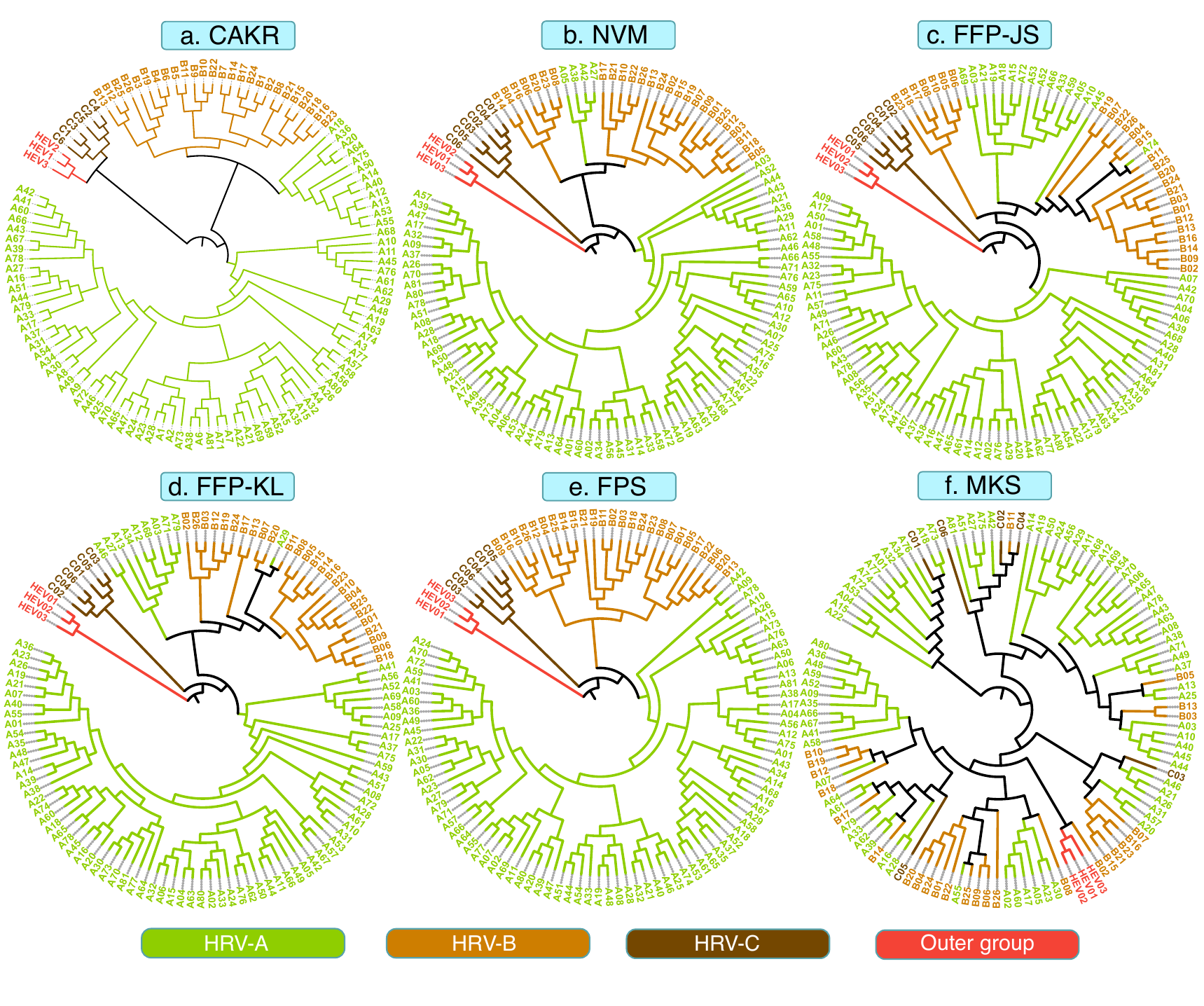}
 	\end{subfigure}
 	\caption{Performance comparison of various methods was carried out on a dataset comprising 113 complete genomes of human rhinoviruses (HRVs), supplemented with three HEV outgroup sequences. The $k=5$ \textsc{CAKR} representation and FPS correctly grouped all HRV genomes and separated them from outgroup sequences. NVM, FFP-JS, and FFP-KL each misclassified one or more HRV-A genomes within the HRV-B clade. MKS failed to produce uniform HRV clades and did not separate the outgroups.}
 	\label{fig:rhinovirus}
 \end{figure}

\begin{figure}[!b]
	\centering
	\begin{subfigure}[b]{1.0\textwidth}
		\includegraphics[width=\textwidth]{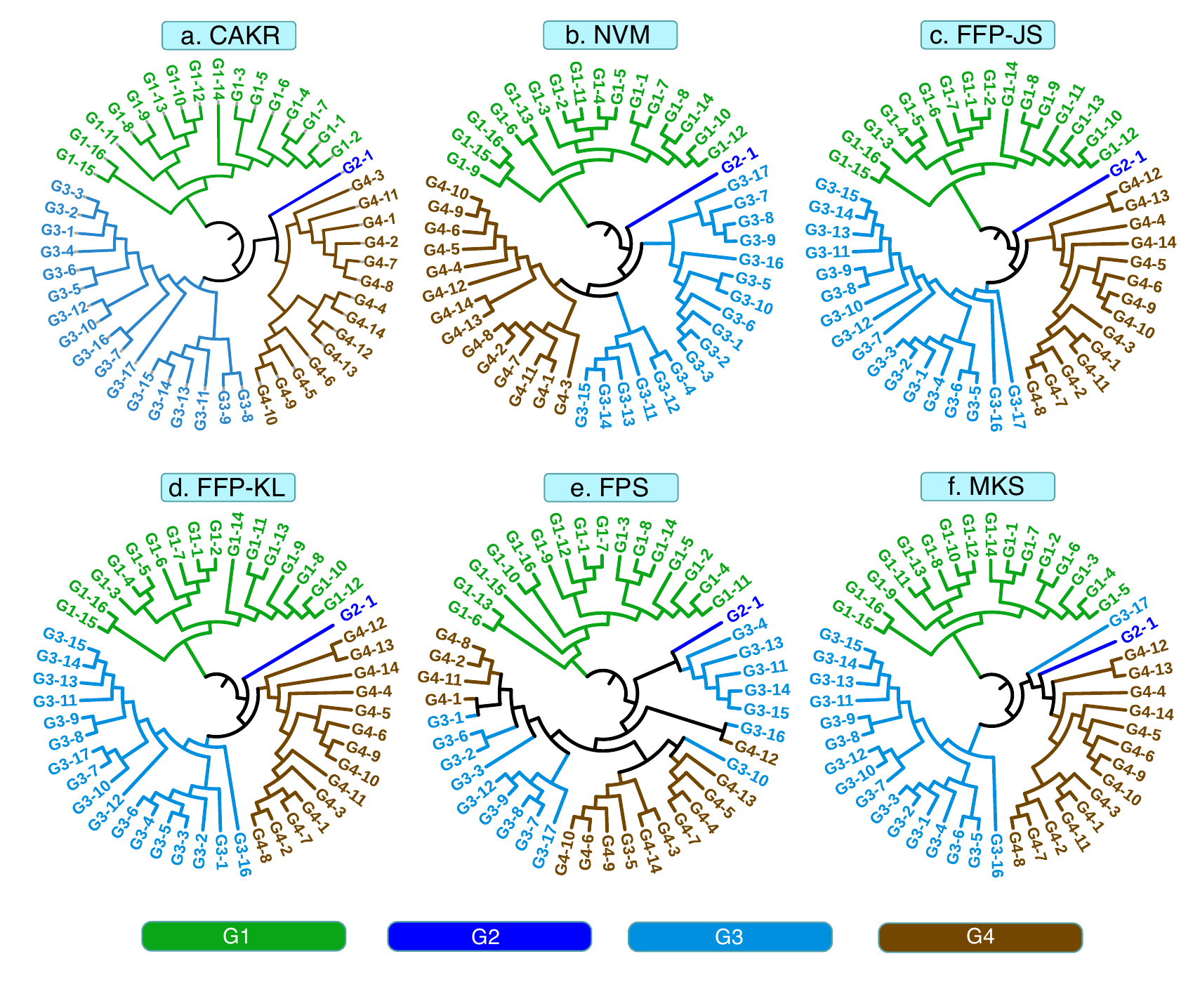}
	\end{subfigure}
	\caption{Performance comparison of various methods on the dataset of 48 complete Hepatitis E virus genomes (HEV). CAKR, FFP-JS, FFP-KL, and NVM correctly grouped all sequences. MKS misclassified one Group 3 genome, and FPS did not separate Groups 3 and 4.}
	\label{fig:HEV}
\end{figure}

\begin{figure}[!b]
	\centering
	\begin{subfigure}[b]{1.0\textwidth}
		\includegraphics[width=\textwidth]{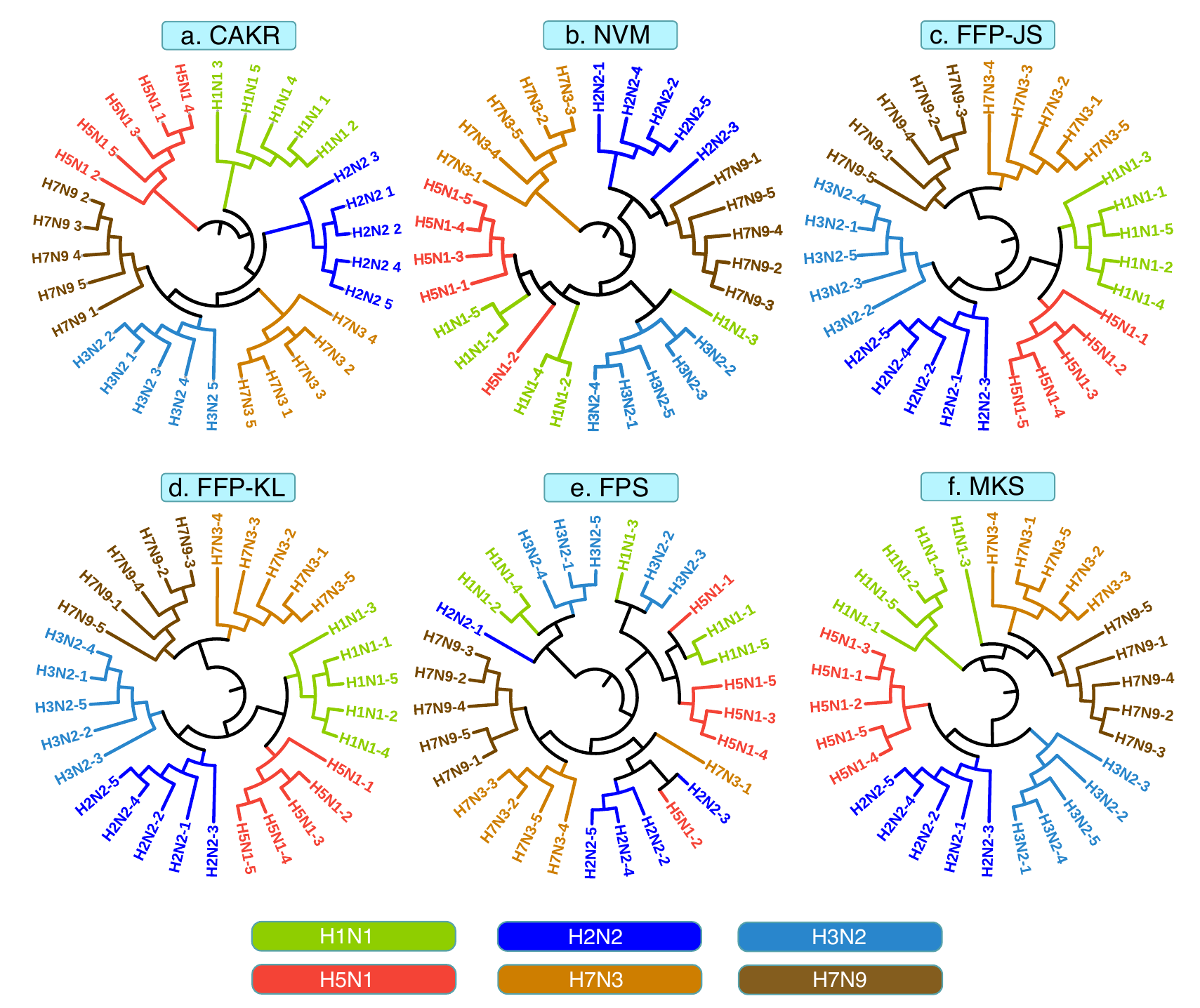}
	\end{subfigure}
	\caption{Performance comparison of various methods on the dataset of 30 influenza HA genes. CAKR, FFP-JS, and FFP-KL formed all clades correctly. NVM failed to group all H1N1 sequences and misclassified one H2N2 sequence; MKS misclassified one H1N1 gene. FPS did not produce clear clustering for most subtypes.}
	\label{fig:influenzaHAgene}
\end{figure}

\begin{figure}[!b]
	\centering
	\begin{subfigure}[b]{1.0\textwidth}
		\includegraphics[width=\textwidth]{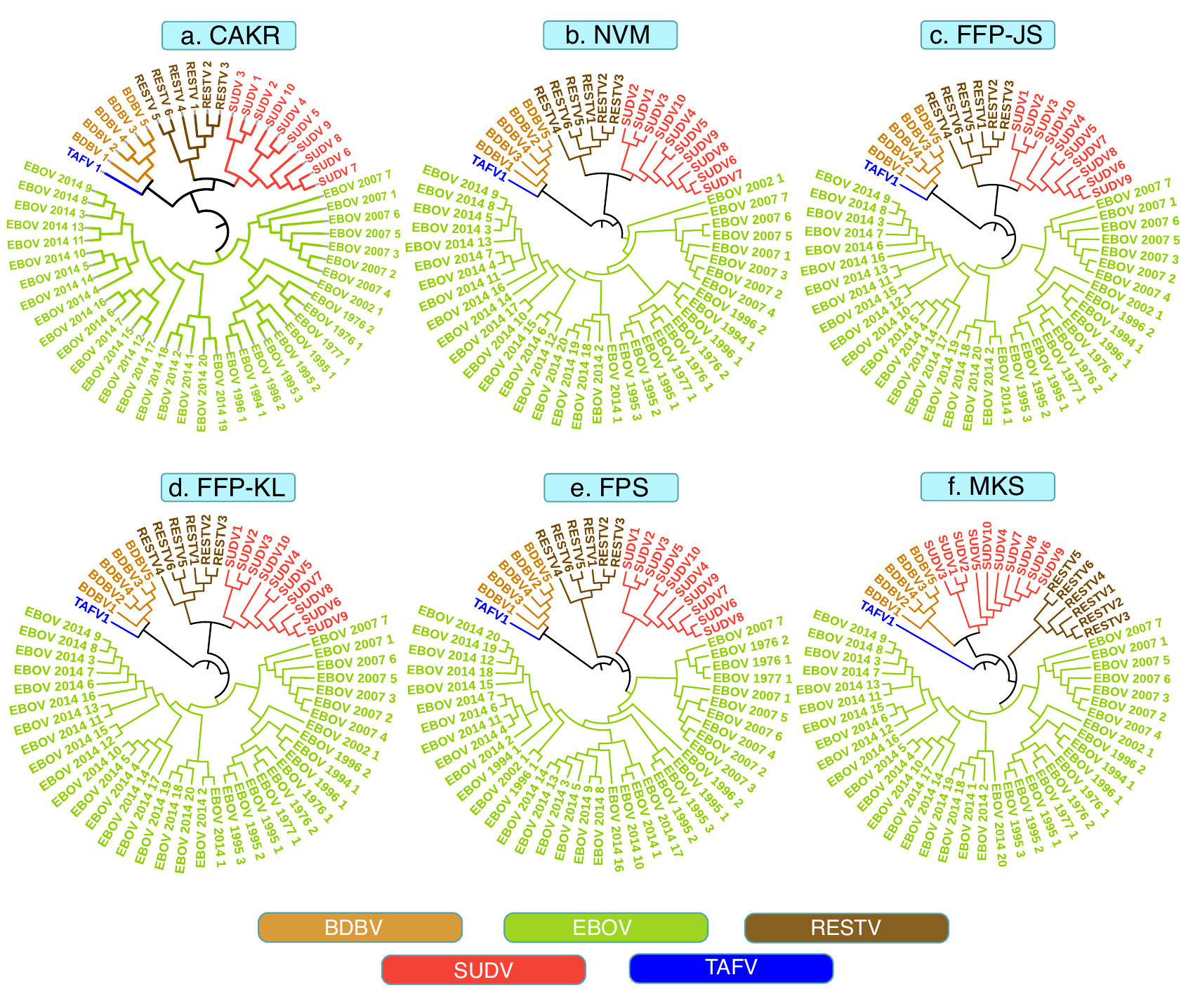}
	\end{subfigure}
	\caption{Performance comparison of various methods on the dataset of 59 complete genomes of ebolaviruses. All methods correctly recovered the viral types. MKS yielded a tree in which EBOV and RESTV shared a node, and FPS did not cluster EBOV epidemic strains distinctly.}
	\label{fig:ebolavirus}
\end{figure}

\begin{figure}[!b]
	\centering
	\begin{subfigure}[b]{1.0\textwidth}
		\includegraphics[width=\textwidth]{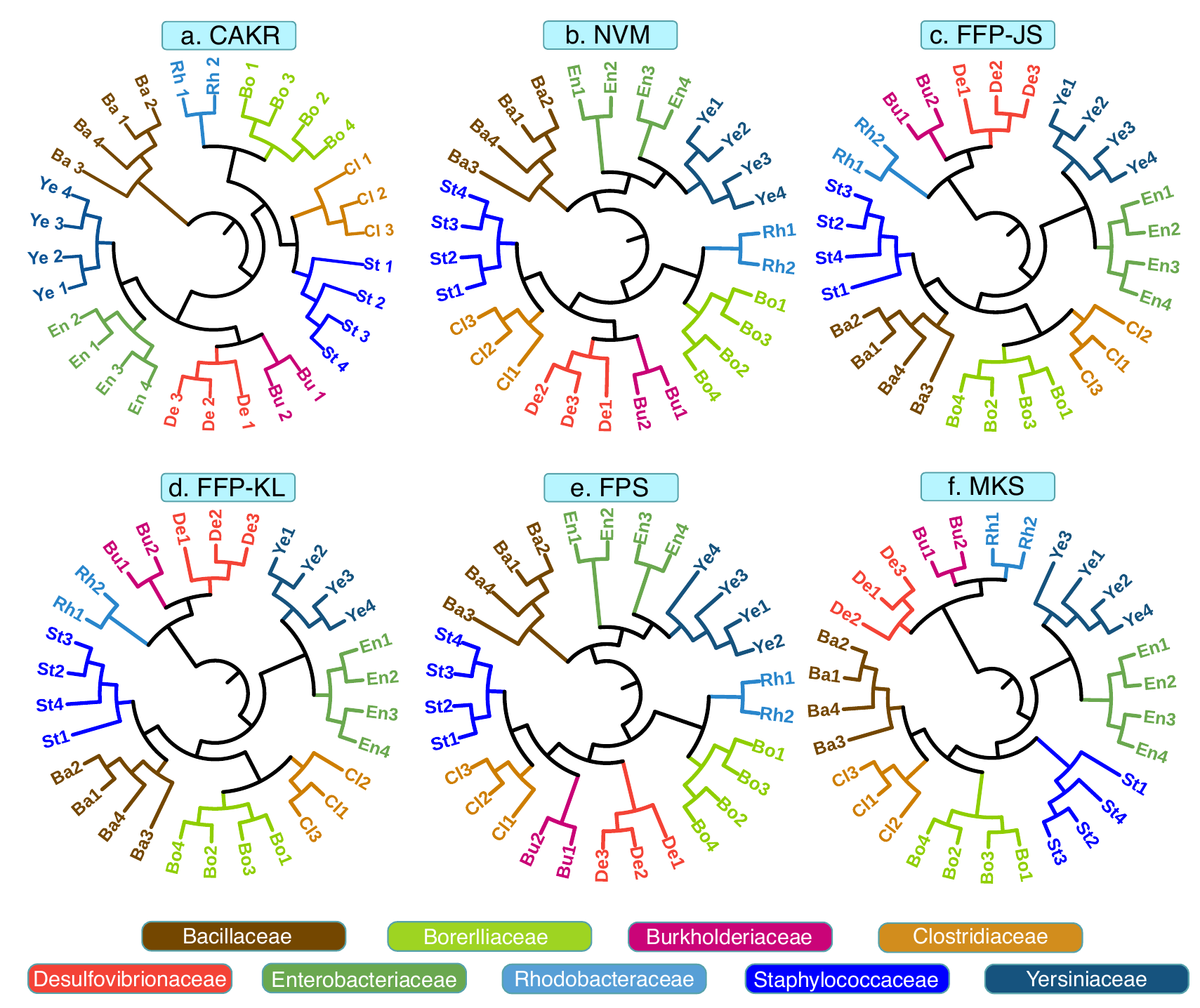}
	\end{subfigure}
	\caption{Performance comparison of various methods on a dataset of 30 complete bacterial genomes revealed that all methods—except NVM and FPS, which split the Enterobacteriaceae clade—successfully recovered the known taxonomic groupings without error.
	}
	\label{fig:bacteria}
\end{figure}

\begin{figure}[!b]
	\centering
	\begin{subfigure}[b]{1.0\textwidth}
		\includegraphics[width=\textwidth]{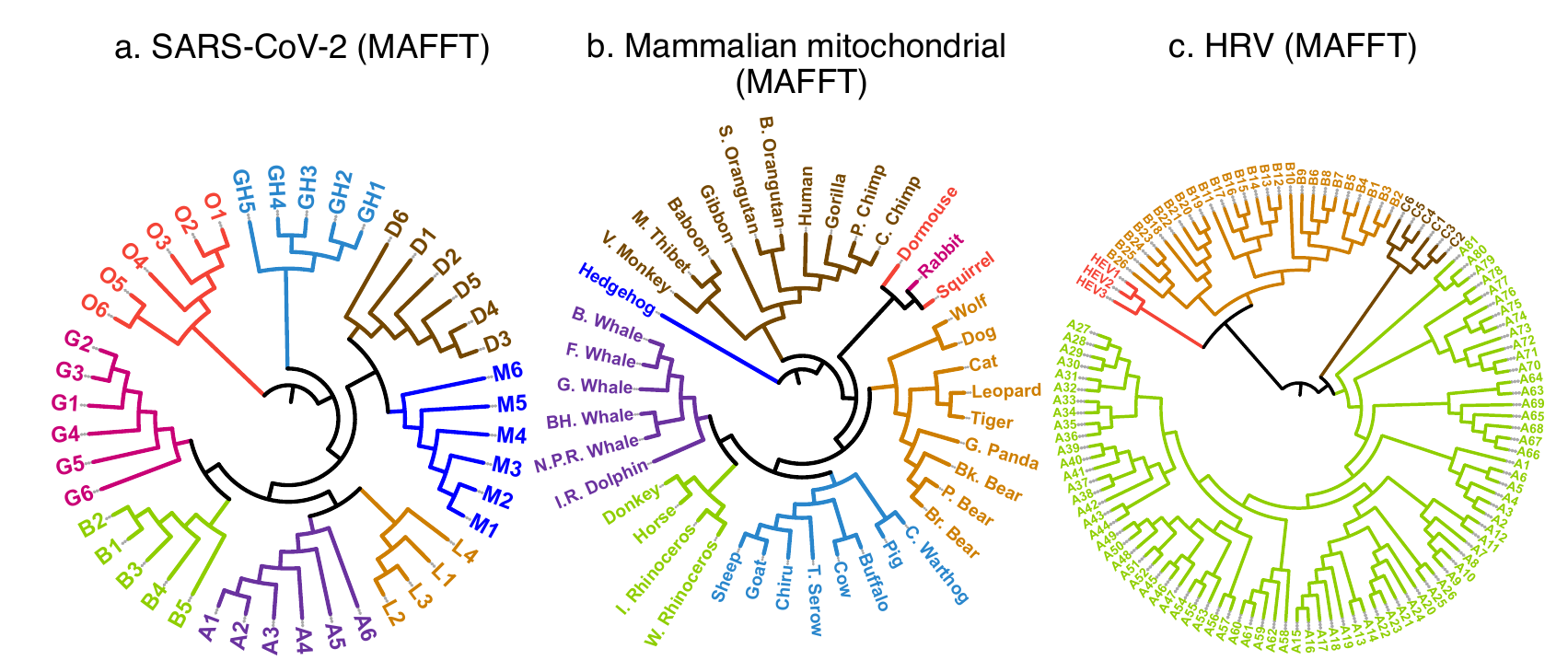}
	\end{subfigure}
	\caption{The phylogenetic trees constructed for the SARS-CoV-2, mammalian mitochondrial genomes, and human rhinovirus (HRV) datasets using MAFFT—an alignment-based method—exhibit a high degree of similarity to the corresponding trees generated by our proposed CAKR framework. This concordance underscores the biological validity of the alignment-free representations produced by CAKR.
	}
	\label{fig:mafft}
\end{figure}

\begin{figure}[!b]
	\centering
	\begin{subfigure}[b]{1.0\textwidth}
		\includegraphics[width=\textwidth]{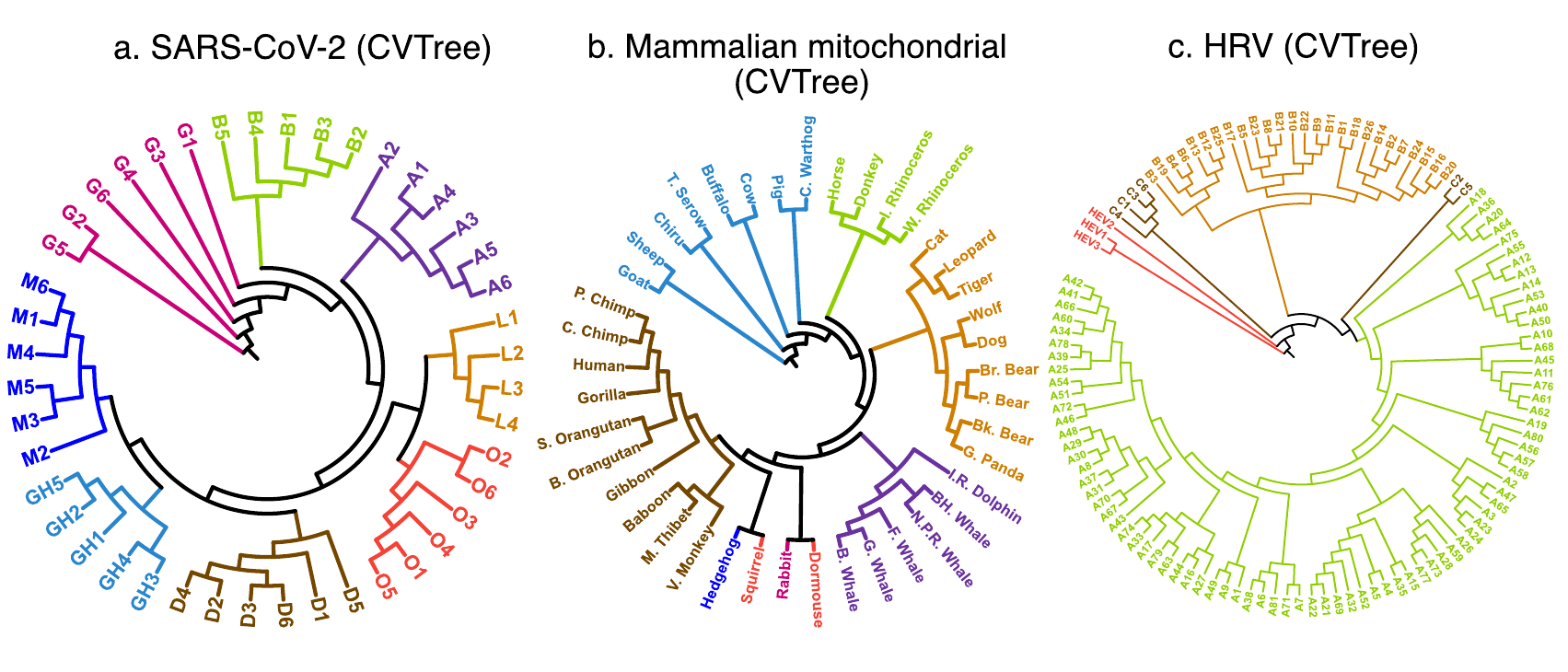}
	\end{subfigure}
	\caption{The phylogenetic trees constructed for the SARS-CoV-2, mammalian mitochondrial genomes, and human rhinovirus (HRV) datasets using CVTree—an alignment-free method.
	}
	\label{fig:cvtree}
\end{figure}

\begin{figure}[!b]
	\centering
	\begin{subfigure}[b]{1.0\textwidth}
		\includegraphics[width=\textwidth]{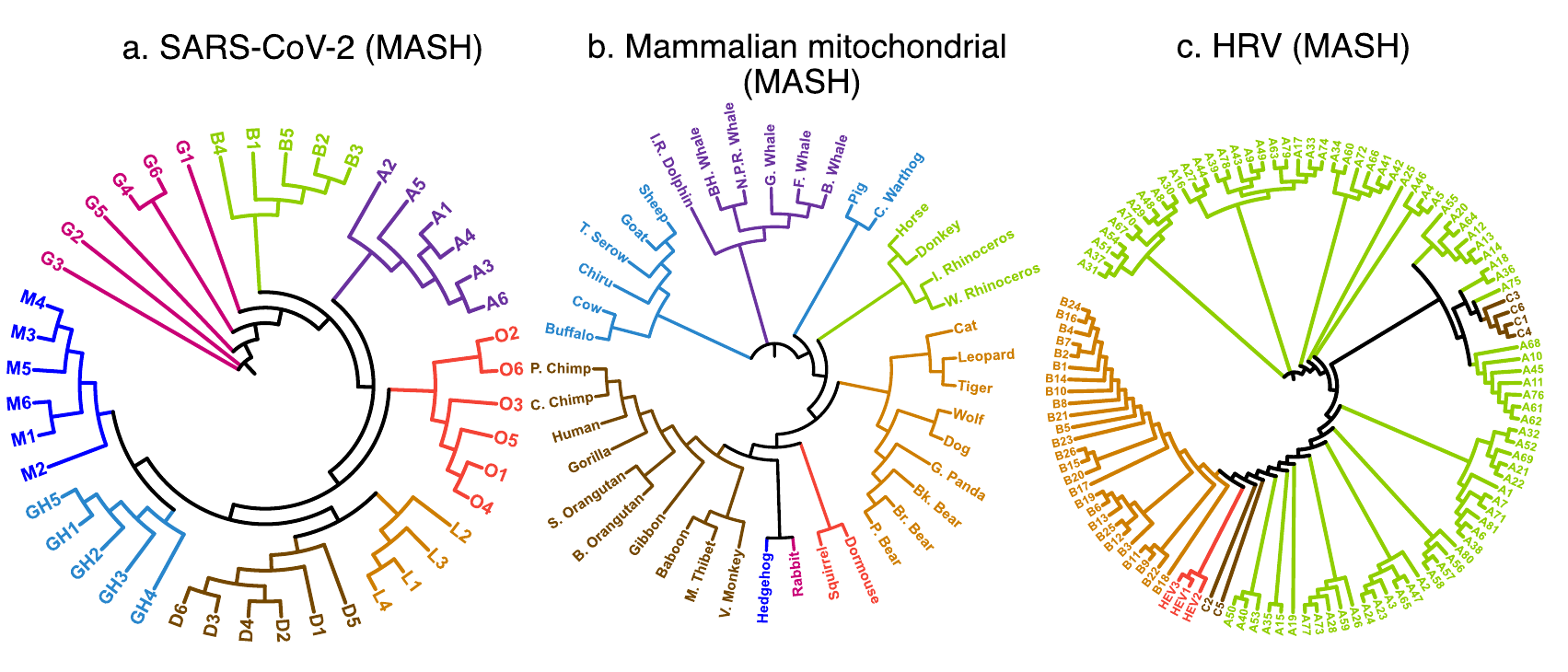}
	\end{subfigure}
	\caption{The phylogenetic trees constructed for the SARS-CoV-2, mammalian mitochondrial genomes, and human rhinovirus (HRV) datasets using Mash - an alignment-free method.
	}
	\label{fig:mash}
\end{figure}

\clearpage

\section{Viral classification}

As discussed in the main manuscript, four datasets for viral classification were taken from \cite{hozumi2024revealing}. A detailed comparison of CAKR and other five methods on classification tasks is given in Table  \ref{tab:performance}.

\begin{table}[!h]
	\centering
	\renewcommand{\arraystretch}{1.1}
	\caption{Comparison of 5-NN classification scores of the six methods.}
	\label{tab:performance}

	\begin{threeparttable}
		\begin{tabular}{llccccc}
			\toprule
			\textbf{Data} & \textbf{Method} & \textbf{ACC} & \textbf{BA} & \textbf{F1} & \textbf{Recall} & \textbf{Precision} \\
			\midrule
			\multirow{6}{*}{NCBI 2020}
			& CAKR    & \textbf{0.913} & \textbf{0.887} & \textbf{0.892} & \textbf{0.887} & \textbf{0.915} \\
			& NVM     & 0.847 & 0.807 & 0.809 & 0.807 & 0.840 \\
			& FFP-JS  & 0.821 & 0.790 & 0.781 & 0.790 & 0.814 \\
			& FFP-KL  & 0.819 & 0.789 & 0.780 & 0.789 & 0.814 \\
			& MKS  & 0.713 & 0.644 & 0.622 & 0.644 & 0.668 \\
			& FPS     & 0.714 & 0.637 & 0.633 & 0.637 & 0.665 \\
			\midrule
			\multirow{6}{*}{NCBI 2022}
			& CAKR    & \textbf{0.902} & \textbf{0.820} & \textbf{0.824} & \textbf{0.819} & \textbf{0.859} \\
			& NVM     & 0.852 & 0.747 & 0.750 & 0.747 & 0.791 \\
			& FFP-JS  & 0.830 & 0.740 & 0.733 & 0.740 & 0.769 \\
			& FFP-KL  & 0.832 & 0.743 & 0.735 & 0.743 & 0.771 \\
			& MKS  & 0.724 & 0.593 & 0.577 & 0.593 & 0.617 \\
			& FPS     & 0.723 & 0.588 & 0.580 & 0.588 & 0.603 \\
			\midrule
			\multirow{6}{*}{NCBI 2024}
			& CAKR    & \textbf{0.876} & \textbf{0.792} & \textbf{0.804} & \textbf{0.792} & \textbf{0.853} \\
			& NVM     & 0.814 & 0.729 & 0.738 & 0.729 & 0.789 \\
			& FFP-JS  & 0.796 & 0.724 & 0.712 & 0.724 & 0.744 \\
			& FFP-KL  & 0.796 & 0.727 & 0.714 & 0.727 & 0.747 \\
			& MKS  & 0.633 & 0.589 & 0.554 & 0.589 & 0.573 \\
			& FPS     & 0.660 & 0.561 & 0.560 & 0.561 & 0.593 \\
			\midrule
			\multirow{6}{*}{NCBI 2024 All}
			& CAKR    & \textbf{0.876} & \textbf{0.794} & \textbf{0.807} & \textbf{0.794} & \textbf{0.853} \\
			& NVM     & 0.809 & 0.729 & 0.738 & 0.729 & 0.788 \\
			& FFP-JS  & 0.799 & 0.721 & 0.711 & 0.721 & 0.745 \\
			& FFP-KL  & 0.799 & 0.724 & 0.712 & 0.724 & 0.745 \\
			& MKS  & 0.638 & 0.589 & 0.555 & 0.589 & 0.575 \\
			& FPS     & 0.651 & 0.561 & 0.559 & 0.561 & 0.591 \\
			\bottomrule
		\end{tabular}

		\begin{tablenotes}[flushleft]
			\footnotesize
			\item \textit{Note:} Boldface indicates the best-performing method for each dataset (within each block of rows).
		\end{tablenotes}
	\end{threeparttable}
\end{table}

\clearpage

\section{Complexity comparison}

To assess the runtime scalability of \textsc{CAKR}, we empirically compared its wall-clock runtimes with those of NVM, FFP-JS, FFP-KL, FPS, MKS, and MAFFT as a function of sequence length and \(k\)-mer size, focusing on the feature-extraction (featurization) step. For each fixed \(k\), all methods exhibit a monotone increase in runtime with respect to sequence length, and the empirical growth rates are very similar across methods. However, \textsc{CAKR} consistently incurs the largest per-sequence runtime at each length and \(k\), reflecting the additional cost of constructing Vietoris-Rips complexes and computing persistent facet invariants. For example, in our implementation the facet-based CAKR features for a single sequence increase from about \(4.6\times 10^{-3}\,\mathrm{s}\) at length \(10\) to roughly \(2.8\times 10^{2}\,\mathrm{s}\) at length \(3\times 10^{6}\). Despite this higher constant factor, \textsc{CAKR} remains practical for genome-scale data and, as shown above, its added computational cost is accompanied by systematic improvements in phylogenetic accuracy and viral classification performance over the other alignment-free methods. We emphasize that these runtimes concern the featurization stage only; subsequent pairwise comparisons are carried out in the resulting feature space. Since our experiments indicate that small \(k\)-mer sizes (\(k=3,4,5\)) are the most informative for \textsc{CAKR}, the corresponding feature vectors have relatively low dimension. As a consequence, distance computation for \textsc{CAKR} is computationally inexpensive, whereas for several competing methods a substantial portion of the total runtime arises in the distance-calculation step between higher-dimensional or sequence-length-dependent representations.

\paragraph{Memory requirements.}
We also estimated the memory required to store the CAKR facet representations.
For a dataset with \(N\) DNA sequences, \(k\)-mer size \(k\), \(F\) filtration values, and maximum facet dimension \(d_{\max}\), the feature matrix has size \(N \times ((d_{\max}+1)F4^k)\). Thus, when stored in double precision, the raw dense numerical storage is \(8N(d_{\max}+1)F4^k\) bytes. In the experiments reported here, we used only facet dimension \(0\), so \(d_{\max}=0\), and \(F=3\) filtration values. Therefore the theoretical dense storage size reduces to \(24N4^k\) bytes.

Table~\ref{tab:cakr-memory-estimates} reports the resulting theoretical dense feature-storage estimates for the datasets used in this study. These estimates describe the final stored CAKR feature arrays and do not include temporary objects created during featurization, such as \(k\)-mer position lists or intermediate Vietoris--Rips and facet structures. In our implementation, the saved NumPy objects may also differ slightly from these estimates because of object-array and file-format overhead. Nevertheless, the estimates show that the final CAKR feature storage is modest for the phylogenetic datasets and remains manageable even for the larger viral classification benchmarks at the small \(k\)-values used in the experiments.

The pairwise distance matrix is independent of \(k\) after the CAKR feature vectors have been computed. For \(N\) sequences, a full double-precision distance matrix requires \(8N^2\) bytes. If only the condensed upper-triangular form is stored, the requirement is approximately half this amount. Thus, for the largest dataset considered here, NCBI 2024 All with \(N=13{,}645\), the full double-precision distance matrix requires about \(1.49\) GB, while the condensed symmetric form requires about \(0.75\) GB. For the smaller phylogenetic datasets, the distance matrix is negligible compared with the feature-extraction step.

\begin{table}[!t]
\centering
\caption{Theoretical dense storage size of CAKR facet features for the datasets used in this study. The estimates assume DNA \(k\)-mers, double-precision storage, maximum facet dimension \(d_{\max}=0\), and \(F=3\) filtration values. Thus the storage size is \(24N4^k\) bytes for a dataset with \(N\) sequences.}
\label{tab:cakr-memory-estimates}
\renewcommand{\arraystretch}{1.15}
\begin{tabular}{lrrrrr}
\toprule
\textbf{Dataset} & \(\boldsymbol{N}\) & \(\boldsymbol{k=3}\) & \(\boldsymbol{k=4}\) & \(\boldsymbol{k=5}\) & \(\boldsymbol{k=6}\) \\
\midrule
SARS-CoV-2        &     44 & 67.6 KB  & 270.3 KB & 1.08 MB   & 4.33 MB \\
Mammalian         &     41 & 63.0 KB  & 251.9 KB & 1.01 MB   & 4.03 MB \\
HRV               &    116 & 178.2 KB & 712.7 KB & 2.85 MB   & 11.40 MB \\
HEV               &     48 & 73.7 KB  & 294.9 KB & 1.18 MB   & 4.72 MB \\
Influenza HA      &     30 & 46.1 KB  & 184.3 KB & 737.3 KB  & 2.95 MB \\
Ebolavirus        &     59 & 90.6 KB  & 362.5 KB & 1.45 MB   & 5.80 MB \\
Bacterial genomes &     30 & 46.1 KB  & 184.3 KB & 737.3 KB  & 2.95 MB \\
Salmonella        &    519 & 797.2 KB & 3.19 MB  & 12.75 MB  & 51.02 MB \\
NCBI 2020         &  6,993 & 10.74 MB & 42.96 MB & 171.86 MB & 687.44 MB \\
NCBI 2022         & 11,428 & 17.55 MB & 70.21 MB & 280.85 MB & 1.12 GB \\
NCBI 2024         & 12,154 & 18.67 MB & 74.67 MB & 298.70 MB & 1.19 GB \\
NCBI 2024 All     & 13,645 & 20.96 MB & 83.83 MB & 335.34 MB & 1.34 GB \\
\bottomrule
\end{tabular}
\end{table}

\begin{figure}[!t]
	\centering
	\begin{subfigure}[b]{0.40\textwidth}
		\centering
		\includegraphics[width=\textwidth]{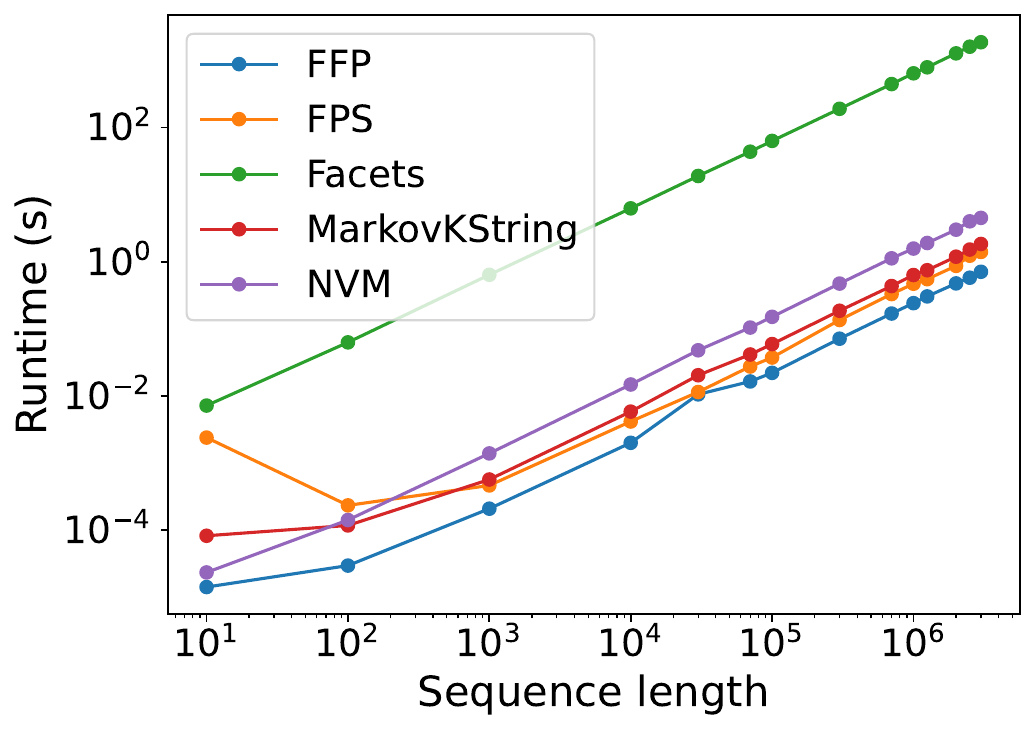}
		\caption{Runtime vs.\ length for \(k=3\).}
		\label{fig:runtime_k3}
	\end{subfigure}
	\begin{subfigure}[b]{0.40\textwidth}
		\centering
		\includegraphics[width=\textwidth]{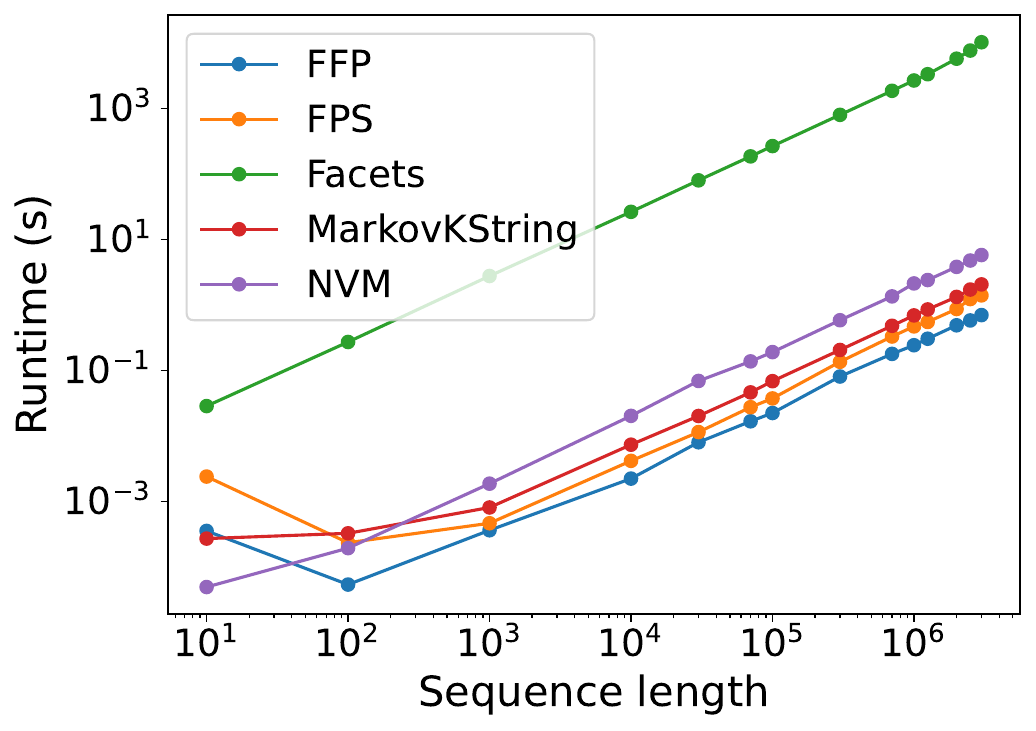}
		\caption{Runtime vs.\ length for \(k=4\).}
		\label{fig:runtime_k4}
	\end{subfigure}

	\vspace{0.5em}

	\begin{subfigure}[b]{0.40\textwidth}
		\centering
		\includegraphics[width=\textwidth]{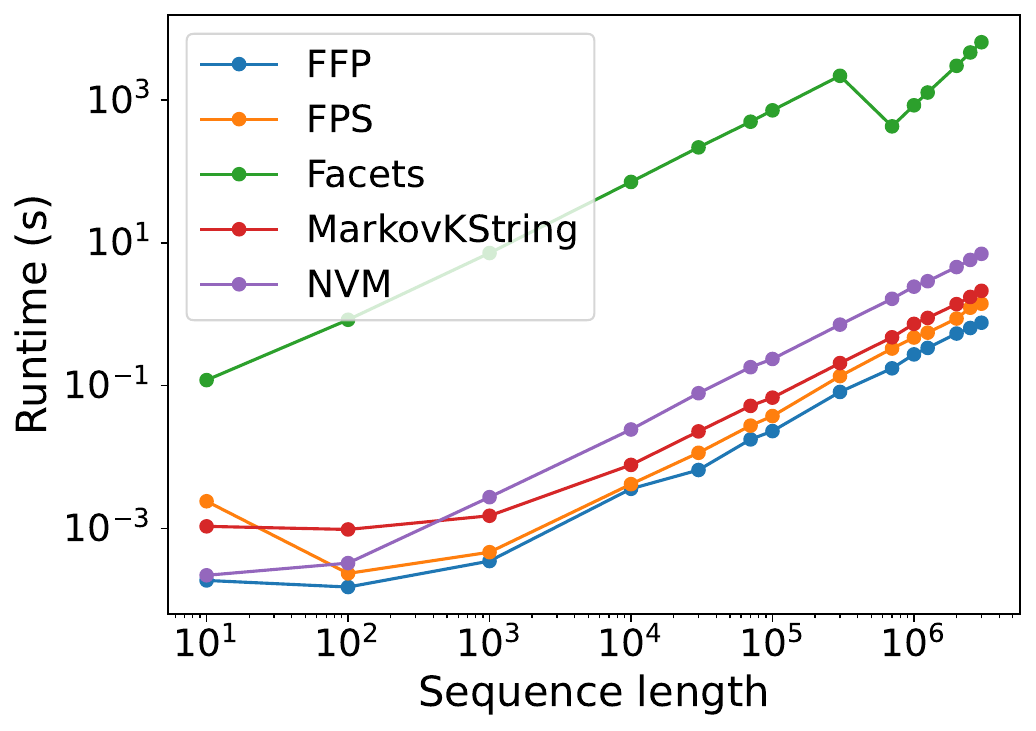}
		\caption{Runtime vs.\ length for \(k=5\).}
		\label{fig:runtime_k5}
	\end{subfigure}
	\begin{subfigure}[b]{0.40\textwidth}
		\centering
		\includegraphics[width=\textwidth]{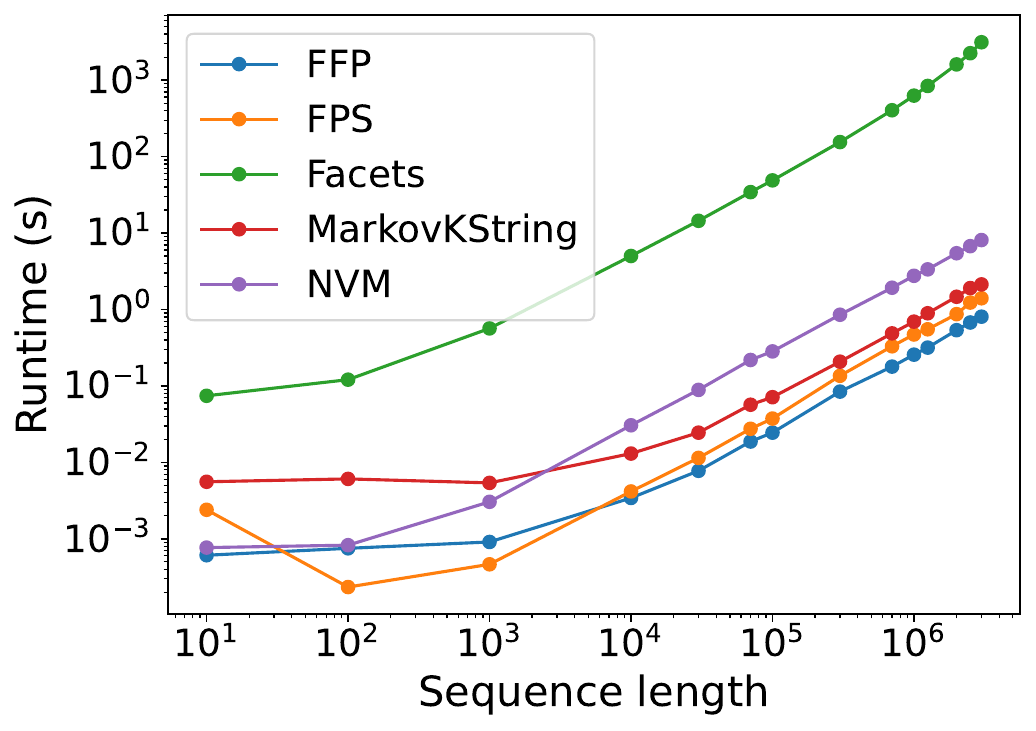}
		\caption{Runtime vs.\ length for \(k=6\).}
		\label{fig:runtime_k6}
	\end{subfigure}
	\caption{Runtime of the proposed facet-based \textsc{CAKR} implementation and all baseline methods as a function of sequence length for different \(k\)-mer sizes. Source data are provided as a Source Data file.}
	\label{fig:runtime_vs_length_allk}
\end{figure}

\clearpage

\section{Effect of \texorpdfstring{$k$}{k}-mer size on performance}\label{sec:kmer_performance}

We next examined how the choice of \(k\)-mer size influences the performance of CAKR and the baseline alignment-free methods across both classification and phylogenetic tasks.

We first report the direct dependence of 1-NN classification accuracy on \(k\) for CAKR, NVM, FFP-JS, FFP-KL, and MKS across the NCBI 2020, NCBI 2022, NCBI 2024, and NCBI 2024 All datasets. As shown in Fig.~\ref{fig:ncbi_1nn_accuracy}, these curves provide an explicit comparison of how sensitive each method is to the choice of \(k\)-mer size. Across all four datasets, CAKR maintains strong accuracy over a broad range of \(k\) values and achieves its best performance at moderate \(k\), whereas several baseline methods show greater variability as \(k\) changes.

To complement this classification-based view, we also summarize the effect of \(k\) on phylogenetic performance for the SARS-CoV-2, rhinovirus, and mammalian datasets. For each dataset, each method, and each \(k \in \{3,4,5,6\}\), we first compute the performance at that specific \(k\). We then report, for each value of \(k\), the cumulative-best performance
\[
\max_{k' \leq k} \mathrm{Perf}(\text{method}, k'),
\]
so that the resulting curves are non-decreasing and indicate how quickly each method approaches its best attained purity as the \(k\)-mer size increases.

Supplementary Fig.~\ref{fig:performance_k_all} summarizes these cumulative-best purity results. This representation highlights both the strongest phylogenetic performance achieved by each method and the smallest \(k\) at which that performance is effectively reached. Across all three datasets, CAKR reaches high purity at relatively small or moderate \(k\) values and remains competitive or superior to the other alignment-free baselines throughout the explored range.

Taken together, the direct 1-NN accuracy curves and the cumulative-best purity curves provide complementary perspectives on \(k\)-dependence. The former shows the explicit sensitivity of classification performance to \(k\), whereas the latter summarizes how rapidly each method reaches its strongest phylogenetic regime. Overall, these results indicate that CAKR combines strong peak performance with comparatively stable behavior across a practically relevant range of \(k\)-mer sizes.

\begin{figure}[t]
	\centering
	\includegraphics[width=\textwidth]{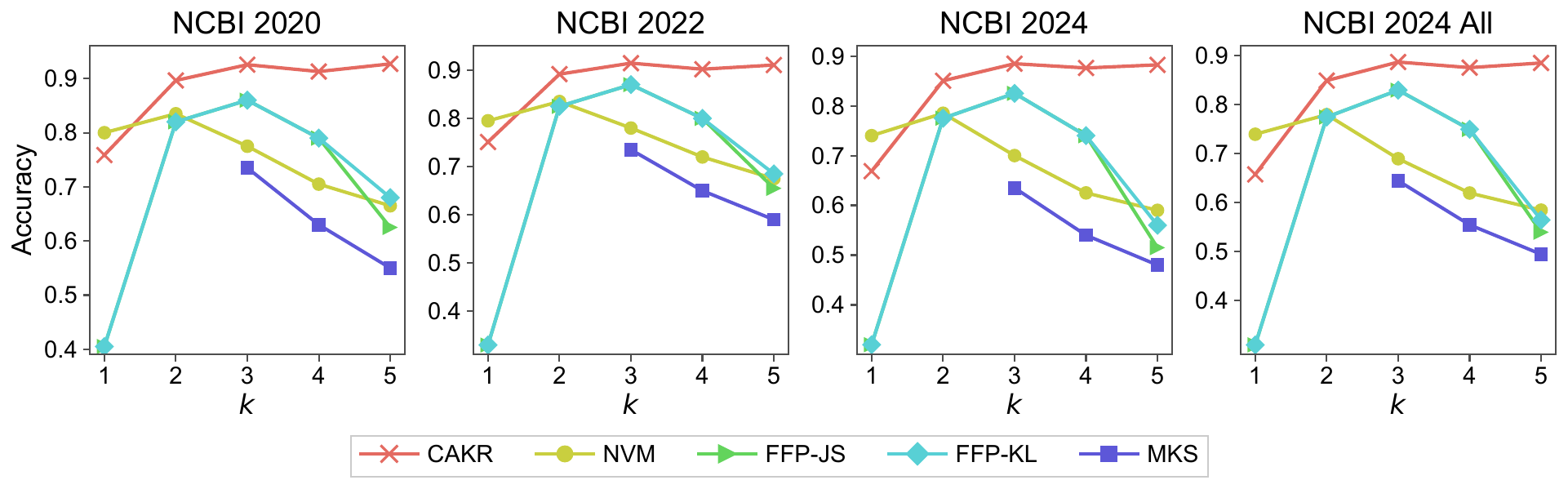}
	\caption{1-NN classification accuracy as a function of \(k\)-mer size for CAKR, NVM, FFP-JS, FFP-KL, and MKS across the NCBI 2020, NCBI 2022, NCBI 2024, and NCBI 2024 All datasets. These plots provide a direct comparison of the sensitivity of each method to the choice of \(k\). Source data are provided as a Source Data file.}
	\label{fig:ncbi_1nn_accuracy}
\end{figure}

\begin{figure}[htbp]
	\centering
	\includegraphics[width=\textwidth]{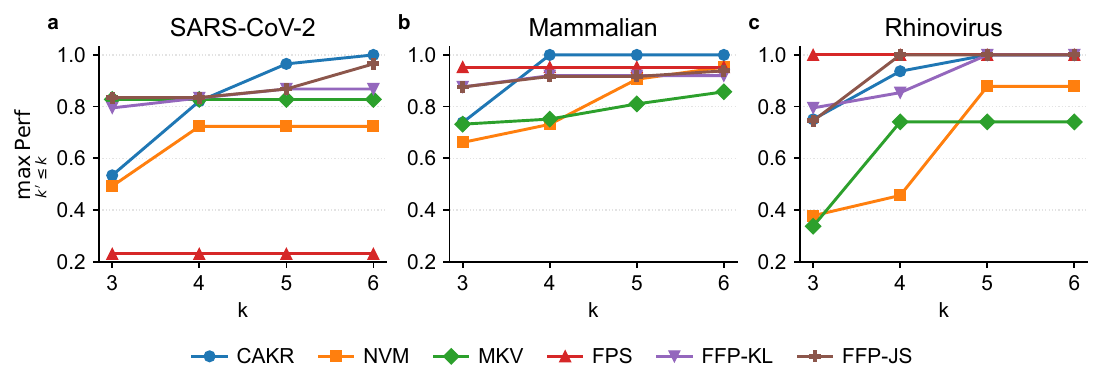}
	\caption{Cumulative-best purity of CAKR and the baseline alignment-free methods as a function of the \(k\)-mer size \(k\) on the SARS-CoV-2, rhinovirus, and mammalian datasets. For each method and each value of \(k\), we report the maximum purity achieved over all \(k' \leq k\), with \(k' \in \{3,4,5,6\}\). These curves summarize how quickly each method approaches its best attained phylogenetic purity as the \(k\)-mer size increases. Dataset sizes were $n=44$ SARS-CoV-2 sequences, $n=41$ mammalian mitochondrial sequences and $n=116$ sequences for the rhinovirus analysis, comprising 113 HRV sequences and 3 HEV outgroups. Each marker represents one computed cumulative-best purity value. No averaging, error bars or statistical tests were used. Source data are provided as a Source Data file.}
	\label{fig:performance_k_all}
\end{figure}

\clearpage

\section{Comparison of UPGMA and Neighbor-Joining for CAKR (purity metric)}

To evaluate the sensitivity of CAKR to the tree-construction procedure, we compared clustering purity under two phylogenetic reconstruction methods, UPGMA and Neighbor-Joining (NJ), while keeping the distance matrices and all other experimental settings fixed. UPGMA (Unweighted Pair Group Method with Arithmetic Mean) is a hierarchical agglomerative clustering method that iteratively merges the closest clusters using the average pairwise distance between their members, thereby producing a rooted tree under an implicit constant-rate assumption. By contrast, NJ is a distance-based phylogenetic method that constructs an unrooted tree by iteratively joining pairs of taxa so as to minimize the total branch length \cite{felsenstein2004inferring}. Under UPGMA, the highest purities achieved across all tested values of \(k\) were
\[
P^{\mathrm{UPGMA}}_{\text{rhinovirus}} = 1.000,\quad
P^{\mathrm{UPGMA}}_{\text{mammalian}} = 1.000,\quad
P^{\mathrm{UPGMA}}_{\text{SARS-CoV-2}} = 1.000.
\]
Using NJ instead, the corresponding best purities were
\[
P^{\mathrm{NJ}}_{\text{rhinovirus}} = 0.649,\quad
P^{\mathrm{NJ}}_{\text{mammalian}} = 0.716,\quad
P^{\mathrm{NJ}}_{\text{SARS-CoV-2}} = 1.000.
\]

These results show that UPGMA achieves perfect clustering on all three datasets, whereas NJ leads to a reduction in purity for the rhinovirus and mammalian datasets while maintaining identical performance on the SARS-CoV-2 dataset. Overall, these results indicate that the quantitative performance of CAKR depends on the choice of tree-construction method. When coupled with UPGMA, CAKR produces consistently strong clustering results across all three datasets. In contrast, when coupled with Neighbor-Joining, the method shows reduced stability on the rhinovirus and mammalian datasets, although performance remains unchanged for SARS-CoV-2 (see Supplementary Fig.~\ref{fig:NJ_purity}).

\begin{figure}[!h]
	\centering
	\begin{subfigure}[b]{0.5\textwidth}
		\includegraphics[width=\textwidth]{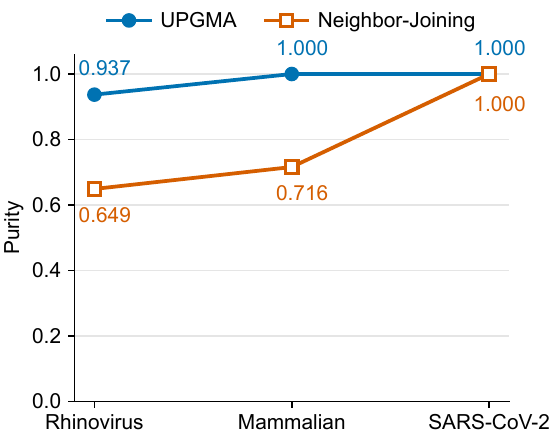}
	\end{subfigure}
	\caption{Comparison of the best clustering purities achieved by CAKR
		over all tested values of \(k\) using UPGMA and Neighbor-Joining (NJ).
		Each point represents one computed purity value for one dataset:
		rhinovirus, \(n=116\) sequences; mammalian mitochondrial,
		\(n=41\) sequences; and SARS-CoV-2, \(n=44\) sequences.
		No averaging or error bars were used. Source data are provided as a Source Data file.}
	\label{fig:NJ_purity}
\end{figure}

\clearpage
\section{Robustness of method rankings}
\label{sec:kendall_analysis}

To evaluate the robustness of the comparative ranking of methods, we used Kendall's coefficient of concordance ($W$), a nonparametric measure of agreement among multiple rankings. Kendall's $W$ ranges from 0 to 1, where $W=1$ indicates complete agreement and $W=0$ indicates no agreement. For $n$ items ranked by $m$ judges, with correction for tied ranks, it is defined as
\begin{equation}
	W \;=\;
	\frac{12 \sum_{i=1}^{n} (R_i - \bar{R})^2}
	{m^2 (n^3 - n)-m\sum_{j=1}^{m}T_j},
	\label{eq:kendallW}
\end{equation}
where $R_i$ is the sum of ranks assigned to item $i$,
\begin{equation}
	\bar{R} \;=\; \frac{m(n+1)}{2}
	\label{eq:Rbar}
\end{equation}
is the mean rank sum, and $T_j=\sum_g(t_{jg}^3-t_{jg})$, where $t_{jg}$ is the number of tied items in the $g$-th group of ties assigned by judge $j$. The significance of concordance was assessed using
\begin{equation}
	\chi^2 \;=\; m(n-1)W,
	\label{eq:chisq}
\end{equation}
which approximately follows a $\chi^2$ distribution with $n-1$ degrees of freedom.

We first examined the agreement of method rankings across evaluation metrics within each dataset. In this analysis, the six methods (CAKR, NVM, FFP-JS, FFP-KL, MKS, and FPS) were treated as the ranked items, and the five evaluation metrics (Accuracy, Balanced Accuracy, F1-score, Recall, and Precision) were treated as judges. For each dataset, methods were ranked within each metric, with rank \(1\) assigned to the best-performing method, and Kendall's $W$ was computed to quantify the consistency of these rankings. As shown in Table~\ref{tab:kendall_typeI}, concordance was high across all datasets, with $W$ ranging from $0.965$ to $0.982$. Using \eqref{eq:chisq} with $m=5$ and $n=6$ gives $\chi^2 = 25W$ with $df=5$, indicating significant agreement. The corresponding $\chi^2$ statistics and $P$-values were $24.195$ and $0.000199$ for NCBI 2020, $24.543$ and $0.000171$ for NCBI 2022, $24.195$ and $0.000199$ for NCBI 2024, and $24.133$ and $0.000205$ for NCBI 2024 All, respectively. These results show that the relative ordering of methods is largely unchanged across different evaluation metrics.

\begin{table}[h]
	\caption{Concordance of method rankings across evaluation metrics within each dataset.}
	\label{tab:kendall_typeI}
	\centering
	\begin{tabular}{lc}
		\toprule
		Dataset & Kendall's $W$ \\
		\midrule
		NCBI 2020     & 0.968 \\
		NCBI 2022     & 0.982 \\
		NCBI 2024     & 0.968 \\
		NCBI 2024 All & 0.965 \\
		\bottomrule
	\end{tabular}
\end{table}

We next assessed the agreement of method rankings across dataset versions for each evaluation metric. Here, the six methods again served as the ranked items, whereas the four dataset versions (NCBI 2020, NCBI 2022, NCBI 2024, and NCBI 2024-All) acted as judges. For each metric, methods were ranked within each dataset version, and Kendall's $W$ was used to measure the consistency of the resulting rankings. As shown in Table~\ref{tab:kendall_typeII}, concordance remained high for all metrics, with $W$ ranging from $0.964$ to $0.979$. Using \eqref{eq:chisq} with $m=4$ and $n=6$ gives $\chi^2 = 20W$ with $df=5$, confirming that method rankings are also stable across dataset versions. The corresponding $\chi^2$ statistics and $P$-values were $19.275$ and $0.001708$ for Accuracy, $19.571$ and $0.001504$ for Balanced Accuracy, F1-score, and Recall, and $19.275$ and $0.001708$ for Precision, respectively.

\begin{table}[h]
	\caption{Concordance of method rankings across dataset versions for each evaluation metric.}
	\label{tab:kendall_typeII}
	\centering
	\begin{tabular}{lc}
		\toprule
		Metric & Kendall's $W$ \\
		\midrule
		Accuracy          & 0.964 \\
		Balanced Accuracy & 0.979 \\
		F1-score          & 0.979 \\
		Recall            & 0.979 \\
		Precision         & 0.964 \\
		\bottomrule
	\end{tabular}
\end{table}

Taken together, these results show that the comparative performance hierarchy is highly robust. The ranking of methods is largely stable both to the choice of evaluation metric and to differences among dataset versions, supporting the stability of the overall conclusion that CAKR consistently ranks among the top-performing approaches.

\clearpage
\section{Choice of $k$-mer size and feature representation}

Rather than fixing a single $k$-mer size a priori, we adopt a data-driven strategy to determine an optimal scale $k^\ast$ for each dataset. For each candidate $k$, let $X^{(k)} \in \mathbb{R}^{n \times 4^k}$ denote the corresponding feature matrix. We evaluate three complementary criteria that capture structural, informational, and stability aspects of the representation.

First, we define the coverage
\begin{equation}
	\mathrm{cov}(k)
	=
	\frac{1}{n}\sum_{r=1}^n
	\frac{\#\{j : X^{(k)}_{rj} > 0\}}{4^k},
\end{equation}
which measures the fraction of active features. The coverage curve typically exhibits an elbow behavior, and we define the structural reference scale as
\begin{equation}
	k_{\mathrm{cov}}
	=
	\arg\max_k \left(-\Delta^2 \mathrm{cov}(k)\right),
\end{equation}
where $\Delta^2$ denotes the discrete second difference, identifying the point of maximal curvature.

Second, we consider the normalized distribution
\[
p^{(k)}_{rj}
=
\frac{X^{(k)}_{rj}}{\sum_{\ell} X^{(k)}_{r\ell}},
\]
and define the entropy
\begin{equation}
	H(k)
	=
	-\frac{1}{n}\sum_{r=1}^n \sum_{j} p^{(k)}_{rj}\log p^{(k)}_{rj}.
\end{equation}
This leads to the information-sparsity score
\begin{equation}
	I(k)
	=
	\frac{H(k)}{1 - \mathrm{cov}(k)},
\end{equation}
which balances information content against feature sparsity.

Third, we quantify the prevalence of rare features via
\begin{equation}
	\mathrm{rare}(k)
	=
	\frac{1}{n}\sum_{r=1}^n
	\frac{\#\{j : X^{(k)}_{rj} = 1\}}{\#\{j : X^{(k)}_{rj} > 0\}},
\end{equation}
and define the stability-adjusted entropy
\begin{equation}
	S(k)
	=
	H(k)\cdot\bigl(1 - \mathrm{rare}(k)\bigr),
\end{equation}
which downweights entropy contributions arising from unstable (singleton) features.

We evaluate these criteria across multiple datasets $j = 1, \dots, N$. For each dataset, we define
\begin{equation}
	k_{\mathrm{cov},j}
	=
	\arg\max_k \left(-\Delta^2 \mathrm{cov}_j(k)\right),
	\qquad
	k_{\mathrm{info},j}
	=
	\arg\max_k I_j(k),
	\qquad
	k_{\mathrm{stab},j}
	=
	\arg\max_k S_j(k).
\end{equation}

A direct comparison of these criteria does not yield a consistent ordering across datasets. In particular, the stability-adjusted entropy tends to align with the structural elbow, while the information-sparsity score systematically favors smaller values of $k$, capturing earlier informative scales. Consequently, standard consensus approaches based on agreement or rank aggregation may overemphasize structurally stable criteria and fail to reflect this early-selection behavior.

To address this, we introduce a score that evaluates each criterion relative to the coverage-based structural reference. For a given criterion $m \in \{\mathrm{cov}, \mathrm{info}, \mathrm{stab}\}$, we define
\begin{equation}
	S(m)
	=
	\frac{1}{N}\sum_{j=1}^N \bigl(k_{\mathrm{cov},j} - k_{m,j}\bigr)
	-
	\lambda \frac{1}{N}\sum_{j=1}^N \left| k_{\mathrm{cov},j} - k_{m,j} \right|,
\end{equation}
where the first term rewards earlier selections relative to the structural elbow, and the second term penalizes deviations from it. The parameter $\lambda \in (0,1)$ controls the trade-off between these two effects.

Based on the resulting ordering of the criteria induced by $S(m)$, we assign weights reflecting their relative importance across datasets. Specifically, the criteria are ranked according to their scores $S(m)$, and weights are assigned in increasing order of performance. For example, the criterion with the lowest value of $S(m)$ is assigned weight $1$, the next lowest is assigned weight $2$, and so on, so that higher-performing criteria receive larger weights. Formally, if $\operatorname{rank}(m)$ denotes the rank of criterion $m$ (with smaller $S(m)$ corresponding to lower rank), then we set
\[
w_m = \operatorname{rank}(m),
\]
ensuring that the weighting scheme reflects the relative alignment of each criterion with the structural reference across datasets. For a given dataset, the final scale $k^\ast$ is selected via a weighted consensus rule:
\begin{equation}
	k^\ast
	=
	\arg\max_k \Bigl[
	w_{\mathrm{cov}}\,\mathbf{1}\{k = k_{\mathrm{cov}}\}
	+
	w_{\mathrm{info}}\,\mathbf{1}\{k = \arg\max I(k)\}
	+
	w_{\mathrm{stab}}\,\mathbf{1}\{k = \arg\max S(k)\}
	\Bigr],
\end{equation}
where the weights $w_m$ are determined by the ranking induced by $S(m)$, and $\mathbf{1}\{\cdot\}$ denotes the indicator function.

To obtain a global scale across datasets, we extend this rule by aggregating the weighted votes over all datasets. Specifically, we select $k^\ast$ as the maximizer of the total weighted support across datasets and criteria:
\begin{equation}
	k^\ast = \arg\max_k \sum_{j=1}^N \sum_{m \in \{\mathrm{cov},\,\mathrm{info},\,\mathrm{stab}\}} w_m \,\mathbf{1}\{k = k_{m,j}\}.
\end{equation}

\clearpage
\begin{figure}[!t]
	\centering
	\includegraphics[width=0.75\textwidth]{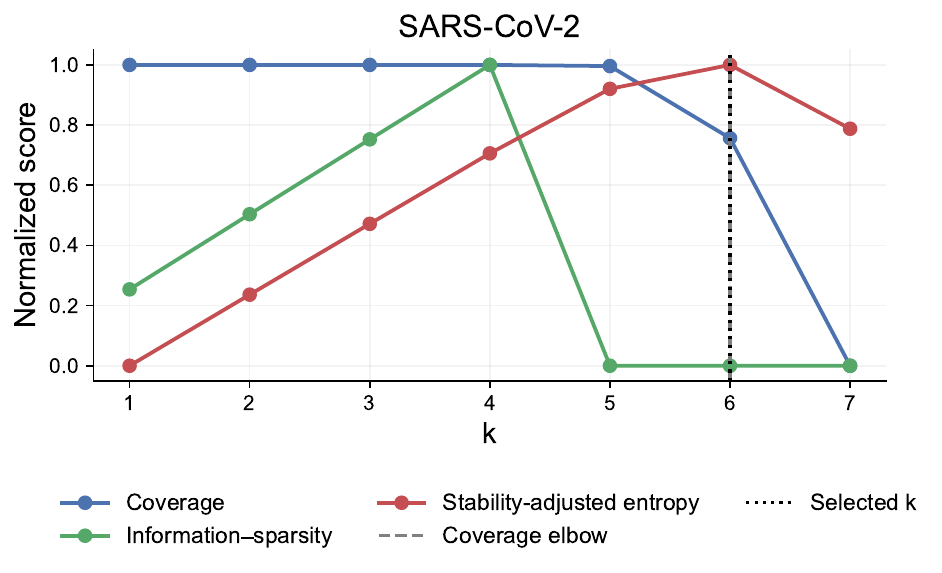}
	\caption{
		Data-driven selection of the $k$-mer scale for the SARS-CoV-2 dataset.
		Normalized coverage, information-sparsity, and stability-adjusted entropy are shown as functions of $k$.
		The dashed vertical line indicates the coverage elbow, while the dotted vertical line denotes the selected scale $k^\ast$. The analysis used $n=44$ SARS-CoV-2 sequences. Each marker represents a normalized score calculated using all sequences at the indicated $k$. No error bars or statistical tests were used. Source data are provided as a Source Data file.}
	\label{fig:k_selection_sars}
\end{figure}

\begin{figure}[!t]
	\centering
	\includegraphics[width=0.75\textwidth]{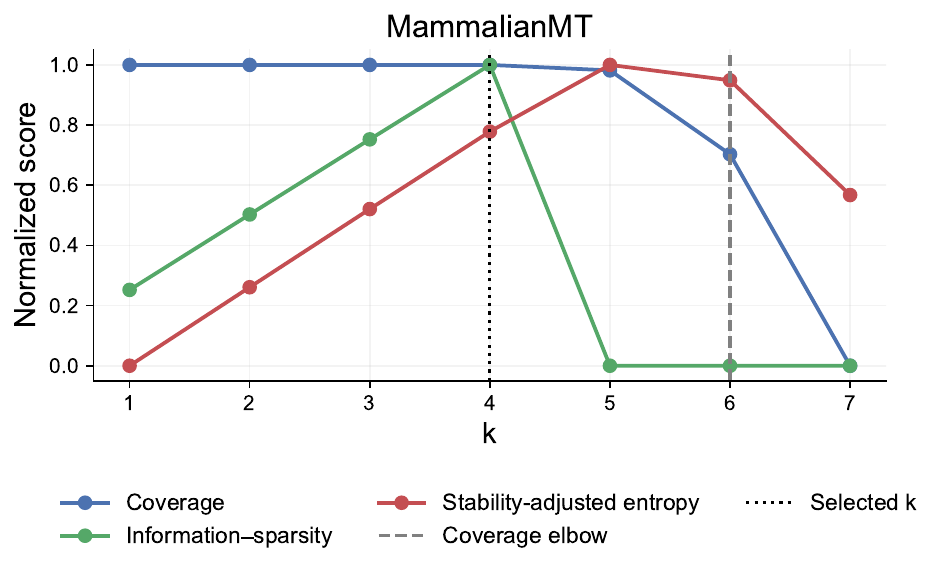}
	\caption{
		Data-driven selection of the $k$-mer scale for the mammalian mitochondrial dataset.
		Normalized coverage, information-sparsity, and stability-adjusted entropy are shown as functions of $k$.
		The dashed vertical line indicates the coverage elbow, while the dotted vertical line denotes the selected scale $k^\ast$. The analysis used $n=41$ mammalian mitochondrial sequences. Each marker represents a normalized score calculated using all sequences at the indicated $k$. No error bars or statistical tests were used. Source data are provided as a Source Data file.}
	\label{fig:k_selection_mammalian}
\end{figure}

\begin{figure}[t!]
	\centering
	\includegraphics[width=0.75\textwidth]{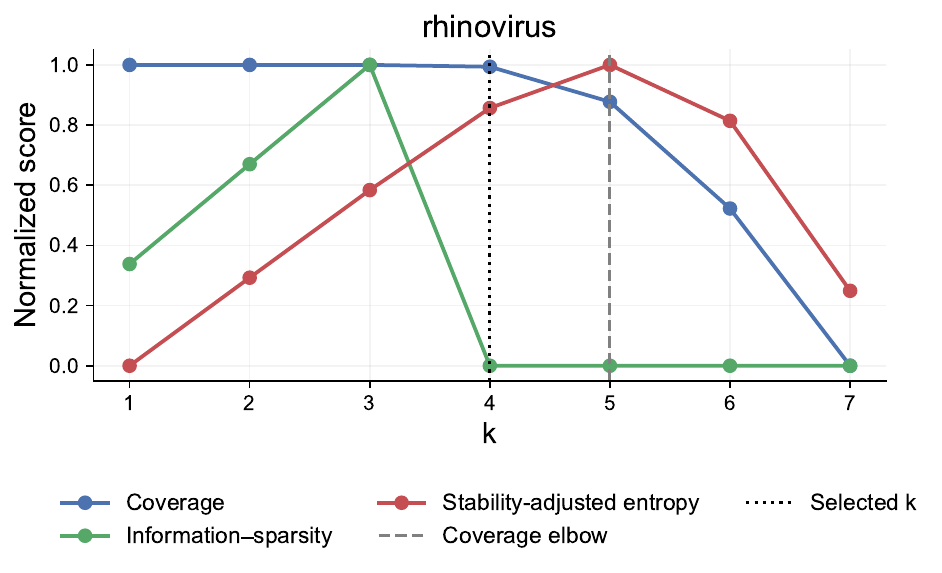}
	\caption{
		Data-driven selection of the $k$-mer scale for the rhinovirus dataset.
		Normalized coverage, information-sparsity, and stability-adjusted entropy are shown as functions of $k$.
		The dashed vertical line indicates the coverage elbow, while the dotted vertical line denotes the selected scale $k^\ast$. The analysis used $n=116$ rhinovirus sequences. Each marker represents a normalized score calculated using all sequences at the indicated $k$. No error bars or statistical tests were used. Source data are provided as a Source Data file.}
	\label{fig:k_selection_rhino}
\end{figure}

\clearpage

\section{Robustness to incomplete and fragmented sequence inputs}

To further examine the practical applicability of CAKR beyond ideal full-length sequences, we evaluated its behavior under two complementary incomplete-sequence settings. First, we assessed clustering robustness under simulated sequence degradation, including truncation, contiguous deletion, and fragmentation, in order to mimic partial recovery and draft-quality assemblies. Second, we considered a fragment placement setting, in which an external partial query sequence was compared against a reference barcode library to determine its closest viral type and homologous gene segment. Together, these experiments provide two distinct perspectives on robustness: stability of clustering under incomplete inputs and the ability to identify the origin of previously unseen sequence fragments in an alignment-free manner.

\subsection{Robustness under simulated incompleteness}
To further evaluate robustness under more realistic sequencing and assembly conditions, we performed an incomplete-sequence experiment in which each benchmark dataset was systematically degraded to mimic partial recovery, missing regions, and draft-quality fragmentation. Starting from the original sequences, we generated modified inputs by varying three factors: the retained fraction of the sequence, the size of a contiguous deleted block, and the number of fragments. Specifically, retained fractions of $1.0$, $0.8$, $0.6$, and $0.4$ were considered; contiguous block deletions of $0.1$ and $0.2$ were applied; and the remaining sequence was split into $1$, $5$, or $10$ fragments. The resulting fragments were then concatenated to produce incomplete assembly-like inputs.

Overall, the results show that the method is largely robust to incomplete-sequence effects for most benchmark datasets. In particular, Ebolavirus, rhinovirus, HEV, and influenza HA gene retain strong clustering purity across a wide range of degradation settings. For Ebolavirus, many conditions achieve perfect purity, indicating near-complete robustness even under substantial incompleteness. Rhinovirus and HEV also remain highly stable, with many settings reaching or approaching purity $1.0$. Influenza HA gene likewise performs strongly overall, although the most severe incomplete settings lead to some reduction in purity.

By contrast, mammalian mitochondrial genomes are more sensitive to degradation. Although the method still yields reasonably strong purity values in many cases, the performance is more variable and consistently lower than that observed for the viral datasets above. This suggests that the method remains applicable, but that robustness is more moderate for more heterogeneous and challenging collections.

A markedly different pattern is observed for SARS-CoV-2. Across nearly all tested conditions and across all examined $k$ values, purity remains low with only minimal variation. This indicates that the difficulty on this dataset is not primarily driven by incompleteness alone, but rather reflects the intrinsically harder discrimination problem posed by these closely related sequences.

A clear global trend is that retention has the strongest effect on performance. In general, the most severe retention level, namely $0.4$, produces the weakest results, while increasing the retained fraction to $0.6$, $0.8$, and $1.0$ typically improves purity substantially. In contrast, increasing the deletion fraction from $0.1$ to $0.2$ usually causes only a mild deterioration, and increasing the number of fragments from $1$ to $5$ or $10$ does not produce a consistent collapse in performance. This suggests that the dominant factor is the total amount of sequence information preserved, rather than fragmentation alone.

The experiment also shows that moderate-to-large $k$ values often provide the best robustness. For several datasets, the highest purity is most frequently achieved at $k=5$, $6$, or $7$. This is especially evident for mammalian mitochondrial genomes, rhinovirus, and Ebolavirus. In contrast, influenza HA gene and HEV exhibit somewhat greater flexibility across different $k$ values, while SARS-CoV-2 shows no meaningful gain from increasing $k$.

Taken together, these results demonstrate that the proposed representation does not rely on perfectly assembled full genomes. For most datasets, strong clustering quality is preserved even after substantial truncation, deletion of contiguous regions, and fragmentation into multiple contigs. This supports the claim that the method is robust to partial or draft-quality genomic inputs. At the same time, the dataset-dependent behavior, especially on SARS-CoV-2, shows that robustness to incompleteness does not eliminate intrinsic classification difficulty in very closely related sequence collections.

Since the current experiment includes one randomized replicate per condition, these findings should be interpreted primarily as strong robustness trends. Additional replicates would further strengthen the statistical support for these observations.

\begin{figure}[h!]
	\centering
	\includegraphics[width=0.65\linewidth]{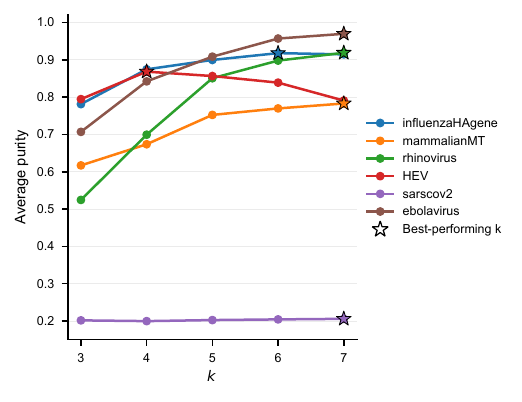}
	\caption{\textbf{Average clustering purity as a function of k under simulated incomplete-sequence conditions.} For each dataset and k, the plotted value is the arithmetic mean over $n=24$ degradation conditions, comprising four retained fractions, two contiguous-deletion fractions and three fragmentation levels. One randomized realization was generated per condition. Dataset sizes were $n=30$ influenza HA sequences, $n=41$ mammalian mitochondrial sequences, $n=116$ sequences for the rhinovirus analysis, $n=48$ HEV sequences, $n=44$ SARS-CoV-2 sequences and $n=59$ Ebolavirus sequences. Stars indicate the best-performing $k$ for each dataset. No error bars or statistical tests were used. Source data are provided as a Source Data file.}
	\label{fig:incomplete_avg_purity_vs_k}
\end{figure}

\subsection{Barcode-based fragment placement of external influenza queries}

Fragment phylogenetic placement seeks to infer the evolutionary origin of a sequence when only a partial gene or genomic fragment is available rather than a complete genome. In such cases, the objective is to determine which known virus and homologous gene segment the fragment is most closely related to. To investigate whether CAKR can perform this task in an alignment-free manner, we designed a proof-of-concept study based on barcode-level comparison. A reference barcode library was constructed from known influenza viral sequences, with one barcode computed for each reference genome segment. The barcode of an external query fragment was then compared against this library to identify its closest viral type and gene segment. As a model system, we used influenza viruses obtained from NIH-NCBI Virus. The reference library included one representative strain from each influenza type and all genome segments available for that type, namely eight segments for influenza A and B and seven segments for influenza C and D. Specifically, the reference set comprised human influenza A(H3N2) and influenza B isolates collected in California in 2020 and 2021, respectively, a human influenza C isolate collected in China in 2020, and a swine influenza D isolate collected in Oklahoma in 2011. All reference segments were annotated as complete coding sequences.

To evaluate the extent to which CAKR can resolve unknown influenza fragments through barcode-based comparison, we computed barcode representations for 14 external influenza query genes within the CAKR framework and compared these against the barcode library established from the genome segments of the reference strains. The correspondence between query and reference barcodes was quantified by cosine similarity, providing a direct measure of whether a previously unseen query could be correctly associated with its influenza type and homologous gene segment without reliance on sequence alignment. The query cohort comprised 11 complete coding sequences and 3 partial coding sequences and was deliberately selected to include variation in collection year, geographic source, and host relative to the reference set. Specifically, influenza A and B queries were drawn from human isolates collected in 2026, influenza C queries from human isolates collected in 2024 and 2025, and influenza D queries from bovine isolates collected in 2024.

As shown in Source Data file, each of the 14 external influenza query gene sequences attained its maximum cosine similarity with the correct viral type and the corresponding homologous gene segment in the reference library. These findings indicate that the CAKR derived barcode representation preserves both type specific and gene segment specific information, even when the query and reference sequences differ with respect to collection year, geographic origin, host background, and coding sequence completeness. This behavior is illustrated by representative cases from all four influenza types. For influenza A, fragmented query PZ251744.1, corresponding to the PB1 gene segment from a human A(H3N2) isolate collected in California in 2026, exhibited its strongest similarity to the influenza A PB1 gene segment reference PX229455.1, derived from a human A(H3N2) isolate collected in California in 2020, with a cosine similarity of 0.967943. For influenza B, query PZ256275.1, representing the NA gene segment from a human isolate collected in New York in 2026 and annotated as a complete coding sequence, was most similar to the influenza B NA gene segment reference OM730287.1, obtained from a human isolate collected in California in 2021, with a cosine similarity of 0.978296. For influenza C, query PV177458.1, corresponding to the P3 gene segment from a human isolate collected in Massachusetts in 2024 and annotated as a complete coding sequence, was correctly assigned to the influenza C P3 gene segment reference OQ733377.1, derived from a human isolate collected in China in 2020, with a cosine similarity of 0.980462. For influenza D, query PX485351.1, representing the PB1 gene segment from a bovine isolate collected in Mexico in 2024 and annotated as a complete coding sequence, exhibited its highest similarity to the influenza D PB1 gene segment reference NC036615.1, derived from a swine isolate collected in Oklahoma in 2011, with a cosine similarity of 0.958346. Particularly in the influenza D case, successful retrieval despite differences in host species, sampling location, and collection year highlights the robustness of the proposed barcode representation for fragment placement style inference. Taken together, these toy demonstration results support the potential of CAKR for alignment free fragment placement for identifying the closest viral type and the corresponding homologous gene segment from a reference library.

\subsection{Comparative performance of CAKR and SEPP on  phylogenetic placement  }

Phylogenetic placement seeks to infer the evolutionary position of a query sequence by inserting it onto a fixed reference alignment and backbone tree. SEPP is a widely used framework for this task. It uses a divide and conquer strategy with an ensemble of hidden Markov models (HMMs), implemented through HMMER, to align fragmentary query sequences to the backbone alignment, and then uses pplacer to insert each aligned fragment into the fixed reference tree. Placement confidence is commonly summarized using the likelihood weight ratio.

To compare CAKR with a standard HMM-based placement approach for fragmentary sequences, we evaluated SEPP \cite{mirarab2012sepp} on the same Type-B influenza gene segment 1 (PB1) setting used for CAKR. The backbone reference set consisted of 15 Type-B PB1 genes, mainly from 2019 to 2021. Fifteen fragmentary queries were generated from five randomly selected 2026 sequences by truncating each sequence into 300 bp, 600 bp, and 900 bp fragments. The reference sequences, aligned by MAFFT, were used to construct the backbone multiple sequence alignment and backbone tree, and SEPP was then run on the fragment queries using this fixed reference structure.

The results, summarized in Source Data file (CAKR vs SEPP sheet), show substantial agreement between the two frameworks at the level of ranked candidate placements. Exact Top 1 agreement was observed for 6 of 15 fragmentary queries. In addition, 11 of 15 fragments showed overlap between the top three CAKR candidates and the top two SEPP placements listed in the table, indicating that the two methods frequently identify related leading placement candidates even when their ranking order differs. Because the two methods use different scoring systems, cosine similarity for CAKR and likelihood weight ratio for SEPP, the raw score magnitudes are not directly comparable. The more meaningful comparison lies in the agreement of their ranked placements and in how confidently each framework separates its leading candidates.

CAKR also showed a clear dependence on fragment length. Its Top 1 cosine similarity increased from an average of approximately 0.5335 for 300 bp fragments to 0.6855 for 600 bp fragments and 0.7878 for 900 bp fragments, indicating that longer fragments carry stronger alignment free signal for barcode based comparison. Across all 15 fragments, CAKR Top 1 cosine scores ranged from 0.5084 to 0.7982. Moreover, the relative difference between the first and second CAKR candidates remained consistently small, ranging from about 0.01\% to 0.76\% in the displayed table. By contrast, the corresponding SEPP T1 to T2 difference ranged from 0.0004\% to 96.7054\%, showing that SEPP sometimes distributed support across several nearby placements and sometimes strongly favored a single placement. Taken together, these results suggest that CAKR provides stable and competitive placement behavior relative to a standard HMM-based framework while retaining the practical advantage of being fully alignment free.

\clearpage

\section{Sensitivity to sequence composition biases}

To assess whether the observed performance is driven primarily by simple compositional biases rather than higher-order sequence organization, we performed a bias-sensitivity control experiment using several targeted perturbations. For each dataset, we generated modified versions of the original sequences by (i) masking low-complexity regions, (ii) masking repeat-dominated regions, (iii) normalizing GC content, (iv) applying mononucleotide shuffling, and (v) applying dinucleotide-preserving shuffling. These perturbations were designed to test the influence of low-complexity tracts, repetitive regions, global compositional bias, and short-range local dependencies on clustering performance.

For clarity, the perturbation tests may be interpreted as follows:

\begin{itemize}
	\item \textbf{Original.}
	This is the baseline. Compare every other test to this.
	If a modified version stays close to the original, that factor is likely not the main source of performance.
	If it drops strongly, that factor is important.

	\item \textbf{Low-complexity masked.}
	Read it by comparing its purity to the original.
	If it stays nearly the same, low-complexity regions are not driving the method.
	If it drops clearly, the method is using low-information sequence stretches.
	Usually, one expects only a small change unless those regions carry important signal.

	\item \textbf{Repeat-masked.}
	Compare with the original.
	If performance remains similar, repeats are not the main reason for success.
	If performance drops, repeat-rich or highly homogeneous regions contribute useful signal.
	Usually, a small or moderate drop means repeats help somewhat but are not dominant.

	\item \textbf{Mononucleotide shuffle.}
	This is a strong test. Since base counts are preserved but order is destroyed, a strong drop means the method depends on sequence arrangement, not only composition.
	If performance stays high, the method may be relying mainly on simple nucleotide frequencies.
	Usually, one expects a clear drop if the method captures real sequence structure.

	\item \textbf{Dinucleotide shuffle.}
	Compare it both to the original and to mononucleotide shuffle.
	If it drops strongly, the method uses information beyond short-range dinucleotide structure.
	If it stays much better than mononucleotide shuffle, then local adjacency information explains part of the signal.
	Usually, one expects performance to lie between the original and mononucleotide shuffle.

	\item \textbf{GC-normalized.}
	Compare with the original.
	If performance stays similar, GC bias is not the main driver.
	If it drops clearly, GC-content differences were contributing strongly to discrimination.
	Usually, a mild drop suggests GC matters somewhat, but is not the full explanation.
\end{itemize}

Overall, the results reveal a consistent pattern across datasets. Masking low-complexity regions has little or no effect in most cases, and repeat masking usually causes only limited or moderate degradation. For example, influenza HA gene, SARS-CoV-2, and Ebolavirus remain essentially unchanged after low-complexity masking, while rhinovirus also retains very strong performance. Repeat masking likewise leaves influenza HA gene and Ebolavirus nearly unchanged, and preserves strong performance in rhinovirus and SARS-CoV-2, although somewhat larger reductions are observed for HEV and mammalian mitochondrial genomes. These findings suggest that low-complexity and repeat-rich regions are generally not the dominant source of discriminatory signal.

GC normalization also preserves much of the original performance in most datasets. In particular, influenza HA gene, HEV, mammalian mitochondrial genomes, rhinovirus, and Ebolavirus remain relatively strong after GC normalization, especially at moderate or larger values of $k$. This indicates that global GC-content differences alone do not explain the observed clustering quality, although compositional effects may still contribute to some extent in certain datasets.

In contrast, mononucleotide shuffling produces a substantial decline in purity for most datasets, including influenza HA gene, HEV, rhinovirus, SARS-CoV-2, and Ebolavirus, and also reduces performance in mammalian mitochondrial genomes. Since mononucleotide shuffling preserves overall nucleotide composition while destroying positional order, this result strongly indicates that the method depends on sequence arrangement rather than merely on base frequencies. Dinucleotide-preserving shuffling yields a more nuanced but still informative pattern. In influenza HA gene, HEV, rhinovirus, and Ebolavirus, it also causes a clear decrease in performance, particularly at larger values of $k$, suggesting that the method captures information beyond simple dinucleotide statistics. By contrast, mammalian mitochondrial genomes are more stable under dinucleotide shuffling, indicating that lower-order local dependencies may explain a larger fraction of its signal.

Taken together, these results show that the proposed representation is generally not driven by low-complexity regions, simple repeats, or global GC bias. Instead, the strong degradation under mononucleotide and, in several datasets, dinucleotide shuffling indicates that the method relies predominantly on preserved sequence order and higher-order structural organization. Thus, the perturbation analysis supports the interpretation that the method captures biologically meaningful sequence structure rather than merely exploiting simple compositional biases.

\clearpage

\begin{figure}[t]
	\centering
	\includegraphics[width=\linewidth]{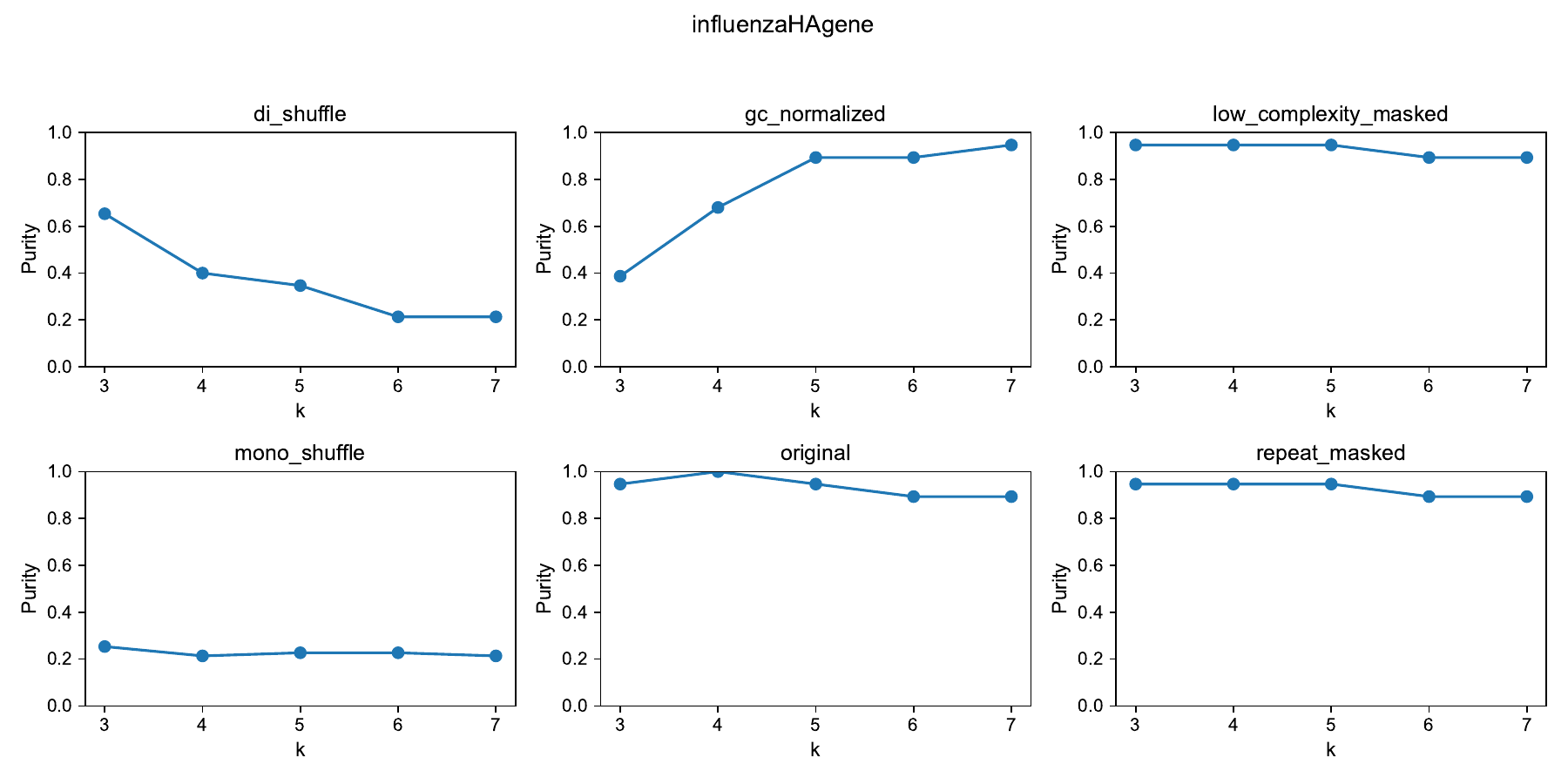}
	\caption{\textbf{Clustering-purity curves for the influenza HA gene dataset under sequence-bias perturbation tests.} The panels compare the average purity across $k$ for the original sequences and the perturbed variants, including low-complexity masking, repeat masking, GC normalization, mononucleotide shuffling, and dinucleotide-preserving shuffling. These results illustrate how sequence-bias controls affect the performance of the method on the influenza HA gene dataset. Source data are provided as a Source Data file.}
	\label{fig:grid_curves_influenzaHAgene}
\end{figure}

\begin{figure}[t]
	\centering
	\includegraphics[width=\linewidth]{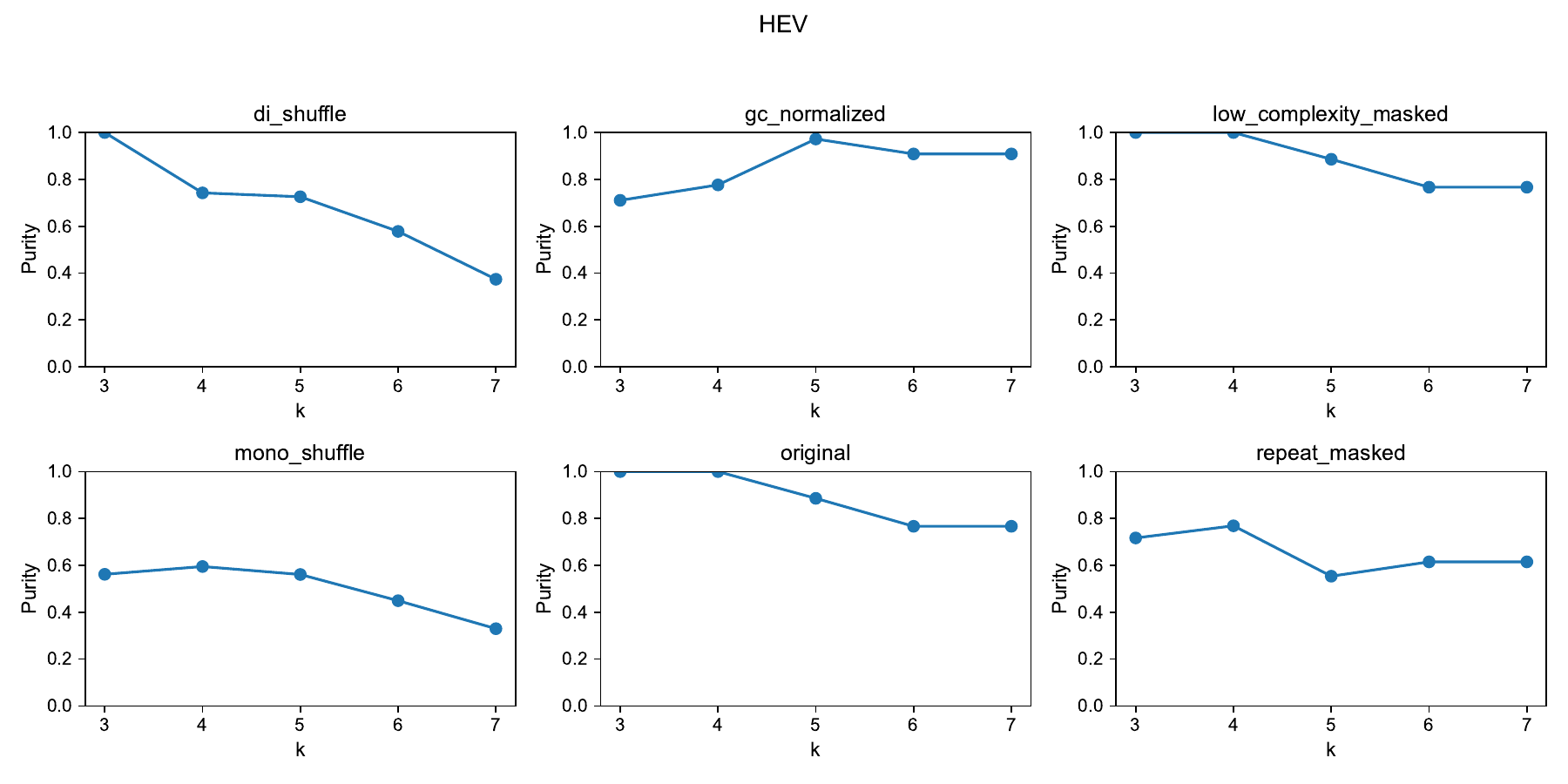}
	\caption{\textbf{Clustering-purity curves for the HEV dataset under sequence-bias perturbation tests.} The panels compare the average purity across $k$ for the original sequences and the perturbed variants, including low-complexity masking, repeat masking, GC normalization, mononucleotide shuffling, and dinucleotide-preserving shuffling. These results illustrate how sequence-bias controls affect the performance of the method on the HEV dataset. Source data are provided as a Source Data file.}
	\label{fig:grid_curves_HEV}
\end{figure}

\begin{figure}[t]
	\centering
	\includegraphics[width=\linewidth]{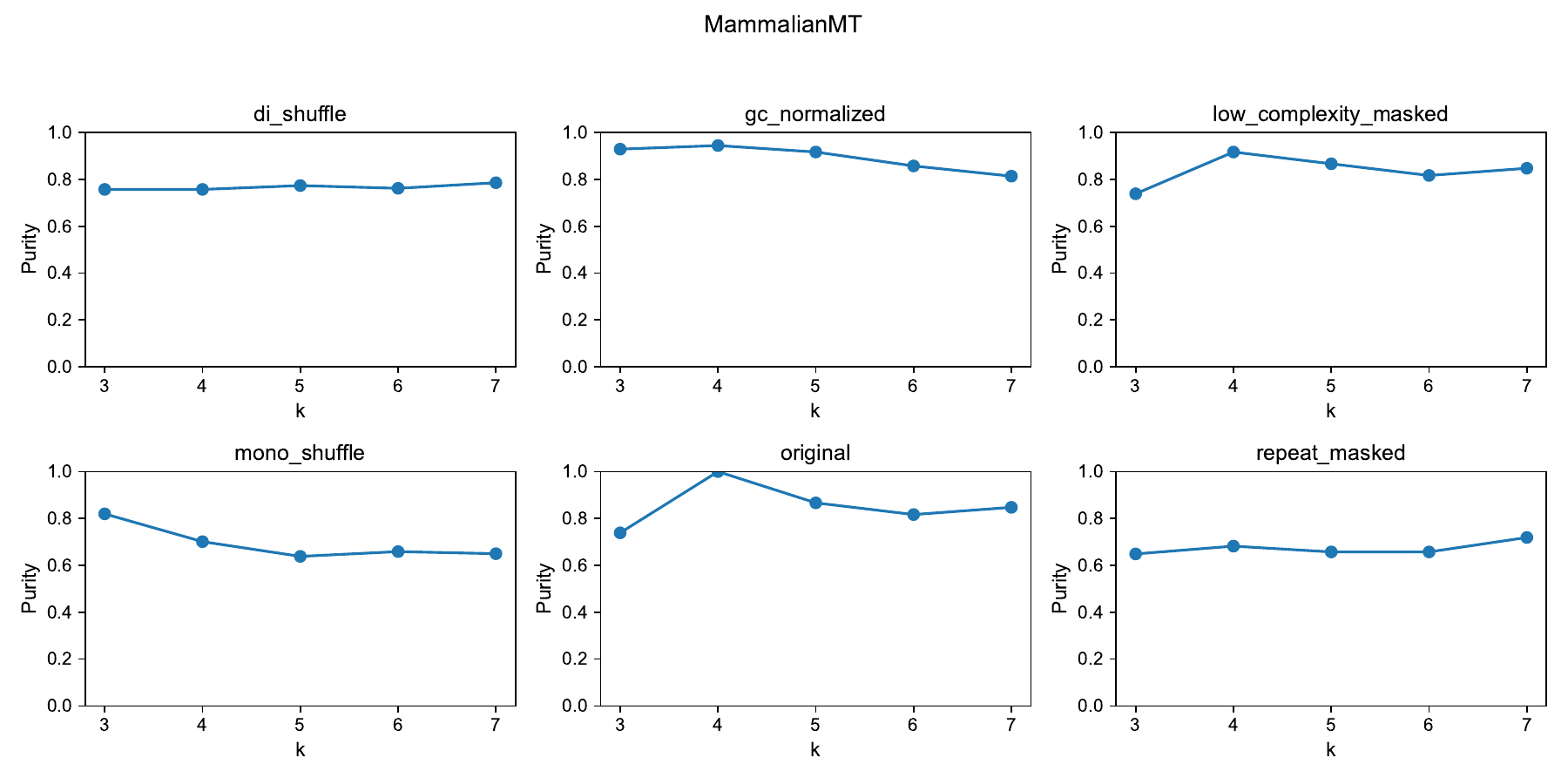}
	\caption{\textbf{Clustering-purity curves for the mammalian mitochondrial genome dataset under sequence-bias perturbation tests.} The panels compare the average purity across $k$ for the original sequences and the perturbed variants, including low-complexity masking, repeat masking, GC normalization, mononucleotide shuffling, and dinucleotide-preserving shuffling. These results illustrate how sequence-bias controls affect the performance of the method on the mammalian mitochondrial genome dataset. Source data are provided as a Source Data file.}
	\label{fig:grid_curves_mammalianMT}
\end{figure}

\begin{figure}[t]
	\centering
	\includegraphics[width=\linewidth]{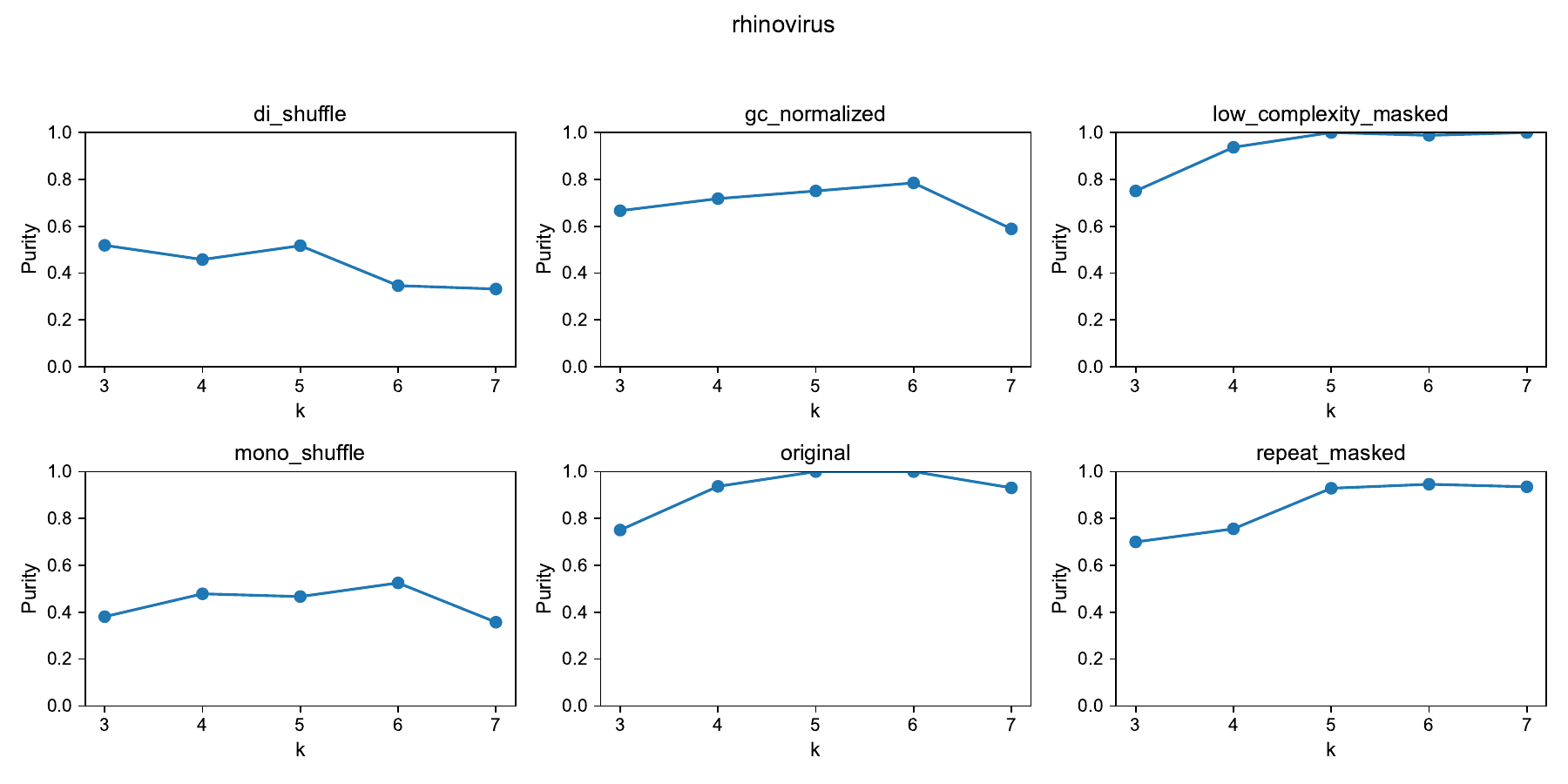}
	\caption{\textbf{Clustering-purity curves for the rhinovirus dataset under sequence-bias perturbation tests.} The panels compare the average purity across $k$ for the original sequences and the perturbed variants, including low-complexity masking, repeat masking, GC normalization, mononucleotide shuffling, and dinucleotide-preserving shuffling. These results illustrate how sequence-bias controls affect the performance of the method on the rhinovirus dataset. Source data are provided as a Source Data file.}
	\label{fig:grid_curves_rhinovirus}
\end{figure}

\begin{figure}[t]
	\centering
	\includegraphics[width=\linewidth]{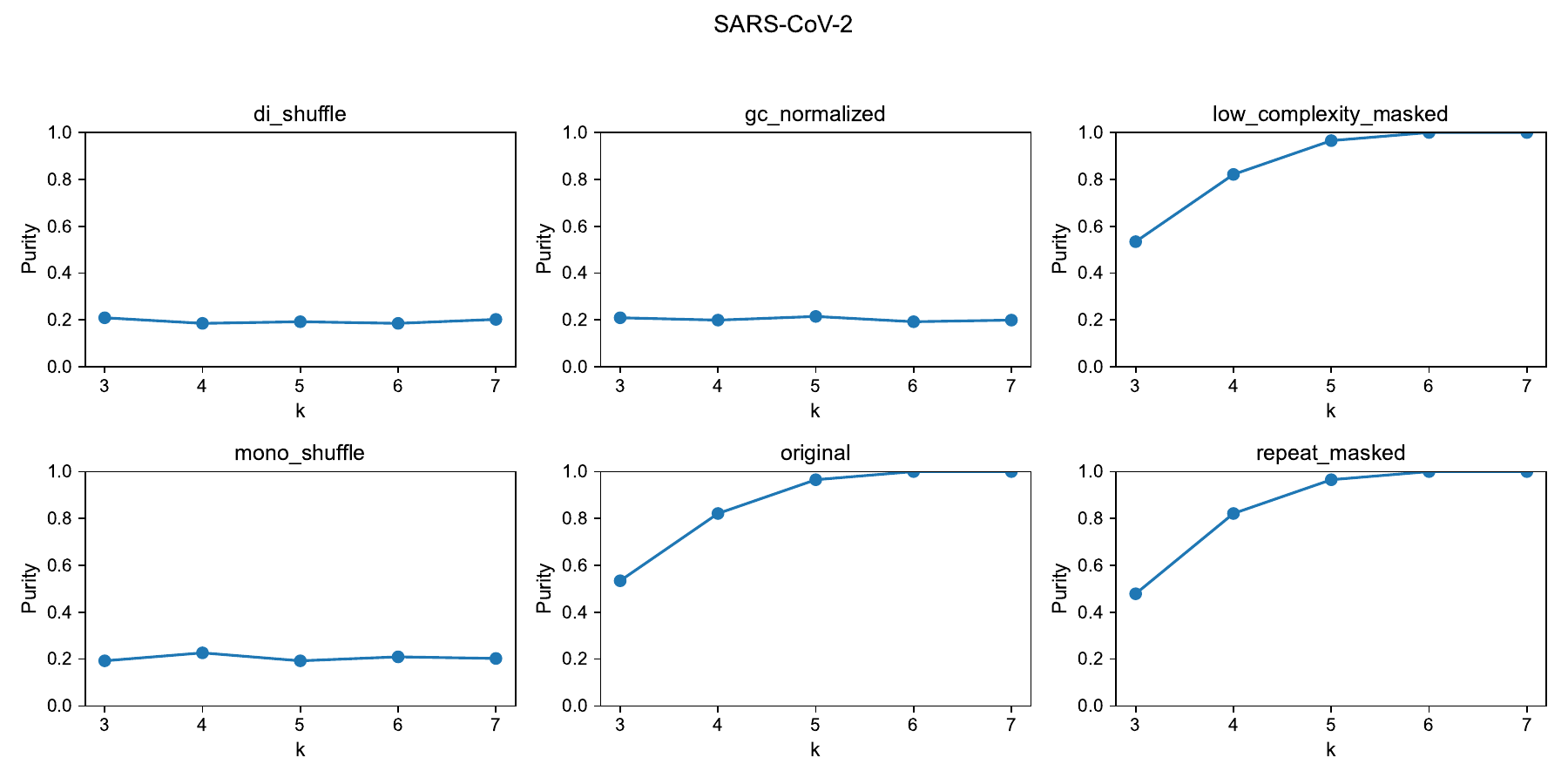}
	\caption{\textbf{Clustering-purity curves for the SARS-CoV-2 dataset under sequence-bias perturbation tests.} The panels compare the average purity across $k$ for the original sequences and the perturbed variants, including low-complexity masking, repeat masking, GC normalization, mononucleotide shuffling, and dinucleotide-preserving shuffling. These results illustrate how sequence-bias controls affect the performance of the method on the SARS-CoV-2 dataset. Source data are provided as a Source Data file.}
	\label{fig:grid_curves_sarscov2}
\end{figure}

\begin{figure}[t]
	\centering
	\includegraphics[width=\linewidth]{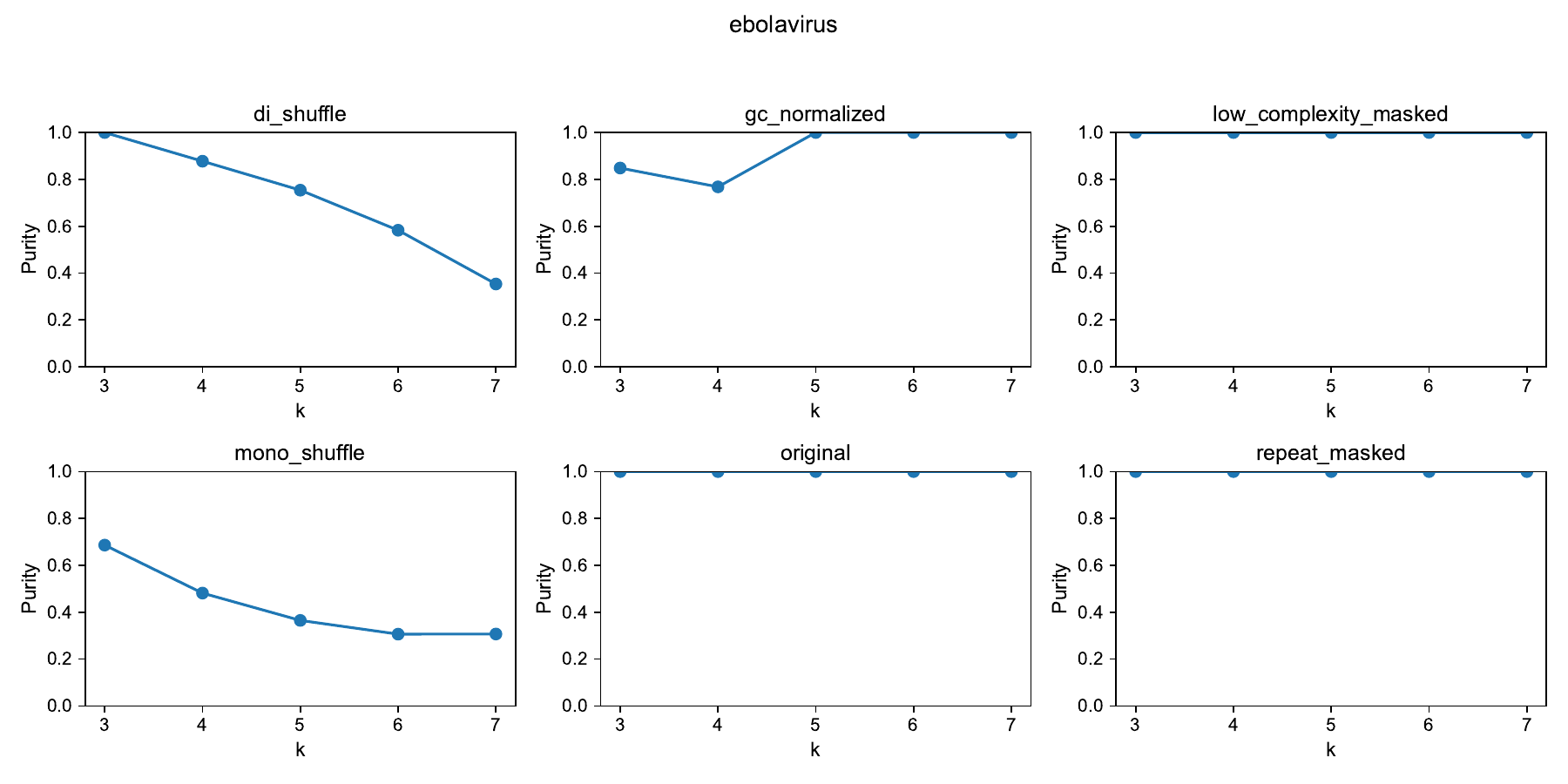}
	\caption{\textbf{Clustering-purity curves for the Ebolavirus dataset under sequence-bias perturbation tests.} The panels compare the average purity across $k$ for the original sequences and the perturbed variants, including low-complexity masking, repeat masking, GC normalization, mononucleotide shuffling, and dinucleotide-preserving shuffling. These results illustrate how sequence-bias controls affect the performance of the method on the Ebolavirus dataset. Source data are provided as a Source Data file.}
	\label{fig:grid_curves_ebolavirus}
\end{figure}

\clearpage

\clearpage
\begin{table}[ht]
	\centering
	\caption{SARS-CoV-2 genome accessions and corresponding variants.}
	\label{tab:sars_variants}
	\begin{tabular}{llllll}
		\toprule
		Accession & Variant & Accession & Variant & Accession & Variant \\
		\midrule
		EPI\_ISL\_10265849 & Alpha & EPI\_ISL\_10451166 & Alpha & EPI\_ISL\_10898816 & Alpha \\
		EPI\_ISL\_10935285 & Alpha & EPI\_ISL\_10953997 & Alpha & EPI\_ISL\_11230282 & Alpha \\

		EPI\_ISL\_10113812 & Beta  & EPI\_ISL\_10932695 & Beta  & EPI\_ISL\_6963122  & Beta \\
		EPI\_ISL\_8482131  & Beta  & EPI\_ISL\_8582087  & Beta  & EPI\_ISL\_10999595 & Delta \\

		EPI\_ISL\_11254301 & Delta & EPI\_ISL\_11320247 & Delta & EPI\_ISL\_11533460 & Delta \\
		EPI\_ISL\_8082681  & Delta & EPI\_ISL\_8582067  & Delta & EPI\_ISL\_10515204 & Gamma \\

		EPI\_ISL\_1708575  & Gamma & EPI\_ISL\_4713792  & Gamma & EPI\_ISL\_7961373  & Gamma \\
		EPI\_ISL\_9547950  & Gamma & EPI\_ISL\_9800275  & Gamma & EPI\_ISL\_6316493  & GH/490R \\

		EPI\_ISL\_7547548  & GH/490R & EPI\_ISL\_8189639 & GH/490R & EPI\_ISL\_9028625 & GH/490R \\
		EPI\_ISL\_9077869  & GH/490R & EPI\_ISL\_2097767 & Lambda & EPI\_ISL\_2862939 & Lambda \\

		EPI\_ISL\_2893886  & Lambda & EPI\_ISL\_3397033 & Lambda & EPI\_ISL\_2651120 & Mu \\
		EPI\_ISL\_2933815  & Mu & EPI\_ISL\_6507611 & Mu & EPI\_ISL\_6783400 & Mu \\
		EPI\_ISL\_6990137  & Mu & EPI\_ISL\_8788228 & Mu & EPI\_ISL\_11221245 & Omicron \\

		EPI\_ISL\_11222860 & Omicron & EPI\_ISL\_11268271 & Omicron & EPI\_ISL\_11370100 & Omicron \\
		EPI\_ISL\_11449733 & Omicron & EPI\_ISL\_11531492 & Omicron & & \\
		\bottomrule
	\end{tabular}
\end{table}

\begin{table}[ht]
	\centering
	\caption{GenBank accession numbers of mammalian mitochondrial genomes and their corresponding taxonomic orders used for phylogenetic reconstruction.}
	\label{tab:mammal_mito_accession_order}
	\begin{tabular}{llll}
		\toprule
		Accession (version) & Order & Accession (version) & Order \\
		\midrule
		V00662.1     & Primates        & D38116.1     & Primates \\
		D38113.1     & Primates        & D38114.1     & Primates \\
		X99256.1     & Primates        & Y18001.1     & Primates \\
		AY863426.1   & Primates        & D38115.1     & Primates \\
		NC\_002083.1 & Primates        & NC\_002764.1 & Primates \\
		U20753.1     & Carnivore       & U96639.2     & Carnivore \\
		EU442884.2   & Carnivore       & EF551003.1   & Carnivore \\
		EF551002.1   & Carnivore       & DQ402478.1   & Carnivore \\
		AF303110.1   & Carnivore       & AF303111.1   & Carnivore \\
		EF212882.1   & Carnivore       & AJ002189.1   & Artiodactyla \\
		AF010406.1   & Artiodactyla    & AF533441.1   & Artiodactyla \\
		V00654.1     & Artiodactyla    & AY488491.1   & Artiodactyla \\
		NC\_007441.1 & Artiodactyla    & NC\_008830.1 & Artiodactyla \\
		NC\_010640.1 & Artiodactyla    & X72204.1     & Cetacea \\
		NC\_005268.1 & Cetacea         & NC\_001321.1 & Cetacea \\
		NC\_005270.1 & Cetacea         & NC\_005275.1 & Cetacea \\
		NC\_006931.1 & Cetacea         & NC\_001788.1 & Perissodactyla \\
		X97336.1     & Perissodactyla  & Y07726.1     & Perissodactyla \\
		NC\_001640.1 & Perissodactyla  & AJ238588.1   & Rodentia \\
		AJ001562.1   & Rodentia        & AJ001588.1   & Lagomorpha \\
		X88898.2     & Erinaceomorpha  &              &  \\
		\bottomrule
	\end{tabular}
\end{table}

\begin{table}[ht]
	\centering
	\caption{Hepatitis E virus (HEV) genome accessions and corresponding genotypes.}
	\label{tab:hev_accessions}
	\begin{tabular}{llllll}
		\toprule
		Accession & Genotype & Accession & Genotype & Accession & Genotype \\
		\midrule
		M73218.1 & G1 & D10330.1 & G1 & X99441.1 & G1 \\
		AF076239.3 & G1 & AF051830.1 & G1 & AF185822.1 & G1 \\
		AF459438.1 & G1 & D11092.1 & G1 & L25595.1 & G1 \\
		L08816.1 & G1 & D11093.1 & G1 & M94177.1 & G1 \\
		M80581.1 & G1 & X98292.1 & G1 & AY230202.1 & G1 \\
		AY204877.1 & G1 & M74506.1 & G2 & AB089824.1 & G3 \\
		AB074918.2 & G3 & AB074920.3 & G3 & AF082843.1 & G3 \\
		AF060668.1 & G3 & AF060669.1 & G3 & AY115488.1 & G3 \\
		AB189070.1 & G3 & AB189071.1 & G3 & AB091394.1 & G3 \\
		AB189072.1 & G3 & AP003430.1 & G3 & AB189073.1 & G3 \\
		AB189074.1 & G3 & AB189075.1 & G3 & AB073912.1 & G3 \\
		AF455784.1 & G3 & AB097812.1 & G4 & AB099347.1 & G4 \\
		AB080575.1 & G4 & AB074915.3 & G4 & AB074917.3 & G4 \\
		AB161717.1 & G4 & AB091395.1 & G4 & AB200239.1 & G4 \\
		AB161718.1 & G4 & AB161719.1 & G4 & AB097811.1 & G4 \\
		AY594199.1 & G4 & AJ272108.1 & G4 & AB108537.1 & G4 \\
		\bottomrule
	\end{tabular}
\end{table}

\begin{table}[ht]
	\centering
	\caption{GenBank accession numbers of Human rhinovirus (HRV-A, HRV-B, HRV-C) genomes with Human enterovirus (HEV) sequences included as an outgroup for phylogenetic tree reconstruction.}
	\label{tab:hrv_hev_groups}
	\begin{tabular}{llllll}
		\toprule
		Accession & Group & Accession & Group & Accession & Group \\
		\midrule
		AF499637.1 & HEV & AF546702.1 & HEV & V01149.1   & HEV \\
		AY751783.1 & A   & DQ473491.1 & A   & EF173414.1 & A   \\
		DQ473492.1 & A   & DQ473493.1 & A   & DQ473494.1 & A   \\
		DQ473496.1 & A   & DQ473497.1 & A   & DQ473499.1 & A   \\
		DQ473500.1 & A   & DQ473504.1 & A   & DQ473505.1 & A   \\
		DQ473506.1 & A   & DQ473507.1 & A   & DQ473508.1 & A   \\
		DQ473510.1 & A   & DQ473511.1 & A   & EF173415.1 & A   \\
		FJ445111.1 & A   & FJ445113.1 & A   & FJ445114.1 & A   \\
		FJ445115.1 & A   & FJ445116.1 & A   & FJ445117.1 & A   \\
		FJ445118.1 & A   & FJ445119.1 & A   & FJ445121.1 & A   \\
		FJ445122.1 & A   & FJ445123.1 & A   & FJ445125.1 & A   \\
		FJ445126.1 & A   & FJ445127.1 & A   & FJ445128.1 & A   \\
		FJ445129.1 & A   & FJ445131.1 & A   & FJ445132.1 & A   \\
		FJ445133.1 & A   & FJ445134.1 & A   & FJ445135.1 & A   \\
		FJ445136.1 & A   & FJ445138.1 & A   & FJ445139.1 & A   \\
		FJ445140.1 & A   & FJ445141.1 & A   & FJ445142.1 & A   \\
		FJ445143.1 & A   & FJ445144.1 & A   & FJ445145.1 & A   \\
		FJ445146.1 & A   & FJ445147.1 & A   & FJ445148.1 & A   \\
		FJ445149.1 & A   & FJ445152.1 & A   & FJ445154.1 & A   \\
		FJ445156.1 & A   & FJ445157.1 & A   & FJ445158.1 & A   \\
		FJ445160.1 & A   & FJ445163.1 & A   & FJ445165.1 & A   \\
		FJ445166.1 & A   & FJ445167.1 & A   & FJ445170.1 & A   \\
		FJ445171.1 & A   & FJ445173.1 & A   & FJ445175.1 & A   \\
		FJ445176.1 & A   & FJ445177.1 & A   & FJ445178.1 & A   \\
		FJ445179.1 & A   & FJ445180.1 & A   & FJ445181.1 & A   \\
		FJ445182.1 & A   & FJ445183.1 & A   & FJ445184.1 & A   \\
		FJ445185.1 & A   & FJ445189.1 & A   & FJ445190.1 & A   \\
		L24917.1   & A   & X02316.1   & A   &             &     \\
		DQ473485.1 & B   & DQ473486.1 & B   & DQ473488.1 & B   \\
		DQ473489.1 & B   & DQ473490.1 & B   & EF173420.1 & B   \\
		EF173423.1 & B   & EF173425.1 & B   & FJ445112.1 & B   \\
		FJ445124.1 & B   & FJ445130.1 & B   & FJ445137.1 & B   \\
		FJ445151.1 & B   & FJ445153.1 & B   & FJ445155.1 & B   \\
		FJ445161.1 & B   & FJ445162.1 & B   & FJ445164.1 & B   \\
		FJ445168.1 & B   & FJ445169.1 & B   & FJ445172.1 & B   \\
		FJ445174.1 & B   & FJ445186.1 & B   & FJ445187.1 & B   \\
		FJ445188.1 & B   & L05355.1   & B   &             &     \\
		EF077279.1 & C   & EF077280.1 & C   & EF186077.2 & C   \\
		EF582385.1 & C   & EF582386.1 & C   & EF582387.1 & C   \\
		\bottomrule
	\end{tabular}
\end{table}

\begin{table}[ht]
	\centering
	\caption{Hemagglutinin (HA) gene accessions of Influenza A viruses and their corresponding subtype labels.}
	\label{tab:influenza_accessions}
	\begin{tabular}{llllll}
		\toprule
		Accession & Subtype & Accession & Subtype & Accession & Subtype \\
		\midrule
		KM409247.1 & H1N1 & CY249795.1 & H1N1 & JN809163.1 & H1N1 \\
		MK557230.1 & H1N1 & MG982609.1 & H1N1 & L11126.1   & H2N2 \\
		KP098400.1 & H2N2 & CY117027.1 & H2N2 & KM885174.1 & H2N2 \\
		CY020549.1 & H2N2 & CY105638.1 & H3N2 & KT839427.1 & H3N2 \\
		KC882695.1 & H3N2 & CY163712.1 & H3N2 & CY067977.1 & H3N2 \\
		KY171504.1 & H5N1 & CY126184.1 & H5N1 & FJ492879.1 & H5N1 \\
		KM821622.1 & H5N1 & JN807989.1 & H5N1 & EU030984.1 & H7N3 \\
		CY203926.1 & H7N3 & CY133177.1 & H7N3 & KU289927.1 & H7N3 \\
		GU051805.1 & H7N3 & CY235363.1 & H7N9 & CY192939.1 & H7N9 \\
		CY192179.1 & H7N9 & CY191691.1 & H7N9 & MH266623.1 & H7N9 \\
		\bottomrule
	\end{tabular}
\end{table}

\begin{table}[ht]
	\centering
	\caption{Genome accessions and corresponding species designations of Ebolavirus strains used in the phylogenetic tree reconstruction.}
	\label{tab:ebov_accessions}
	\begin{tabular}{llll}
		\toprule
		Accession & Group & Accession & Group \\
		\midrule
		FJ217161.1 & BDBV & KC545393.1 & BDBV \\
		KC545395.1 & BDBV & KC545394.1 & BDBV \\
		KC545396.1 & BDBV & FJ217162.1 & TAFV \\
		AF522874.1 & RESTV & AB050936.1 & RESTV \\
		JX477166.1 & RESTV & FJ621585.1 & RESTV \\
		FJ621583.1 & RESTV & JX477165.1 & RESTV \\
		FJ968794.1 & SUDV & KC242783.2 & SUDV \\
		EU338380.1 & SUDV & AY729654.1 & SUDV \\
		JN638998.1 & SUDV & KC545389.1 & SUDV \\
		KC545390.1 & SUDV & KC545391.1 & SUDV \\
		KC545392.1 & SUDV & KC589025.1 & SUDV \\
		KC242801.1 & EBOV & NC\_002549.1 & EBOV \\
		KC242791.1 & EBOV & KC242792.1 & EBOV \\
		KC242793.1 & EBOV & KC242794.1 & EBOV \\
		AY354458.1 & EBOV & KC242796.1 & EBOV \\
		KC242799.1 & EBOV & KC242784.1 & EBOV \\
		KC242786.1 & EBOV & KC242787.1 & EBOV \\
		KC242789.1 & EBOV & KC242785.1 & EBOV \\
		KC242790.1 & EBOV & KC242788.1 & EBOV \\
		KC242800.1 & EBOV & KM034555.1 & EBOV \\
		KM034562.1 & EBOV & KM233039.1 & EBOV \\
		KM034557.1 & EBOV & KM034560.1 & EBOV \\
		KM233050.1 & EBOV & KM233053.1 & EBOV \\
		KM233057.1 & EBOV & KM233063.1 & EBOV \\
		KM233072.1 & EBOV & KM233110.1 & EBOV \\
		KM233070.1 & EBOV & KM233099.1 & EBOV \\
		KM233097.1 & EBOV & KM233109.1 & EBOV \\
		KM233096.1 & EBOV & KM233103.1 & EBOV \\
		KJ660346.2 & EBOV & KJ660347.2 & EBOV \\
		KJ660348.2 & EBOV &  &  \\
		\bottomrule
	\end{tabular}
\end{table}

\begin{table}[ht]
	\centering
	\caption{Bacterial genome accessions and corresponding family-level taxonomic assignments.}
	\label{tab:accessions}
	\begin{tabular}{llll}
		\toprule
		Accession & Family & Accession & Family \\
		\midrule
		CP001598.1 & Bacillaceae & AE016879.1 & Bacillaceae \\
		CP001215.1 & Bacillaceae & AE017225.1 & Bacillaceae \\
		CP000976.1 & Borreliaceae & CP000048.1 & Borreliaceae \\
		CP000993.1 & Borreliaceae & CP000049.1 & Borreliaceae \\
		CP000246.1 & Clostridiaceae & CP000312.1 & Clostridiaceae \\
		BA000016.3 & Clostridiaceae & CP000527.1 & Desulfovibrionaceae \\
		AE017285.1 & Desulfovibrionaceae & CP002297.1 & Desulfovibrionaceae \\
		AM260480.1 & Burkholderiaceae & CP000091.1 & Burkholderiaceae \\
		CP000578.1 & Rhodobacteraceae & CP001151.1 & Rhodobacteraceae \\
		AM295250.1 & Staphylococcaceae & AE015929.1 & Staphylococcaceae \\
		AP006716.1 & Staphylococcaceae & CP001837.1 & Staphylococcaceae \\
		AL590842.1 & Yersiniaceae & CP001585.1 & Yersiniaceae \\
		AE009952.1 & Yersiniaceae & CP001593.1 & Yersiniaceae \\
		CP001671.1 & Enterobacteriaceae & CP000468.1 & Enterobacteriaceae \\
		CP001383.1 & Enterobacteriaceae & AE005674.2 & Enterobacteriaceae \\
		\bottomrule
	\end{tabular}
\end{table}

\clearpage

\endgroup
\end{document}